\documentclass[a4paper,UKenglish,cleveref,hyperref,autoref]{lipics-v2021}

\title{Twin-width VIII: delineation and win-wins}
\titlerunning{Twin-width VIII: delineation and win-wins}

\author{\'{E}douard Bonnet}{Univ Lyon, CNRS, ENS de Lyon, Université Claude Bernard Lyon 1, LIP UMR5668, France \and \url{http://perso.ens-lyon.fr/edouard.bonnet/}}{edouard.bonnet@ens-lyon.fr}{https://orcid.org/0000-0002-1653-5822}{}
\author{Dibyayan Chakraborty}{Univ Lyon, CNRS, ENS de Lyon, Université Claude Bernard Lyon 1, LIP UMR5668, France}{dibyayan.chakraborty@ens-lyon.fr}{}{}
\author{Eun Jung Kim}{Universit\'{e} Paris-Dauphine, PSL University, CNRS UMR7243, LAMSADE, Paris, France}{eun-jung.kim@dauphine.fr}{https://orcid.org/0000-0002-6824-0516}{}
\author{Noleen K\"ohler}{Universit\'{e} Paris-Dauphine, PSL University, CNRS UMR7243, LAMSADE, Paris, France}{noleen.kohler@dauphine.psl.eu}{}{}
\author{Raul Lopes}{Universit\'{e} Paris-Dauphine, PSL University, CNRS UMR7243, LAMSADE, Paris, France}{raul-wayne.teixeira-lopes@dauphine.psl.eu}{https://orcid.org/0000-0002-7487-3475}{}
\author{St\'{e}phan Thomass\'{e}}{Univ Lyon, CNRS, ENS de Lyon, Universit\'{e} Claude Bernard Lyon 1, LIP UMR5668, France}{stephan.thomasse@ens-lyon.fr}{}{}

\authorrunning{\'E. Bonnet, D. Chakraborty, E. J. Kim, N. K\"ohler, R. Lopes, S. Thomassé}

\Copyright{Édouard Bonnet, Dibyayan Chakraborty, Eun Jung Kim, Noleen K\"ohler, Raul Lopes, Stéphan Thomassé}



\category{}

\relatedversion{}

\supplement{}

\acknowledgements{}

\nolinenumbers 

\hideLIPIcs  

\EventEditors{John Q. Open and Joan R. Access}
\EventNoEds{2}
\EventLongTitle{42nd Conference on Very Important Topics (CVIT 2016)}
\EventShortTitle{CVIT 2016}
\EventAcronym{CVIT}
\EventYear{2016}
\EventDate{December 24--27, 2016}
\EventLocation{Little Whinging, United Kingdom}
\EventLogo{}
\SeriesVolume{42}
\ArticleNo{23}

\usepackage[utf8]{inputenc}  

\usepackage[T1]{fontenc}
\usepackage{lmodern}

\usepackage[colorinlistoftodos,bordercolor=orange,backgroundcolor=orange!20,linecolor=orange,textsize=normalsize]{todonotes}

\usepackage{amsmath}  
\usepackage{amssymb}     
\usepackage{bbm}
\usepackage{accents}
\usepackage{complexity}

\usepackage{booktabs}
\usepackage{hyperref}
\usepackage{mathrsfs}
\usepackage{paralist}
\usepackage{fixmath}
\usepackage[noadjust]{cite}
\usepackage{ifthen}

\crefformat{equation}{\textup{#2(#1)#3}}
\crefrangeformat{equation}{\textup{#3(#1)#4--#5(#2)#6}}
\crefmultiformat{equation}{\textup{#2(#1)#3}}{ and \textup{#2(#1)#3}}
{, \textup{#2(#1)#3}}{, and \textup{#2(#1)#3}}
\crefrangemultiformat{equation}{\textup{#3(#1)#4--#5(#2)#6}}%
{ and \textup{#3(#1)#4--#5(#2)#6}}{, \textup{#3(#1)#4--#5(#2)#6}}{, and \textup{#3(#1)#4--#5(#2)#6}}

\Crefformat{equation}{#2Equation~\textup{(#1)}#3}
\Crefrangeformat{equation}{Equations~\textup{#3(#1)#4--#5(#2)#6}}
\Crefmultiformat{equation}{Equations~\textup{#2(#1)#3}}{ and \textup{#2(#1)#3}}
{, \textup{#2(#1)#3}}{, and \textup{#2(#1)#3}}
\Crefrangemultiformat{equation}{Equations~\textup{#3(#1)#4--#5(#2)#6}}%
{ and \textup{#3(#1)#4--#5(#2)#6}}{, \textup{#3(#1)#4--#5(#2)#6}}{, and \textup{#3(#1)#4--#5(#2)#6}}

\makeatletter
\newtheorem*{rep@theorem}{\rep@title}
\newcommand{\newreptheorem}[2]{%
\newenvironment{rep#1}[1]{%
 \def\rep@title{#2 \ref{##1}}%
 \begin{rep@theorem}}%
 {\end{rep@theorem}}}
\makeatother

\newcommand{\Ram}[2]{\text{Ram}(#1,#2)}
\newcommand{\Bipram}[2]{\text{Bip-Ram}(#1,#2)}

\newcommand{\col}{\text{cols}}

\newcommand{\mat}[2]{\mathcal M_{#1}(#2)}

\newcommand{\adj}[2]{\text{Adj}_{#1}(#2)}

\newreptheorem{theorem}{Theorem}
\newreptheorem{lemma}{Lemma}

\usepackage{xspace}
\usepackage{tikz}

\usepackage[ruled,vlined,linesnumbered]{algorithm2e}

\usetikzlibrary{fit}
\usetikzlibrary{arrows}
\usetikzlibrary{patterns}
\usetikzlibrary{calc}
\usetikzlibrary{shapes}
\usetikzlibrary{positioning}
\usetikzlibrary{math}
\usetikzlibrary{shapes.symbols}
\usetikzlibrary{decorations.pathreplacing}
\tikzset{draw half paths/.style 2 args={%
  decoration={show path construction,
    lineto code={
      \draw [#1] (\tikzinputsegmentfirst) -- 
         ($(\tikzinputsegmentfirst)!0.5!(\tikzinputsegmentlast)$);
      \draw [#2] ($(\tikzinputsegmentfirst)!0.5!(\tikzinputsegmentlast)$)
        -- (\tikzinputsegmentlast);
    }
  }, decorate
}}

\usepackage[scr=boondox,scrscaled=1.05]{mathalfa}

\renewcommand{\geq}{\geqslant}
\renewcommand{\leq}{\leqslant}

\newcommand{\ie}{i.e.}

\theoremstyle{definition}

\newcommand{\tww}{\text{tww}}
\newcommand{\tw}{\text{tw}}

\renewcommand{\C}{{\mathcal C}}

\newcommand{\white}{{\sf white}}
\newcommand{\gray}{{\sf gray}}
\newcommand{\black}{{\sf black}}
\newcommand{\reserved}{{\sf reserved}}
\newcommand{\parent}{{\sf parent}}

\newcommand{\gr}{\text{gr}}

\newcommand{\high}{\textsf{high}\xspace}
\newcommand{\low}{\textsf{low}\xspace}

\newcommand{\I}{\mathcal I}

\newcommand\abs[1]{\lvert #1\rvert}

\begin{document}

\maketitle

\begin{abstract}
  We introduce the notion of delineation. 
  A graph class $\mathcal C$ is said~\emph{delineated by twin-width} (or simply, \emph{delineated}) if for every hereditary closure $\mathcal D$ of a subclass of~$\mathcal C$, it holds that $\mathcal D$ has bounded twin-width if and only if $\mathcal D$ is monadically dependent.
  An effective strengthening of \emph{delineation} for a class $\mathcal C$ implies that tractable FO model checking on $\mathcal C$ is perfectly understood: On hereditary closures $\mathcal D$ of subclasses of $\mathcal C$, FO model checking is fixed-parameter tractable (FPT) exactly when $\mathcal D$ has bounded twin-width.  
  Ordered graphs [BGOdMSTT, STOC '22] and permutation graphs [BKTW, JACM '22] are effectively delineated, while subcubic graphs are not.
  On the one hand, we prove that interval graphs, and even, rooted directed path graphs are delineated.
  On the other hand, we observe or show that segment graphs, directed path graphs (with arbitrarily many roots), and visibility graphs of simple polygons are not delineated.

  In an effort to draw the delineation frontier between interval graphs (that are delineated) and axis-parallel two-lengthed segment graphs (that are not), we investigate the twin-width of restricted segment intersection classes.
  It was known that (triangle-free) pure axis-parallel unit segment graphs have unbounded twin-width~[BGKTW, SODA '21].
  We show that $K_{t,t}$-free segment graphs, and axis-parallel $H_t$-free unit segment graphs have bounded twin-width, where $H_t$ is the half-graph or ladder of height $t$.
  In contrast, axis-parallel $H_4$-free two-lengthed segment graphs have unbounded twin-width.
  We leave as an open question whether unit segment graphs are delineated.

  More broadly, we explore which structures (large bicliques, half-graphs, or independent sets) are responsible for making the twin-width large on the main classes of intersection and visibility graphs.  
  Our new results, combined with the FPT algorithm for FO model checking on graphs given with $O(1)$-sequences [BKTW, JACM '22], give rise to a variety of algorithmic win-win arguments.
  They all fall in the same framework: If $p$ is an FO definable graph parameter that effectively functionally upperbounds twin-width on a class $\mathcal C$, then $p(G) \geqslant k$ can be decided in FPT time $f(k) \cdot |V(G)|^{O(1)}$.
  For instance, we readily derive FPT algorithms for \textsc{$k$-Ladder} on visibility graphs of 1.5D terrains, and \textsc{$k$-Independent Set} on visibility graphs of simple polygons.
  This showcases that the theory of twin-width can serve outside classes of bounded twin-width.
\end{abstract}
\maketitle

\section{Introduction}\label{sec:intro}

A~\emph{trigraph} $G$ has a vertex set $V(G)$, and two disjoint edge sets, the black edge set $E(G)$ and the red edge set $R(G)$.  
A~(vertex) \emph{contraction} consists of merging two (non-necessarily adjacent) vertices, say, $u, v$ into a~vertex~$w$, and keeping every existing edge $wz$ black if and only if $uz$ and $vz$ were previously black edges.
The other edges incident to $w$ turn red (if not already), while the rest of the trigraph remains the same.
A~\emph{contraction sequence} of an $n$-vertex (tri)graph $G$ is a sequence of trigraphs $G=G_n, \ldots, G_1=K_1$ such that $G_i$ is obtained from $G_{i+1}$ by performing one contraction.
A~\mbox{\emph{$d$-sequence}} is a contraction sequence wherein all trigraphs have red degree at most~$d$.
The~\emph{twin-width} of $G$, denoted by $\tww(G)$, is the minimum integer~$d$ such that $G$ admits a $d$-sequence.
See~\cref{fig:contraction-sequence} for an illustration of a 2-sequence of a graph.

\begin{figure}[h!]
  \centering
  \resizebox{400pt}{!}{
  \begin{tikzpicture}[
      vertex/.style={circle, draw, minimum size=0.68cm}
    ]
    \def\s{1.2}
    \foreach \i/\j/\l in {0/0/a,0/1/b,0/2/c,1/0/d,1/1/e,1/2/f,2/1/g}{
      \node[vertex] (\l) at (\i * \s,\j * \s) {$\l$} ;
    }
    \foreach \i/\j in {a/b,a/d,a/f,b/c,b/d,b/e,b/f,c/e,c/f,d/e,d/g,e/g,f/g}{
      \draw (\i) -- (\j) ;
    }

    \begin{scope}[xshift=3 * \s cm]
    \foreach \i/\j/\l in {0/0/a,0/1/b,0/2/c,1/0/d,2/1/g}{
      \node[vertex] (\l) at (\i * \s,\j * \s) {$\l$} ;
    }
    \foreach \i/\j/\l in {1/1/e,1/2/f}{
      \node[vertex,opacity=0.2] (\l) at (\i * \s,\j * \s) {$\l$} ;
    }
    \node[draw,rounded corners,inner sep=0.01cm,fit=(e) (f)] (ef) {ef} ;
    \foreach \i/\j in {a/b,a/d,b/c,b/d,b/ef,c/ef,c/ef,d/g,ef/g,ef/g}{
      \draw (\i) -- (\j) ;
    }
    \foreach \i/\j in {a/ef,d/ef}{
      \draw[red, very thick] (\i) -- (\j) ;
    }
    \end{scope}

    \begin{scope}[xshift=6 * \s cm]
    \foreach \i/\j/\l in {0/1/b,0/2/c,2/1/g,1/1/ef}{
      \node[vertex] (\l) at (\i * \s,\j * \s) {$\l$} ;
    }
    \foreach \i/\j/\l in {0/0/a,1/0/d}{
      \node[vertex,opacity=0.2] (\l) at (\i * \s,\j * \s) {$\l$} ;
    }
    \draw[opacity=0.2] (a) -- (d) ;
    \node[draw,rounded corners,inner sep=0.01cm,fit=(a) (d)] (ad) {ad} ;
    \foreach \i/\j in {ad/b,b/c,b/ad,b/ef,c/ef,c/ef,ef/g,ef/g}{
      \draw (\i) -- (\j) ;
    }
    \foreach \i/\j in {ad/ef,ad/g}{
      \draw[red, very thick] (\i) -- (\j) ;
    }
    \end{scope}

    \begin{scope}[xshift=9 * \s cm]
    \foreach \i/\j/\l in {0/2/c,2/1/g,0.5/0/ad}{
      \node[vertex] (\l) at (\i * \s,\j * \s) {$\l$} ;
    }
    \foreach \i/\j/\l in {0/1/b,1/1/ef}{
      \node[vertex,opacity=0.2] (\l) at (\i * \s,\j * \s) {$\l$} ;
    }
    \draw[opacity=0.2] (b) -- (ef) ;
    \node[draw,rounded corners,inner sep=0.01cm,fit=(b) (ef)] (bef) {bef} ;
    \foreach \i/\j in {ad/bef,bef/c,bef/ad,c/bef,c/bef,bef/g}{
      \draw (\i) -- (\j) ;
    }
    \foreach \i/\j in {ad/bef,ad/g,bef/g}{
      \draw[red, very thick] (\i) -- (\j) ;
    }
    \end{scope}

    \begin{scope}[xshift=11.7 * \s cm]
    \foreach \i/\j/\l in {0/2/c}{
      \node[vertex] (\l) at (\i * \s,\j * \s) {$\l$} ;
    }
     \foreach \i/\j/\l in {0.5/0/adg,0.5/1.1/bef}{
      \node[vertex] (\l) at (\i * \s,\j * \s) {\footnotesize{\l}} ;
    }
    \foreach \i/\j in {c/bef}{
      \draw (\i) -- (\j) ;
    }
    \foreach \i/\j in {adg/bef}{
      \draw[red, very thick] (\i) -- (\j) ;
    }
    \end{scope}

    \begin{scope}[xshift=13.7 * \s cm]
    \foreach \i/\j/\l in {0.5/0/adg,0.5/1.1/bcef}{
      \node[vertex] (\l) at (\i * \s,\j * \s) {\footnotesize{\l}} ;
    }
    \foreach \i/\j in {adg/bcef}{
      \draw[red, very thick] (\i) -- (\j) ;
    }
    \end{scope}

    \begin{scope}[xshift=15 * \s cm]
    \foreach \i/\j/\l in {1/0.75/abcdefg}{
      \node[vertex] (\l) at (\i * \s,\j * \s) {\tiny{\l}} ;
    }
    \end{scope}
    
  \end{tikzpicture}
  }
  \caption{A 2-sequence witnesses that the initial 7-vertex graph has twin-width at most~2.}
  \label{fig:contraction-sequence}
\end{figure}
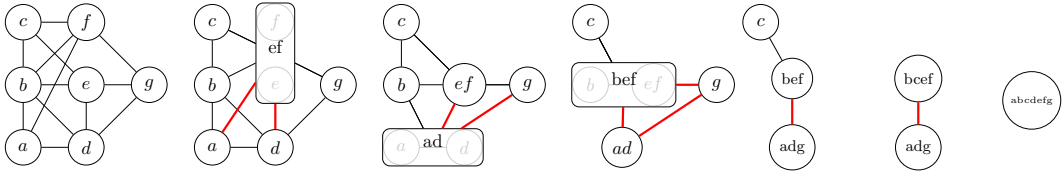

A graph class $\mathcal C$ has then \emph{bounded twin-width} if there is a constant $t$ such that every graph $G \in \mathcal C$ satisfies $\tww(G) \leqslant t$. 
In the present paper, the following characterization will be more important than the original definition. 
\begin{theorem}[\cite{twin-width4}]\label{thm:mixed-number-gen}
  A class $\mathcal C$ has bounded twin-width if and only if there is an integer $k$ such that every graph of $\mathcal C$ admits an adjacency matrix \emph{without} rank-$k$ division, i.e., $k$-division such that every cell has combinatorial rank at least~$k$.
\end{theorem}
In the previous theorem, a $k$-division of a matrix is a partition of its column (resp.~row) set into $k$ intervals, called \emph{column} (resp.~\emph{row}) \emph{parts}, of consecutive columns (resp.~rows).
A~$k$-division naturally defines $k^2$ \emph{cells} (contiguous submatrices) made by the entries at the intersection of a column part with a row part. 
\cref{thm:mixed-number-gen} is effective: There is a computable function $f$, such that, given a vertex ordering along which the adjacency matrix of a graph $G$ has no rank-$k$ division, one can efficiently find an $f(k)$-sequence for $G$, witnessing that $\tww(G) \leqslant f(k)$. 

Classes of graphs (or more generally of binary structures, since the definition of twin-width extends to them) with bounded twin-width include classes with bounded clique-width, $H$-minor free graphs for any fixed $H$, posets with antichains of bounded size, strict subclasses of permutation graphs, map graphs, bounded-degree string graphs~\cite{twin-width1}, as well as $\Omega(\log n)$-subdivisions of $n$-vertex graphs, and some particular classes of cubic expanders~\cite{twin-width2}.
In contrast, (sub)cubic graphs, interval graphs, triangle-free unit segment graphs, unit disk graphs have unbounded twin-width~\cite{twin-width2}.

The main algorithmic application of twin-width is that first-order (FO) model checking, that is, deciding if a first-order sentence $\varphi$ holds in a graph $G$, can be decided in fixed-parameter time (FPT) $f(|\varphi|,d) \cdot |V(G)|$ for some computable function $f$, when given a $d$-sequence of $G$~\cite{twin-width1}. 
As for almost all classes known to have bounded twin-width, one can compute $O(1)$-sequences in polynomial time for members of the class, the latter result unifies and extends several known results~\cite{Flum01,Gajarsky14,Ganian15,Gajarsky15,Guillemot14} for hereditary (i.e., closed under removing vertices), but not necessarily monotone, classes.
Faster FPT algorithms with (almost) single-exponential dependence in the parameter are known for specific FO definable problems (like \textsc{$k$-Independent Set} and \textsc{$k$-Dominating Set}) on graphs of bounded twin-width given with an $O(1)$-sequence~\cite{twin-width3}.  

For monotone (i.e., closed under removing vertices and edges) classes, the FPT algorithm of Grohe, Kreutzer, and Siebertz~\cite{Grohe17} for FO model checking on nowhere dense classes, is complemented by W$[1]$-hardness on classes that are somewhere dense (i.e., \emph{not} nowhere dense)~\cite{Dvorak13}, and even AW$[*]$-hardness on classes that are \emph{effectively} somewhere dense~\cite{Kreutzer09}.
The latter two results imply that, for monotone classes, FO model checking is unlikely to be FPT beyond nowhere dense classes.
Thus the classification of monotone classes admitting an FPT FO model checking is over.
However such a classification remains an active line of work for the more general hereditary classes of graphs and binary structures~\cite{GajarskyHOLR20,GajarskyKNMPST20,Eickmeyer20}.
It is conjectured (see for instance \cite[Conjecture 8.2]{GajarskyHOLR20}) that:
\begin{conjecture}\label{conj:fomc}
  For every hereditary class $\mathcal C$ of structures, FO model checking is FPT on $\mathcal C$ if and only if $\mathcal C$ is monadically dependent.\footnote{A model-theoretic notion which roughly says that not every graph $G$ can be built from a nondeterministic $O(1)$-coloring of some $\mathcal S \in \mathcal C$ by means of a first-order formula $\varphi(x,y)$, in the relations of $\mathcal S$ and the added colors, imposing the edge set of $G$; see~\cref{sec:prelim} for a definition.}
\end{conjecture}

The missing element for an FPT FO model-checking algorithm on any class of bounded twin-width is a polynomial-time algorithm and a computable function $f$, that given a constant integer bound $c$ and a graph $G$, either finds an $f(c)$-sequence for $G$, or correctly reports that $\tww(G) > c$.
The runtime of the algorithm could be $n^{g(c)}$, for some function $g$.
However to get an FPT algorithm in the combined parameter \emph{size of the sentence + bound on the twin-width}, one would further require that the approximation algorithm takes FPT time in $c$ (now thought of as a parameter), i.e., $g(c) n^{O(1)}$.
Such an algorithm exists on ordered graphs (more generally, ordered binary structures)~\cite{twin-width4}, graphs of bounded clique-width, proper minor-closed classes~\cite{twin-width1}, but not on general graphs.
Let us observe that exactly computing the twin-width, and even distinguishing between $\tww(G) = 4$ and $\tww(G) = 5$, is NP-complete~\cite{Berge21}.

A parallel can be drawn with treewidth and clique-width.
Both are NP-hard to compute~\cite{Arnborg87,Fellows09}, but --unlike twin-width-- treewidth admits an FPT algorithm~\cite{Bodlaender96}, while clique-width can be approximated in FPT time~\cite{Oum08}.
Celebrated theorems by Courcelle~\cite{Courcelle90} and by Courcelle, Makowsky, and Rotics~\cite{Courcelle00}, together with the latter two results, imply that monadic second-order (MSO) model checking is FPT on classes (of incidence graphs) with bounded treewidth, and on classes with bounded clique-width.

In this paper, we get around the two main caveats of using twin-width for algorithm design.
Namely:
\begin{compactitem}
\item an FPT (or XP) approximation of twin-width is still missing, and
\item a priori \emph{only} classes of bounded twin-width are concerned.  
\end{compactitem}
We consider geometric graph classes with \emph{unbounded} twin-width: interval graphs, (rooted) directed path graphs, segment graphs, visibility graphs of polygons and terrains.
For all these classes, a vertex-ordering procedure either comes naturally\footnote{Think, ordering by increasing left endpoints in interval graphs, or following the boundary of polygons or terrains.} or can be worked out and efficiently computed.
Two cases are to be considered.
Either the obtained adjacency matrix has no rank-$k$ division --and we get a favorable contraction sequence by~\cref{thm:mixed-number-gen}-- or it does have such a division. 
In the latter case, a large structured object of variable complexity may be found, such as a biclique, a half-graph (or ladder), or even an obstacle to an FPT FO model checking in the form of a \emph{transversal pair of half-graphs} (or \emph{transversal pair}, for short) or some variant of it; see the 3 leftmost graphs in~\cref{fig:bip}, \cref{sec:prelim} for a formal definition, and for why transversal pairs indeed are such obstacles.

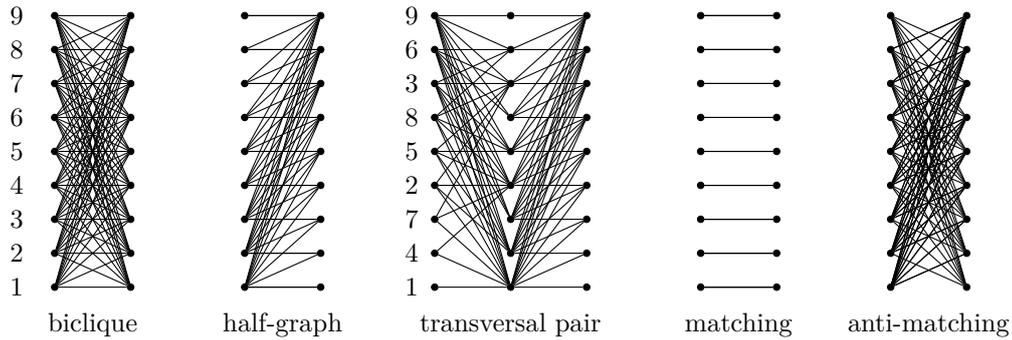
\begin{figure}
  \centering
  \begin{tikzpicture}[vertex/.style={draw,fill,circle,inner sep=0.03cm}]
    \def\n{9}
    \def\s{1}
    \def\z{0.45}
    \foreach \i in {1,...,\n}{
       \node at (-0.5,\i * \z - \n * \z - \z) {\i} ;
    }

    \foreach \symb/\symbp/\symbpp/\xsh/\sign in {{<}/{>}/{=}/0/biclique,{>}/{=}/{=}/2.5/{half-graph},{=}/{=}/{=}/8.5/matching,{<}/{>}/{>}/11/{anti-matching}}{
      \begin{scope}[xshift=\xsh cm]
        \node at (\s / 2,-\n * \z-0.5) {\sign} ;
    \foreach \i in {1,...,\n}{
      \node[vertex] (l\i) at (0,-\i * \z) {} ;
      \node[vertex] (r\i) at (\s,-\i * \z) {} ;
    }
    \foreach \i in {1,...,\n}{
      \foreach \j in {1,...,\n}{
        \ifthenelse{\i \symb \j}{\draw (l\i) -- (r\j);}{}
        \ifthenelse{\i \symbp \j}{\draw (l\i) -- (r\j);}{}
        \ifthenelse{\i \symbpp \j}{\draw (l\i) -- (r\j);}{}
      }
    }
      \end{scope}
    }
    \begin{scope}[xshift=6 cm]
        \node at (0,-\n * \z-0.5) {transversal pair} ;
        \foreach \i in {1,...,\n}{
      \node[vertex] (c\i) at (-\s,-\i * \z) {} ;
      \node[vertex] (a\i) at (0,-\i * \z) {} ;
      \node[vertex] (b\i) at (\s,-\i * \z) {} ;
        }
        \foreach \i/\j in {9/1,8/4,7/7,6/2,5/5,4/8,3/3,2/6,1/9}{
          \node at (-\s-0.3,-\i * \z) {$\j$} ;
        }
    \foreach \i in {1,...,\n}{
      \foreach \j in {1,...,\n}{
        \ifthenelse{\i = \j}{\draw (a\i) -- (b\j);}{}
        \ifthenelse{\i > \j}{\draw (a\i) -- (b\j);}{}
      }
    }
    \foreach \i/\j in {9/9,8/9,8/6,7/9,7/6,7/3, 6/9,6/6,6/3,6/8, 5/9,5/6,5/3,5/8,5/5, 4/9,4/6,4/3,4/8,4/5,4/2, 3/9,3/6,3/3,3/8,3/5,3/2,3/7, 2/9,2/6,2/3,2/8,2/5,2/2,2/7,2/4, 1/9,1/6,1/3,1/8,1/5,1/2,1/7,1/4,1/1}{
      \draw (c\i) -- (a\j) ;
    }
    \end{scope}
  \end{tikzpicture}
  \caption{Biclique, half-graph (or ladder), transversal pair of half-graphs, matching, anti-matching, all of height 9. Bicliques and half-graphs are semi-induced by default. The number next to each leftmost vertex $v$ of the transversal pair indicates the height of the neighbor of $v$ in the central column which is not also a neighbor of the vertex just below $v$.}
  \label{fig:bip}
\end{figure}

\medskip

\textbf{Delineation.}
When, on a class $\mathcal C$, a~(variant of a) transversal pair can systematically be found as a result of unbounded twin-width, it shows that the classification of FPT FO model checking for hereditary subclasses of $\mathcal C$ is entirely settled by the algorithm on graphs of bounded twin-width~\cite{twin-width1}. 

More generally if for every hereditary closure $\mathcal D$ of a subclass of $\mathcal C$,\footnote{The reason we do not simply quantify over hereditary subclasses of $\mathcal C$ is to have a notion that is also meaningful when $\mathcal C$ is not hereditary.} $\mathcal D$ has bounded twin-width if and only if $\mathcal D$ is monadically dependent, we say that $\mathcal C$ is \emph{delineated by twin-width} (or simply, \emph{delineated}).
As every class of bounded twin-width is monadically dependent~\cite{twin-width1}, equivalently, we can only require that $\mathcal D$ has bounded twin-width if $\mathcal D$ is monadically dependent.

Although not stated in those terms, permutation graphs were already proven to be delineated~\cite{twin-width1}, as well as ordered graphs~\cite{twin-width4}.
We add interval graphs and rooted directed path graphs (see~\cref{subsec:rdpg} for a definition) to the list of delineated classes.
For every hereditary subclass of these classes the classification of FPT FO model checking, \cref{conj:fomc}, is now provably settled.\footnote{We actually need an effective strengthening of delineation that also holds for these classes and will be defined in~\cref{sec:prelim}.}

\begin{theorem}\label{thm:delineation-informal}
  Interval graphs, and more generally rooted directed path graphs, are delineated.
\end{theorem}

In contrast, we rule out delineation for directed path graphs (with multiple roots), intersection graphs of pure axis-parallel segments with two distinct lengths, and visibility graphs of simple polygons.
If segment graphs and visibility graphs of simple polygons do not yield in their subfamilies of unbounded twin-width complex enough structures to settle~\cref{conj:fomc}, unbounded twin-width still imply in those classes that some other graph parameters are unbounded.
This gives rise to a win-win approach to compute these parameters.
To give some context, we first draw a parallel with what happens with treewidth.  

\medskip

\textbf{Treewidth win-wins.}
The algorithmic use of a parameter like treewidth extends beyond classes wherein treewidth is bounded.
Any problem admitting an FPT algorithm parameterized by treewidth (like MSO definable problems~\cite{Courcelle90}), and a trivial answer (such as a systematic YES or a systematic NO) when the treewidth is large, subjects itself to a straightforward win-win argument.
Either the treewidth is below a specific threshold and one runs the supposed FPT algorithm, or the treewidth is above that threshold, and one reports the corresponding trivial answer.
In the case of treewidth, the latter answer is most often ensured by the existence of a large grid minor (implied for instance by~\cite{ROBERTSON198692,RobertsonST94,ChuzhoyT21}).

In proper minor-closed classes, and even more so in planar graphs, the dependence between treewidth and the size of a largest grid minor is sharply understood. 
A planar graph $G$ has a maximum square grid minor of dimension $\Theta(\tw(G)) \times \Theta(\tw(G))$~\cite{RobertsonST94}.
This is at the basis of the so-called \emph{bidimensionality} theory~\cite{bidim}.
Since a problem like~\textsc{$k$-Vertex Cover}\footnote{where one seeks a subset of $k$ vertices incident to every edge} admits a $2^{\tw(G)} n^{O(1)}$-time algorithm~\cite{complexity} and a systematic NO answer in presence of a, say, $(2 \sqrt k+1) \times (2 \sqrt k+1)$ grid minor, one then derives for this problem an FPT algorithm running in time $2^{O(\sqrt k)} n^{O(1)}$ in planar graphs.

Let us forget one moment the intermediary role of the grid minor.
Efficiently computing a parameter $p(G)$ --like the \emph{vertex cover number $\tau(G)$}-- can boil down to establishing an upperbound of the form $\tw(G) \leqslant f(p(G))$ for some function $f$ --such as, in the previous example and the special case of planar graphs, $\tw(G) \leqslant c \cdot \sqrt{\tau(G)}$ for some constant $c$.

\medskip

\textbf{Twin-width win-wins.}
The second main motivation of the present paper (besides delineation) is to explore such upperbounds, and resultant win-wins, with twin-width in place of treewidth.
This endeavor is incomparable with its treewidth counterpart.
On the one hand, bounded twin-width being strictly more general than bounded treewidth (and than bounded clique-width), twin-width can more often be upperbounded by another parameter.
On the other hand, due to its greater generality, fewer problems are FPT in classes of bounded twin-width than in classes of bounded treewidth.
Notably, any problem expressible in FO falls in the former category, while any problem definable in the richer MSO falls in the latter. 

Given two graph parameters $p, q$, and a graph class $\mathcal C$, we will write \emph{$p \sqsubseteq q$ on $\mathcal C$} to signify that there is a computable function $f$ such that $\forall G \in \mathcal C$, $p(G) \leqslant f(q(G))$.
By a similar argument to what was presented in the previous paragraphs, one gets the following.

\begin{theorem}[informal]\label{thm:inf-win-win}
  Let $\mathcal C$ be a graph class and $p$ be a graph invariant such that
  \begin{compactenum}
  \item \label{it1} computing $p$ is FPT in the combined parameter $p+\tww$ on $\mathcal C$, and
  \item \label{it2} $\tww \sqsubseteq p$ on $\mathcal C$.
  \end{compactenum}
  Then, computing $p$ is FPT on $\mathcal C$.  
\end{theorem}

First-order logic yields a natural pool of invariants $p$ that are fixed-parameter tractable with respect to $p+\tww$ (even on general graphs)~\cite{twin-width1}.
As a first example of~\cref{it2}, we show the following.

\begin{theorem}\label{thm:biclique-free-seg}
  Biclique-free segment graphs have bounded twin-width.
  Furthermore, if a geometric representation is given, an $O(1)$-sequence of the graph is found in polynomial time.
\end{theorem}

This principally uses the degeneracy of these graphs~\cite{Lee17}, their interpretability in planar graphs augmented at their faces by well-behaved matchings, and the fact that first-order interpretations and transductions\footnote{transformations based on FO formulas, and related to dependence and monadic dependence; see~\cref{sec:prelim} for definitions} preserve bounded twin-width~\cite{twin-width1}.

A reformulation is that, in segment graphs, twin-width is upperbounded by a function of the largest biclique; or, denoting by $\beta(G)$ the largest integer $t$ such that $G$ admits a biclique  $K_{t,t}$, it holds that $\tww \sqsubseteq \beta$ on segment graphs.
The corresponding problem \textsc{$k$-Biclique} was famously shown W$[1]$-hard by Lin~\cite{Lin18}, after its parameterized complexity has been open for over a decade~\cite{df13}.
From~\cref{thm:win-win,thm:biclique-free-seg} one obtains that \textsc{$k$-Biclique} is FPT on segment graphs given with a geometric representation.
However this fact can be alternatively derived since $K_{t,t}$-free segment graphs have truly sublinear balanced separators (see for instance~\cite{Lee17}), hence polynomial expansion~\cite{Dvorak16}, and one thus concludes with an FPT algorithm for FO model checking on sparse classes (\cite{Dvorak13} or \cite{Grohe17}).

\medskip

The counterpart of the large grid minor (in treewidth win-wins) is a large rank division in \emph{every} adjacency matrix of the graph (recall~\cref{thm:mixed-number-gen}).
Large twin-width in a class $\mathcal C$ in particular implies a large rank division in the adjacency matrix along a vertex ordering that, at least partially, respects the structure of $\mathcal C$.
In turn, this complex structure --despite being along a canonical order-- may help lowerbounding other parameters (like the grid minor was lowerbounding the vertex cover number in our example).
We give two such examples, both on classes of visibility graphs.

A \emph{simple polygon} is a polygon without holes.
Two vertices (more generally, points) of a polygon \emph{see} each other if the line segment defined by these vertices (or points) is entirely contained in the polygon.
The problem of~\cref{thm:kis-polygon} is sometimes advertised as \emph{hiding (people) in polygons}, and its solution is called a~\emph{hidden set}.
It is NP-complete~\cite{Shermer89}, even APX-hard~\cite{Eidenbenz02}, and can be equivalently defined as \textsc{$k$-Independent Set} in visibility graphs of simple polygons given with a representation.
We turn a large rank division in the adjacency matrix along a Hamiltonian path describing the boundary of the polygon into a large independent set.
This uses Ramsey-like arguments, and the forbidden pattern that off the diagonal of such adjacency matrices no 0 entry is surrounded by four 1 entries, one in every direction.  
In conclusion: $\tww \sqsubseteq \alpha$ in visibility graphs of simple polygons (where $\alpha(G)$ is the independence number of~$G$), and we immediately obtain the following rather surprising result.

\begin{theorem}\label{thm:kis-polygon}
  Given a simple polygon $\mathcal P$ and an integer $k$, finding $k$ vertices of $\mathcal P$ pairwise not seeing each other is FPT. 
\end{theorem}
In contrast, \textsc{$k$-Dominating Set} remains W$[1]$-hard on visibility graphs of simple polygons~\cite{BonnetM20}, thus likely does \emph{not} admit an FPT algorithm.

A~\emph{1.5D terrain} (or here, \emph{terrain} for short) is an $x$-monotone polygonal chain in the plane.
Two vertices of a terrain \emph{see} each other if the line segment they define entirely lies above the terrain. 
Let $\lambda(G)$, the \emph{ladder index} of $G$, be the greatest height of a semi-induced half-graph in $G$.
A folklore structural property of terrains, often called \emph{Order Claim}, imposes the existence of large half-graphs in a large rank division along the left-right ordering.
Thus $\tww \sqsubseteq \lambda$ in visibility graphs of terrains.
We conclude:

\begin{theorem}\label{thm:ladder-terrain}
  \textsc{$k$-Ladder} and \textsc{$k$-Biclique} are FPT on visibility graphs of 1.5D terrains given with a geometric representation.
\end{theorem}

We recap, in a compact form, the structural results discussed in the previous paragraphs, which readily imply FPT algorithms by~\cref{thm:inf-win-win}.
\begin{compactenum}
\item \label{it:beta} $\tww \sqsubseteq \beta$ in segment graphs;
\item \label{it:alpha} $\tww \sqsubseteq \alpha$ in visibility graphs of simple polygons;
\item \label{it:lambda} $\tww \sqsubseteq \lambda$ in visibility graphs of 1.5D terrains.
\end{compactenum}

Hliněný et al.~\cite{hlinveny2019fo} conjecture that FO model checking is FPT on weak visibility graphs of simple polygons additionally parameterized by the independence number.
The proof of~\cref{it:alpha} confirms this conjecture, even for the more general (non-weak) visibility graphs.
We observe that the approach would not work with a classic width measure, since none of the three items hold replacing \emph{twin-width} by \emph{clique-width} (mainly because grids and long paths of consistently ordered half-graphs have bounded twin-width but unbounded clique-width). 


A~slightly alternative perspective to describe our results would read as follows:
We start drawing a more precise frontier of \emph{bounded twin-width} and \emph{delineation} in geometric graph classes.  

\section{Preliminaries and first results}\label{sec:prelim}

We may denote the set of integers between $i$ and $j$ by $[i,j]$, and $[k]$ may be used as a~short-hand for $[1,k]$.
We start this section with the relevant definitions and background in graph theory, finite model theory, and recall some crucial theorems related to twin-width.

\subsection{Graph theory}

We use the standard graph-theoretic definitions and notations.
We denote by $V(G)$, and $E(G)$, the vertex set, and the edge set, of a graph $G$, and by $G[S]$ the subgraph of $G$ induced by $S \subseteq V(G)$.
When $A,B \subseteq V(G)$ are two disjoint sets, we denote by $G[A,B]$ the bipartite graph $(A,B,\{ab~:~a \in A,~b \in B,~ab \in E(G)\})$.
The~\emph{bipartite complement} of a bipartite graph $G=(A,B,E(G))$ is the bipartite graph $(A,B,\{ab~:~a \in A,~b \in B,~ab \notin E(G)\})$.

We denote by $\adj{\prec}{G}$ the adjacency matrix of $G$ \emph{along the total order $\prec$ of $V(G)$}.
If $A$ and $B$ are disjoint sets of $V(G)$, we also denote by $\adj{\prec}{G \langle A,B \rangle}$ the \emph{biadjacency matrix} of $(A,B,E(G))$ \emph{along the total order $\prec$ of $V(G)$}, with $|A|$ columns $a_1 \prec a_2 \prec \ldots \prec a_{|A|} \in A$ and $|B|$ rows $b_1 \prec b_2 \prec \ldots \prec b_{|B|} \in B$, a~1 entry at position $(b_j,a_i)$ if $a_ib_j \in E(G)$, and a~0 entry otherwise.
Note that the biadjacency matrix $\adj{\prec}{G \langle A,B \rangle}$ is different from the adjacency matrix $\adj{\prec}{G[A,B]}$; the former is a block of the other.
An~\emph{interval along $\prec$} or \emph{$\prec$-interval} is a set of \emph{consecutive} elements with respect to that order.
The \emph{mirror} of a total order is obtained by reversing all inequalities. 

A \emph{clique} is a set of pairwise adjacent vertices, while an \emph{independent set} is a set of pairwise non-adjacent vertices. 
A~\emph{biclique} is given by a disjoint pair of vertex subsets $(A,B)$ such that every vertex of $A$ is adjacent to every vertex of $B$.
The \emph{biclique of height $t$} is a biclique where both sides have size $t$, and is denoted by $K_{t,t}$.
A~\emph{half-graph} (or~\emph{ladder}) \emph{of height $t$} consists of $2t$ distinct vertices $a_1, \ldots, a_t, b_1, \ldots, b_t$ such that $a_i$ is adjacent to $b_j$ if and only if $i \leqslant j$.
We denote by $H_t$ the half-graph of height~$t$ when $\{a_1, \ldots, a_t\}$ and $\{b_1, \ldots, b_t\}$ are independent sets.
An~\emph{induced matching} is a set of vertices inducing a disjoint union of edges.
A~\emph{semi-induced matching} consists of $2k$ distinct vertices $a_1, \ldots, a_k, b_1, \ldots, b_k$ such that $a_i$ is adjacent to $b_j$ if and only if $i = j$.
An~\emph{induced anti-matching} and \emph{semi-induced anti-matching} are their bipartite complements.
See~\cref{fig:bip} for illustrations.

More generally, a bipartite graph $H$ is \emph{semi-induced} in $G$ if there are two disjoint $A,B \subseteq V(G)$ such that $H$ is isomorphic to $G[A,B]$.
Bicliques and half-graphs will always be meant semi-induced. 
A graph is \emph{$K_{t,t}$-free} (resp.~$H_t$-free) if it does not contain $K_{t,t}$ (resp.~$H_t$) as a semi-induced subgraph.
More generally, \emph{$G$ is $H$-free} if $G$ does not contain $H$ as an induced subgraph, or as a semi-induced subgraph when $H$ is presented as a bipartite or multipartite graph.

A~\emph{generalized transversal pair of half-graphs} consists of $3+\ell$ sets $A = \{a_{i,j}~:~i, j \in [t]\}$, $B_0 = \{b_{i,j}^0~:~i, j \in [t]\}, \ldots, B_\ell = \{b_{i,j}^{\ell}~:~i, j \in [t]\}$, and $C = \{c_{i,j}~:~i, j \in [t]\}$  such that there is an edge between $a_{i,j}$ and $b_{i',j'}^0$ if and only if $(i,j) \leqslant_{\text{lex}} (i',j')$, for $k\in [\ell]$ there is an edge between $b_{i,j}^{k-1}$ and $b_{i',j'}^{k}$ if and only if $(i,j)=(i',j')$ and there is an edge between $b_{i,j}^{\ell}$ and $c_{i',j'}$ if and only if $(j,i) \leqslant_{\text{lex}} (j',i')$, where $\leqslant_{\text{lex}}$ denotes the lexicographic (left-right) order.
We denote this graph by $T_{t,\ell}$, and a \emph{semi-induced $T_{t,\ell}$} is such a graph with possibly some extra edges between two sets $X,Y\in \{A,B_0,\ldots,B_\ell,C\}$ with no predefined edges.
Note that $A \cup B_0$ and $B_{\ell} \cup C$ both induce a half-graph, but the ``order'' these two half-graphs put on the endpoints of the paths $(b_{i,j}^0,\dots,b_{i,j}^{\ell})$ is different. We define $T_k:=T_{k,0}$ and we call $T_k$ a transversal pair (of half-graphs); see middle of~\cref{fig:bip}.

We will come back to (generalized) transversal pairs after we introduce interpretations and transductions, and show that they are complex, universal objects as far as closure by FO transductions is concerned.

\subsection{Ramsey Theory}

We will need Ramsey's theorem and its bipartite version, in their simplest form, and the Erd\H{o}s-Szekeres theorem.

\begin{theorem}[Ramsey's theorem \cite{Ramsey1930}]
 \label{thm:Ramsey}
 There is a function $\text{Ram}: \mathbb N \times \mathbb N \to \mathbb N$ such that for every $s, t \geq 1$, every 2-edge colored complete graph $K_{\Ram{s}{t}}$ contains a subset of $s$ vertices on which all edges are of the first color, or a subset of $t$ vertices on which all edges are of the second color.  
\end{theorem}
Observe that the previous theorem can be equivalently stated as every $\Ram{s}{t}$-vertex graph has a clique of size $s$ or an independent set of size $t$.

\begin{theorem}[Bipartite Ramsey's theorem]
 \label{thm:bipRamsey}
  There is a function $\text{Bip-Ram}: \mathbb N \times \mathbb N \to \mathbb N$ such that for every $s, t \geq 1$, every 2-edge colored complete bipartite graph $K_{\Bipram{s}{t},\Bipram{s}{t}}$ contains a biclique $K_{s,s}$ whose edges are all of the first color, or a biclique $K_{t,t}$ whose edges are all of the second color. 
\end{theorem}

\begin{theorem}[Erd\H{o}s-Szekeres theorem~\cite{Erdos1987}]
\label{thm:erdos-szekeres}
For every $k \geq 1$ any sequence of distinct $(k-1)^2 + 1$ real numbers contains an non-increasing or non-decreasing subsequence of size at least $k$.
\end{theorem}

\subsection{Model checking, interpretations, transductions, and dependence}\label{subsec:prelim-fmt}

A relational \emph{signature} $\sigma$ is a finite set of relation symbols~$R$, each having a specified arity $r \in \mathbb N$. 
\mbox{A \emph{$\sigma$-structure}~$\mathbf A$} is defined by a set~$A$ (the \emph{domain} of $\mathbf A$) and a relation $R^{\mathbf A} \subseteq A^{r}$ for each relation symbol \mbox{$R \in \sigma$} with arity $r$.

A~\emph{binary structure} is a relational structure with symbols of arity at most 2.
The syntax and semantics of first-order formulas over $\sigma$ (or \emph{$\sigma$-formulas} for short), are defined as usual.
We recall that a \emph{sentence} is a formula without free variable.
Most of the time we will consider $\sigma$-structures with $\sigma$ consisting of a single binary relation symbol $E$, and identify them to graphs.
But we will also deal with binary structures that are graphs augmented with a total order (called \emph{totally ordered graphs}, or \emph{ordered graphs} for short) and/or some unary relations.

\medskip

\textbf{Model checking.}
The \emph{first-order (FO) model checking} inputs a $\sigma$-structure $\mathbf A$ and a $\sigma$-sentence $\varphi$, and outputs whether or not $\varphi$ holds in $\mathbf A$, that we denote $\mathbf A \models \varphi$.
The brute-force algorithm solves FO model checking in time $|A|^{O(|\varphi|)}$ where $|\varphi|$ is the quantifier depth of $\varphi$, while no algorithm can solve FO model checking in time $|A|^{o(|\varphi|)}$ when the structures range over general graphs, unless \textsc{SAT} can be solved in subexponential time (see for instance~\cite{Chen06}).

The focus is then on designing fixed-parameter tractable (FPT) algorithms on restricted classes of structures, that is, algorithms running in time $f(|\varphi|) |A|^{O(1)}$ (or better in $f(|\varphi|) |A|$) for some computable function $f$.
FO model checking is a single framework for a large class of problems.
We give the sentences of three examples, which will be relevant in the paper.
\begin{compactitem}
\item \textsc{$k$-Independent Set}: $$\exists x_1, \ldots, x_k \bigwedge_{i \neq j \in [k]} \neg (x_i = x_j)~\land~\neg E(x_i,x_j).$$
\item \textsc{$k$-Biclique}: $$\exists x_1, \ldots, x_k, y_1, \ldots, y_k \bigwedge_{i \neq j \in [k]} \neg (x_i = x_j)~\land~\neg (y_i = y_j)~\land \bigwedge_{i, j \in [k]} \neg (x_i = y_j)~\land~\bigwedge_{i, j \in [k]} E(x_i,y_j).$$
\item \textsc{$k$-Ladder}: $$\exists x_1, \ldots, x_k, y_1, \ldots, y_k \bigwedge_{i \neq j \in [k]} \neg (x_i = x_j)~\land~\neg (y_i = y_j)~\land \bigwedge_{i, j \in [k]} \neg (x_i = y_j)~\land~\bigwedge_{i \leqslant j \in [k]} E(x_i,y_j)$$ $$\land \bigwedge_{i > j \in [k]} \neg E(x_i,y_j).$$ 
\end{compactitem}

\medskip

\textbf{Interpretations, transductions, and monadic dependence.}
Let $\sigma,\tau$ be relational signatures.
A~\emph{simple FO interpretation} (here, \emph{FO interpretation} for short) $\mathsf I$ of $\tau$-structures in $\sigma$-structures consists of the following $\sigma$-formulas: a \emph{domain} formula $\nu(x)$, and for each relation symbol $R \in \tau$ of arity $r$, a~formula $\varphi_R(x_1, \ldots, x_r)$.
If $\mathbf A$ is a $\sigma$-structure, the $\tau$-structure $\mathsf I(\mathbf A)$ has domain $\nu(\mathbf A)=\{v \in A~:~\mathbf A \models \nu(v)\}$ and the interpretation of a relation symbol $R \in \sigma$ of arity $r$ is $R^{\mathsf I(\mathbf A)}=\{(v_1,\dots,v_{r})\in \nu(\mathbf A)^{r}:\mathbf A\models \varphi_R(v_1,\dots,v_r)\}$.
If $\mathcal C$ is a class of $\sigma$-structures, we set $\mathsf I(\mathcal C)=\{\mathsf I(\mathbf A)~:~\mathbf A \in \mathcal C\}$.

Let $\sigma \subseteq \sigma^+$ be relational signatures.
The \emph{$\sigma$-reduct} of a $\sigma^+$-structure $\mathbf A$, denoted by $\text{reduct}_{\sigma^+ \to \sigma}(\mathbf A)$, is the structure obtained from $\mathbf A$ by deleting all the relations not in $\sigma$.
A~\emph{monadic $h$-lift} of a $\sigma$-structure $\mathbf A$ is a $\sigma^+$-structure $\mathbf A^+$, where $\sigma^+$ is the union of $\sigma$ and a~set of $h$ unary relation symbols, and $\text{reduct}_{\sigma^+ \to \sigma}(\mathbf A^+) = \mathbf A$. 

A~\emph{simple non-copying FO transduction} (here, \emph{FO transduction} for short) $\mathsf T$ of $\tau$-structures in $\sigma$-structures is an interpretation of $\tau$-structures in $\sigma^+$-structures, where the $\sigma^+$-structures are \emph{monadic $h$-lifts} of $\sigma$-structures for some fixed integer $h$.
As there are many ways of interpreting the extra unary relations, a transduction (contrary to an interpretation) builds on a given input structure several output structures.
If $\mathcal C$ is a class of $\sigma$-structures, $\mathsf T(\mathcal C)$ denotes the class of all the $\tau$-structures output on any~$\sigma$-structure $\mathbf A \in \mathcal C$.  

We say that a class $\mathcal C$ \emph{interprets} a class $\mathcal D$ (or that \emph{$\mathcal D$ interprets in $\mathcal C$}) if there is an interpretation $\mathsf I$ such that $\mathcal D \subseteq \mathsf I(\mathcal C)$.
Further, a~class $\mathcal C$ \emph{efficiently interprets} $\mathcal D$ if additionally a polytime algorithm inputs $\mathbf A \in \mathcal D$, and outputs a structure $\mathbf B \in \mathcal C$ such that $\mathsf I(\mathbf B)$ is isomorphic to $\mathbf A$.
Similarly, we say that a class $\mathcal C$ \emph{transduces} a class $\mathcal D$ (or that \emph{$\mathcal D$ transduces in $\mathcal C$}) if there is a transduction $\mathsf T$ such that $\mathcal D \subseteq \mathsf T(\mathcal C)$.
Two classes $\mathcal C$ and $\mathcal D$ are \emph{transduction equivalent} if $\mathcal C$ transduces $\mathcal D$, and $\mathcal D$ transduces $\mathcal C$. 
We will frequently use the fact that one can compose transductions: If $\mathcal C$ transduces $\mathcal D$, and $\mathcal D$ transduces $\mathcal E$, then $\mathcal C$ transduces $\mathcal E$. 

 We will not need the original definition of monadic dependence; solely the following characterization:
 \begin{theorem}[Baldwin and Shelah \cite{BS1985monadic}]\label{thm:baldwin-shelah}
 $\mathcal C$ is monadically dependent if and only if $\mathcal C$ does not transduce the class $\mathcal G$ of all finite graphs.
 \end{theorem}

Since FO~model checking on the class of all graphs is $\AW[*]$-hard~\cite{Downey96}, one gets:
\begin{theorem}[\cite{BS1985monadic,Downey96}]\label{thm:AW-hard}
  If $\mathcal C$ efficiently interprets the class of all graphs then FO model checking on $\mathcal C$ is $\AW[*]$-hard.
\end{theorem}
\cref{conj:fomc} anticipates that every hereditary class of structures not transducing the class of all graphs admits an FPT FO model checking, and no other hereditary class does.

A~notion more restrictive to that of monadic dependence is \emph{monadic stability}.
A~class is said \emph{monadically stable} if it does not transduce the class of all half-graphs.
A~monadically stable class has in particular to be $H_t$-free. 

Let us finally notice that, throughout the paper, the ambient logic will be first-order.
Thus we may for instance drop the ``FO'' before ``interpretation'' and ``transduction.''

\subsection{Rank divisions, universal patterns, and twin-width}\label{subsec:rd}

A~\emph{division} $\mathcal D$ of a matrix $M$ is a pair $(\mathcal D^R,\mathcal D^C)$, where $\mathcal D^R$ (resp.~$\mathcal D^C$) is a partition of the rows (resp.~columns) of $M$ into intervals of consecutive rows (resp.~columns).
Each element of $\mathcal D^R$ (resp.~$\mathcal D^C$) is called a~\emph{row part} (resp.~\emph{column part}). 
A~\emph{$k$-division} is a division satisfying $|\mathcal D^R| = |\mathcal D^C| = k$.
We often list the row (resp.~column) parts of $D^R$ (resp.~$D^C$) $R_1, R_2, \ldots, R_k$ (resp.~$C_1, C_2, \ldots, C_k$) when $R_i$ is just below $R_{i+1}$ (resp.~$C_j$ is just to the left of $C_{j+1}$). 
For every pair $R_i \in \mathcal D^R$, $C_j \in \mathcal D^C$, the (contiguous) submatrix $R_i \cap C_j$ is called~\emph{cell} or \emph{zone} of~$\mathcal D$, or more precisely, the \emph{$(i,j)$-cell} of~$\mathcal D$.
Note that a~$k$-division defines $k^2$ zones.
We say that a cell, or more generally a matrix, is \emph{empty} or \emph{full 0} if all its entries are 0.
The \emph{dividing lines} of $\mathcal D^R=R_1,R_2,\ldots$ (resp.~$\mathcal D^C=C_1,C_2,\ldots$) are the strips (of width 2) made by the last row of $R_i$ and the first row of $R_{i+1}$ (resp.~last column of $C_j$ and the first column of $C_{j+1}$.
A dividing line of $\mathcal D^R$ (resp.~$\mathcal D^C$) \emph{stabs} a set of rows (resp.~of columns) if it intersects~it. 
We may call~\emph{regular $k$-division} a $k$-division where every row part and column part have the same size.

A~\emph{rank-$k$ division} of $M$ is a $k$-division $\mathcal D$ such that for every $R_i \in {\mathcal D}^R$ and $C_j \in {\mathcal D}^C$ the zone $R_i \cap C_j$ has at least $k$ distinct rows or at least $k$ distinct columns (that is, combinatorial rank at least $k$).
By \emph{large rank division}, we informally mean a rank-$k$ division for arbitrarily large values of $k$.
The maximum integer $k$ such that $M$ admits a rank-$k$ division is called \emph{grid rank}, and is denoted by $\gr(M)$.

An adjacency matrix $M$ of a binary structure encodes in any\footnote{As the chosen encoding will not matter, we do not specify a \emph{canonical} one.} bijective fashion the \emph{atomic type} of every pair of vertices $(u,v)$ (i.e., the set of atomic propositions the pair $(u,v)$ satisfies) at position $(u,v)$ in $M$.
We denote by $\adj{\prec}{\mathbf A}$ the adjacency matrix of $\mathbf A$ \emph{along $\prec$}, a total order on $A$.
The~\emph{grid rank} of a binary structure $\mathbf A$, denoted by $\gr(\mathbf A)$, is the least integer $k$ such that there is a total order $\prec$ of $A$ satisfying $\gr(\adj{\prec}{\mathbf A})=k$.

We will not need the original definition of twin-width (presented in the introduction) generalized to binary structures.\footnote{The definition is similar with red edges appearing between the contraction of $u$ and $v$, and vertex $z$ whenever $(u,z)$ and $(v,z)$ have different atomic types. We refer the curious reader to~\cite{twin-width1}.}
So we do \emph{not} reproduce it here.
Instead we recall that the twin-width and the grid rank of a binary structure are functionally equivalent, and we encourage the reader to think of the twin-width of $\mathbf A$, $\tww(\mathbf A)$, simply as its grid rank $\gr(\mathbf A)$.

\begin{theorem}[\cite{twin-width4}]\label{thm:mixed-number-gen2}
  There is a computable function $f: \mathbb N \to \mathbb N$ such that for every binary structure~$\mathbf A$, the following two implications hold:
  \begin{compactitem}
  \item If $\tww(\mathbf A) \leqslant k$, then there is a total order $\prec$ of $A$ such that $\gr(\adj{\prec}{\mathbf A}) \leqslant f(k)$, and
  \item If there is a total order $\prec$ of $A$ such that $\gr(\adj{\prec}{\mathbf A}) \leqslant k$, then $\tww(\mathbf A) \leqslant f(k)$.
  \end{compactitem}
  Furthermore there are computable functions $g, h: \mathbb N \to \mathbb N$ and an algorithm running in time $h(k) \cdot |A|^{O(1)}$ which inputs an adjacency matrix $\adj{\prec}{\mathbf A}$ without rank-$k$ division and outputs a $g(k)$-sequence of $\mathbf A$.
\end{theorem}

It was shown in a previous paper of the series \cite{twin-width4} that highly-structured rank divisions can always be found in large rank divisions.
We now make that statement precise.
Let $\mat{k}{0}$ be the $k^2 \times k^2$ permutation matrix such that if $\mat{k}{0}$ is divided in $k$ row parts and $k$ column parts, each of size $k$, there is exactly one 1 entry in each cell of the division, and this 1 entry is at position $(i,j)$ of the $(j,i)$-cell; see~leftmost matrix in~\cref{fig:canonicaler}. 
For every positive integer $k$ and $s \in \{1,\uparrow,\downarrow,\leftarrow,\rightarrow\}$, let $\mat{k}{s}$ be the $k^2 \times k^2$ $0,1$-matrix defined from $\mat{k}{0}$ by doing one of the following operations:
\begin{compactitem}
\item switching 1 entries and 0 entries, if $s=1$,
\item turning 0 entries into 1 entries if there is a 1 entry somewhere below them, if $s =\, \uparrow$,
\item turning 0 entries into 1 entries if there is a 1 entry somewhere above them, if $s =\, \downarrow$,
\item turning 0 entries into 1 entries if there is a 1 entry somewhere to their right, if $s =\, \leftarrow$,
\item turning 0 entries into 1 entries if there is a 1 entry somewhere to their left, if $s =\, \rightarrow$.
\end{compactitem}  
We call $\mat{k}{s}$ a \emph{universal pattern} and $\{\mat{k}{s}~:~k \in \mathbb N\}$ a \emph{permutation-universal family}; see~\cref{fig:canonicaler} for an illustration.

\begin{figure}[h!]
    \centering
    \begin{tikzpicture}[scale=.235]

   \foreach \symb/\b/\xsh/\ysh in {{==}/0/-30/-10, {>}/0/-10/-10,{<}/0/0/-10,{<}/1/10/-10,{>}/1/20/-10}{
   \begin{scope}[xshift=\xsh cm,yshift=\ysh cm]     
    \foreach \i in {0,...,8}{
      \foreach \j in {0,...,8}{
        \pgfmathsetmacro{\ip}{\i+1}
        \pgfmathsetmacro{\jp}{\j+1}
        \pgfmathsetmacro{\col}{ifthenelse(\i == mod(\j,3)*3+floor(\j/3),"black",ifthenelse(\b==1,ifthenelse(\i \symb mod(\j,3)*3+floor(\j/3),"black","white"),ifthenelse(\j \symb mod(\i,3)*3+floor(\i/3),"black","white")))}
        \fill[\col] (\i,\j) -- (\i,\jp) -- (\ip,\jp) -- (\ip,\j) -- cycle;
      }
    }
    \draw[line width=0.75pt, scale=1, color=black!20!blue] (0, 0) grid (9, 9);
    \draw[line width=1.25pt, scale=3, color=black!10!yellow] (0, 0) grid (3, 3);
   \end{scope}
   }
   \foreach \symb/\b/\xsh/\ysh in {{!=}/0/-20/-10}{
   \begin{scope}[xshift=\xsh cm,yshift=\ysh cm]     
    \foreach \i in {0,...,8}{
      \foreach \j in {0,...,8}{
        \pgfmathsetmacro{\ip}{\i+1}
        \pgfmathsetmacro{\jp}{\j+1}
        \pgfmathsetmacro{\col}{ifthenelse(\i == mod(\j,3)*3+floor(\j/3),"white","black")}
        \fill[\col] (\i,\j) -- (\i,\jp) -- (\ip,\jp) -- (\ip,\j) -- cycle;
      }
    }
    \draw[line width=0.75pt, scale=1, color=black!20!blue] (0, 0) grid (9, 9);
    \draw[line width=1.25pt, scale=3, color=black!10!yellow] (0, 0) grid (3, 3);
   \end{scope}
   }
    \end{tikzpicture}
    \caption{The six universal patterns with $k=3$. The black cells always represent 1 entries, and white cells, 0 entries. From left to right, $\mat{3}{0}$, $\mat{3}{1}$, $\mat{3}{\uparrow}$, $\mat{3}{\downarrow}$, $\mat{3}{\leftarrow}$, and $\mat{3}{\rightarrow}$.
    We always adopt the convention that the matrix entry at position $(1,1)$ is the bottom-left one.}
    \label{fig:canonicaler}
  \end{figure}
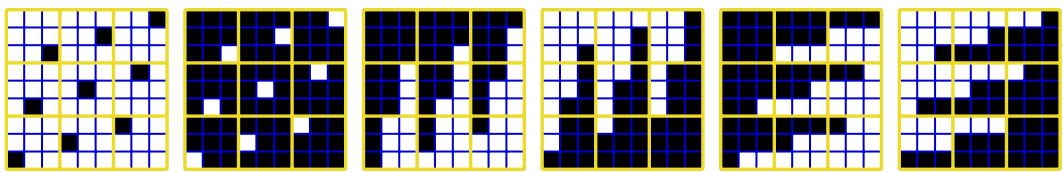

It was shown that, taking the adjacency matrix of a graph $G$ along some order, either yields a matrix with bounded grid rank, and \cref{thm:mixed-number-gen2} effectively gives a favorable contraction sequence of $G$, or yields a matrix with huge grid rank, wherein a large universal pattern can be extracted:

\begin{theorem}[\cite{twin-width4}]\label{thm:canonicaler}
Given $M$ an adjacency matrix of an $n$-vertex graph $G$, and an integer~$k$, there is an $f(k)n^{O(1)}$-time algorithm which either returns
\begin{compactitem}
\item $\mat{k}{s}$ as an off-diagonal submatrix of $M$, for some $s \in \{0,1,\uparrow,\downarrow,\leftarrow,\rightarrow\}$,
\item or a contraction sequence certifying that $\tww(G) \leqslant g(k)$.
\end{compactitem}
where $f$ and $g$ are computable functions.
\end{theorem}

In the previous theorem, an \emph{off-diagonal} submatrix of a square matrix is entirely contained strictly above the diagonal, or entirely contained strictly below it.
In particular, its row indices and column indices are disjoint.
For graphs (and symmetric matrices in general), we can further impose that the off-diagonal submatrix is above the diagonal.

The previous theorem is a consequence of~\cref{thm:mixed-number-gen2} and:
\begin{theorem}\label{thm:rd-to-up}
  There is a computable function $f: \mathbb N \to \mathbb N$ such that every $0,1$-matrix with grid rank at least $f(k)$ has $\mat{k}{s}$ as an off-diagonal submatrix for some $s \in \{0,1,\uparrow,\downarrow,\leftarrow,\rightarrow\}$.
\end{theorem}

We will sometimes need a generalization of~\cref{thm:canonicaler} to arbitrary binary structures.
\begin{theorem}[\cite{twin-width4}]\label{thm:canonicaler-gen}
Given $M$ an adjacency matrix of a binary structure $\mathbf A$ on $n$ elements, and an integer~$k$, there is an $f(k)n^{O(1)}$-time algorithm which either returns
\begin{compactitem}
\item an isomorph of $\mat{k}{s}$ as an off-diagonal submatrix of $M$, for some $s \in \{0,1,\uparrow,\downarrow,\leftarrow,\rightarrow\}$,
\item or a contraction sequence certifying that $\tww(\mathbf A) \leqslant g(k)$.
\end{compactitem}
where $f$ and $g$ are computable functions.
\end{theorem}
An~\emph{isomorph} $N$ of a matrix $M$ is obtained by applying a bijection to the entries of $M$.

In~\cite{twin-width4},~\cref{thm:canonicaler-gen} was used for ordered binary structures.
More specifically, the order~$\prec$ such that $M=\adj{\prec}{\mathbf A}$ was a binary relation of the structure $\mathbf A$.
In that context, the first item was lowerbounding the twin-width; thereby providing an FPT approximation for computing twin-width of ordered binary structures.
In the current paper, $\prec$ need not be part of the structure $\mathbf A$.
Thus the isomorph of $\mat{k}{s}$ does not \emph{a priori} provide any lower bound on the twin-width.
For all we know, another order $\prec'$ could be such that $\adj{\prec'}{\mathbf A}$ has no large rank division.


We now give a couple of lemmas that are useful to bound the grid rank of a matrix. 
\begin{lemma}\label{lem:x-vs-yz}
  There is a computable function $f$ such that the following holds.
  Let $G$ be a graph, $\prec$ a total order on $V(G)$, and $X,Y,Z \subseteq V(G)$.
  If $\adj{\prec}{G \langle X,Y \rangle}$ and $\adj{\prec}{G \langle X,Z \rangle}$ have both grid rank less than $k$, then $\adj{\prec}{G \langle X,Y \cup Z \rangle}$ has grid less than~$f(k)$.
\end{lemma}
\begin{proof}
  Let $k$ be such that $\max(\gr(\adj{\prec}{G \langle X,Y \rangle}), \gr(\adj{\prec}{G \langle X,Z \rangle})) < k$ and let $q = \max(\Bipram{k}{k},2^k)$.
  We claim that $\gr(\adj{\prec}{G \langle X,Y \cup Z \rangle})<q$.

  Assume for the sake of contradiction that $\mathcal D$ is a rank-$q$ division of $\adj{\prec}{G \langle X,Y \cup Z \rangle}$.
  Color each cell $C$ of $\mathcal D$ by $Y$ if its intersection with $Y$ has at least $k$ distinct rows, and by $Z$, otherwise.
  Observe that if the latter holds, then the intersection of $C$ with $Z$ has at least $\min(\lfloor \log q \rfloor,q-k) \geqslant k$ distinct rows.
  Otherwise $C$ would have less than $q$ distinct rows \emph{and} less than $q$ distinct columns; a contradiction to $\mathcal D$ being a rank-$q$ division.

  Since $q \geqslant \Bipram{k}{k}$, by~\cref{thm:bipRamsey}, $\mathcal D$ contains $k$ row parts and $k$ column parts whose $k^2$ cells are all colored $Y$, or all colored $Z$; contradicting $\gr(\adj{\prec}{G \langle X,Y \rangle}) < k$ or $\gr(\adj{\prec}{G \langle X,Z \rangle}) < k$. 
\end{proof}


If two structures $\mathbf A, \mathbf B$ have the same domain (but not necessarily the same signatures $\sigma_A$ and $\sigma_B$), then we denote by $\mathbf A \cup \mathbf B$ the $(\sigma_A \cup \sigma_B)$-structure obtained by concatenating the relations of $\mathbf A$ and the relations of $\mathbf B$.
\begin{lemma}\label{lem:union}
  There is a computable function $f$ such that the following holds.
  For every pair of structures $\mathbf A, \mathbf B$ on the same domain $A$, and total order $\prec$ on $A$, such that $\adj{\prec}{\mathbf A}$ and $\adj{\prec}{\mathbf B}$ both have grid rank less than~$k$, then $\gr(\adj{\prec}{\mathbf A \cup \mathbf B}) < f(k)$. 
\end{lemma}
\begin{proof}
  Let $k$ be such that $\max(\gr(\adj{\prec}{\mathbf A}), \gr(\adj{\prec}{\mathbf B})) < k$ and let $q = \max(\Ram{k}{k},$ $\Bipram{k}{k})$.
  We claim that $\gr(\adj{\prec}{\mathbf A \cup \mathbf B})<q$.
  Note that the entries $\adj{\prec}{\mathbf A \cup \mathbf B}$ can be seen as ranging in the direct product $\Sigma = \Sigma_A \times \Sigma_B$ of the alphabet $\Sigma_A$ of $\adj{\prec}{\mathbf A}$ with the alphabet $\Sigma_B$ of $\adj{\prec}{\mathbf B}$.

  For the sake of contradiction, assume that there is a rank-$q$ division $\mathcal D$ of $\adj{\prec}{\mathbf A \cup \mathbf B}$.
  For each cell $C$ of $\mathcal D$, say, $C$ has a set $Z$ of $q$ distinct rows ($q$ distinct columns would be dealt with similarly).
  Since two rows $z,z' \in Z$ are different if their projections on $\Sigma_A$ or on $\Sigma_B$ are different.
  In the former case, we put an edge of color $c_A$ between $z$ and $z'$, and in the latter case, an edge of color $c_B$.
  Applying~\cref{thm:Ramsey} to this auxiliary graph, we get a monochromatic clique of $k$ rows of $Z$ with distinct projection on $\Sigma_X$ with $X \in \{A,B\}$.
  This is because $q \geqslant \Ram{k}{k}$.

  We color cell $C$ by $X \in \{A,B\}$, and apply~\cref{thm:bipRamsey} to the division $\mathcal D$.
  We obtain a $k$-division $\mathcal D'$ (included in $\mathcal D$) where all $k^2$ cells have the same color $X$ for some $X \in \{A,B\}$.
  This is because $q \geqslant \Bipram{k}{k}$.
  The division $\mathcal D'$ is a rank-$k$ division in either $\mathbf A$ or $\mathbf B$; a contradiction.
\end{proof}


\medskip

We recall that FO model checking is FPT on classes of bounded twin-width, when the inputs are given with a contraction sequence.
\begin{theorem}[\cite{twin-width1}]\label{thm:fomc}
  Let $\mathbf A$ be a binary structure of bounded twin-width given with a~\mbox{$d$-sequence}, and $\varphi$ be an FO sentence of quantifier rank~$k$.
  Then one can decide $\mathbf A \models \varphi$ in linear FPT time $f(d,k) \cdot |A|$.
\end{theorem}

Finally the following is a particularly useful fact to bound the twin-width of a class.
\begin{theorem}[\cite{twin-width1}]\label{thm:transduction}
  Every FO transduction of a class with bounded twin-width has bounded twin-width.
  
Furthermore, given an FO transduction $\mathsf T$ and a class $\mathcal C$ on which $0(1)$-sequences can be computed in polynomial time, one can also compute $O(1)$-sequences for graphs of $\mathsf T(\mathcal C)$ in polynomial time. 
\end{theorem}

It is known \cite{twin-width2,twin-width&permutations} that classes of bounded twin-width have exponential growth.
Thus by the contrapositive, classes of super-exponential growth, like the following ones, have unbounded twin-width.
\begin{theorem}[\cite{twin-width2}]\label{thm:unbounded-tww}
  The following classes have unbounded twin-width:
  \begin{compactitem}
  \item the class $\mathcal G_{\leqslant 3}$ of every subcubic graph;
  \item the class $\mathcal B_{\leqslant 3}$ of every bipartite subcubic graph;
  \item the 2-subdivision of every biclique $K_{n,n}$.
  \end{compactitem}
\end{theorem}

\subsection{Twin-width win-wins}

A graph parameter $p$ is said~\emph{FO definable} if there is a function that inputs a positive integer~$k$ and outputs a first-order sentence $\varphi_k$ such that for every graph $G$, $p(G)=k$ if and only if $G \models \varphi_k$.
It is further~\emph{effectively FO definable} if an algorithm realizes that function and takes time $f(k)$ for some computable function $f$.

Examples of effectively FO definable parameters include $\alpha, \beta, \gamma, \lambda, \omega$, the size of a largest independent set, largest semi-induced biclique, smallest dominating set, largest semi-induced ladder, and largest clique, respectively.
On the contrary, the chromatic number $\chi$ is not FO definable: No FO sentence can express the fact that a graph is 2-chromatic (let alone, 3-chromatic).
All the mentioned parameters are at least W$[1]$-hard to compute (with $\gamma$ being even W$[2]$-complete, and $\chi$ paraNP-complete).

We say that a parameter $q$ is \emph{$p$-bounded on class $\mathcal C$}, denoted by \emph{$q \sqsubseteq p$ on $\mathcal C$} or \emph{$q \sqsubseteq^{\mathcal C} p$}, if there is a non-decreasing function $f$ such that for every graph $G \in \mathcal C$, $q(G) \leqslant f(p(G))$.
We say that twin-width is \emph{effectively $p$-bounded on $\mathcal C$}, denoted by \emph{$\tww \sqsubseteq_{\text{eff}} p$ on $\mathcal C$} or \emph{$\tww \sqsubseteq^{\mathcal C}_{\text{eff}} p$}, if further there is an algorithm that outputs a $g(p(G))$-sequence for every graph $G \in \mathcal C$ in time $h(p(G)) \cdot |V(G)|^{O(1)}$ for some computable functions $g, h$.

The following reduces the task of showing that an FO definable parameter $p$ is FPT on $\mathcal C$ to showing that $\tww \sqsubseteq^{\mathcal C}_{\text{eff}} p$ holds.

\begin{theorem}\label{thm:win-win}
  Let $p$ be an effectively FO definable parameter, and $\mathcal C$ a class such that $\tww \sqsubseteq^{\mathcal C}_{\text{eff}} p$.
  Then $p(G) \geqslant k$ for $G \in \mathcal C$ can be decided in FPT time $f(p(G)) \cdot |V(G)|^{O(1)}$ for some computable function $f$.
\end{theorem}
\begin{proof}
  We run the algorithm witnessing that twin-width is effectively $p$-bounded on $\mathcal C$.
  We obtain a $d$-sequence $\mathcal S$ of the input graph $G$, with $d \leqslant g_{\mathcal C}(p(G))$ and $g_{\mathcal C}$, the corresponding bounding function.
  This takes time $f_1(p(G)) \cdot |V(G)|^{O(1)}$, for some computable function~$f_1$.

  If $d \leqslant g_{\mathcal C}(k)$, we call the FO model-checking algorithm of~\cite{twin-width1} (see \cref{thm:fomc}) on input $(G, \mathcal S, \varphi_i)$ for $i$ going from 0 to $k-1$.
  As $p$ is effectively FO definable, computing each $\varphi_i$ takes time $f_2(i)$ for some computable function $f_2$.
  The first time a call gets a positive answer, we correctly report that $p(G) < k$, and we have actually computed the exact value of $p(G)$.
  If all these calls get negative answers, we know that $p(G) \geqslant k$.
  If instead $d > g_{\mathcal C}(k)$, then as by assumption $d \leqslant g_{\mathcal C}(p(G))$ and $g_{\mathcal C}$ is non-decreasing, we can correctly report that $p(G) \geqslant k$.

  The overall running time is
  $$h(p(G)) \cdot |V(G)|^{O(1)} + \sum\limits_{0 \leqslant i \leqslant \min(k-1,p(G))} f_2(i) \cdot h(i,g_{\mathcal C}(p(G))) |V(G)|$$
  when $h(\cdot,\cdot)$ is the dependency in the parameters of~\cref{thm:fomc}.
  This running time is bounded by $f(p(G)) \cdot |V(G)|^{O(1)}$ for some computable function $f$.
\end{proof}

\subsection{Delineation}

A class $\mathcal D$ is said \emph{delineated by twin-width} (or \emph{delineated} for short) if for every hereditary closure $\mathcal C$ of a subclass of $\mathcal D$, it holds that $\mathcal C$ has bounded twin-width if and only if $\mathcal C$ is monadically dependent.
A class $\mathcal D$ is \emph{effectively delineated by twin-width} (or \emph{effectively delineated}) if further there is an FPT approximation of contraction sequences on~$\mathcal D$, that is, two computable functions $f, g$ and an algorithm that takes a binary structure $\mathbf A \in \mathcal D$ and an integer $k$ as input, and in time $f(k) \cdot |A|^{O(1)}$ outputs a $g(k)$-sequence for $\mathbf A$ or correctly reports that $\tww(\mathbf A) > k$.

Showing that a class $\mathcal D$ is effectively delineated establishes that, as far as efficient (that is, FPT) FO model checking is concerned,~\cref{thm:fomc} gives a complete picture of what happens on $\mathcal D$.
Indeed it is unlikely that a monadically independent class admits an FPT algorithm for FO model checking (see~\cref{subsec:prelim-fmt}).
Trivially, every class with bounded twin-width is delineated, and every class where $O(1)$-sequences can be found in polynomial time (see~\cite{twin-width2}) is effectively delineated.
We now list some non-trivial examples of (effectively) delineated classes.

\begin{theorem}[\cite{twin-width1,hlineny22,twin-width4} $+$ this paper]\label{thm:delineated}
  The following classes of binary structures are effectively delineated:
  \begin{compactitem}
  \item permutation graphs~\cite{twin-width1}; and even,
  \item circle graphs~\cite{hlineny22};\footnote{Hlinen{\'{y}} and Pokr{\'{y}}vka~\cite{hlineny22} show that any hereditary subclass of circle graphs excluding a~permutation graph has (effectively) bounded twin-width, which implies delineation, since the class of all permutation graphs have unbounded twin-width.}
  \item ordered graphs~\cite{twin-width4};
  \item interval graphs (see~\cref{subsec:ig}); and even,
  \item rooted directed path graphs (see~\cref{subsec:rdpg}).
  \end{compactitem}
\end{theorem}

Our proofs of (effective) delineation will follow the same path.
Either an $O(1)$-sequence of the graph is found (bounded twin-width) or an arbitrarily large semi-induced generalized transversal pair is detected.
We now see that the latter implies monadic independence (hence, in particular, unbounded twin-width).

\begin{lemma}\label{thm:tk-mi}
  Let $\ell$ be a fixed non-negative integer. 
  Let $\mathcal C$ be a hereditary class containing a semi-induced generalized transversal pair of half-graphs $T_{n,\ell}$, for every positive integer $n$.
  Then $\mathcal C$ is monadically independent. 
\end{lemma}
\begin{proof}
  It is folklore that the class $\mathcal M_b$ of all totally ordered bipartite matchings is monadically independent (see for instance~\cite{twin-width4,twin-width&permutations}).
  By \emph{totally ordered bipartite matching}, we mean two sets $X, Y$ of same cardinality, with a total order over $X \cup Y$ such that $X$ and $Y$ are each an interval along that order, and a matching between $X$ and $Y$.
  We shall just argue that $\mathcal M_b$ transduces in $\mathcal C$.
  We first show the lemma when $\ell=0$, that is, $\mathcal C$ contains a semi-induced $T_{n,0} = T_n$ for every $n$.  

  Let $(G=(X,Y,E(G)),\prec)$ be any member of $\mathcal M_b$.
  Let $x_1 \prec x_2 \prec \ldots \prec x_n$ be the elements of $X$, and $y_1 \prec y_2 \prec \ldots \prec y_n$, the elements of $Y$.
  Finally let $\pi$ be the permutation such that, for every $i \in [n]$, $x_iy_j \in E(G)$ if and only if $j=\pi(i)$.

  Let $(A,B,C)$ be the tripartition of a semi-induced $T_n$ in $\mathcal C$.
  The transduction $\mathsf T$ guesses the tripartition $(A,B,C)$ with 3 corresponding unary relations.
  Eventually $(X,Y)$ will be a subset of $(A,B)$.
  We interpret a total order on $A \cup B$ by $$x \prec y \equiv (A(x)~\land~B(y))~\lor~\big(x \neq y~\land~A(x)~\land~A(y)~\land~\forall z (B(z)~\land~E(x,z)) \rightarrow E(y,z)\big)$$
  $$\lor~\big(x \neq y~\land~B(x)~\land~B(y)~\land~\forall z (C(z)~\land~E(x,z)) \rightarrow E(y,z)\big).$$
  We then interpret a matching between $A$ and $B$ by $$\varphi(x,y) \equiv A(x)~\land~B(y)~\land~E(x,y)~\land~\forall z (z \prec x \rightarrow \neg E(z,y)).$$
  Observe that $\varphi$ and $\prec$ define a universal structure for totally ordered bipartite matchings on $2n$ vertices.

  In particular, a fourth unary relation can guess the domain $(X \subseteq A,Y \subseteq B)$, by picking the rows and columns of the biadjacency matrix $\adj{\prec}{A,B,\{ab~:~T_n \models \varphi(a,b)\}}$ corresponding to the 1 entries, which, in the regular $n$-division falls in the $(i,\pi(i))$-cells with $i \in [n]$.
  Thus $\mathsf T(T_n)$ outputs $(G,\prec)$.

  We now deal with the general case by reducing it to $\ell=0$.
  For that, we transduce a~semi-induced $T_n$ in a~semi-induced $T_{n,\ell}$.
  The transduction is imply based on the definition of generalized transversal pairs.
  It uses $3+\ell$ unary relations $A, B_0, \ldots, B_\ell, C$, redefines the domain as $A \cup B_0 \cup C$, keeps the edges between $A$ and $B_0$, and adds an edge between $x \in B_0$ and $y \in C$ if and only if there is a path from $x$ to $y$ going through $B_1, B_2, \ldots, B_\ell$, in this order.
  All of this is easily expressible in first-order logic.
\end{proof}

We deduce the following.
\begin{lemma}\label{thm:seq-or-tk}
  Let $\ell$ be a fixed non-negative integer. 
  Let $f: \mathbb N \to \mathbb N$ be any computable function, and $\mathcal C$ be a graph class.
  If for every natural $k$ and $G \in \mathcal C$, either $G$ admits an $f(k)$-sequence or $G$ has a semi-induced generalized transversal pair $T_{k,\ell}$, then $\mathcal C$ is delineated.

  Furthermore, if the contraction sequence can be found in time $g(k) \cdot |V(G)|^{O(1)}$ for some computable function $g$, then $\mathcal C$ is effectively delineated.
\end{lemma}
\begin{proof}
  Let $\mathcal D$ be a hereditary subclass of $\mathcal C$.
  Let $k \in \mathbb N \cup \{+\infty\}$ be the maximum generalized natural such that $\mathcal D$ admits a semi-induced generalized transversal pair $T_{k,\ell}$.
  If $k = +\infty$, then we conclude that $\mathcal D$ is monadically independent by~\cref{thm:tk-mi}, and thus has unbounded twin-width.
  If $k$ is finite, then $\mathcal D$ has bounded twin-width (by $f(k)$), and hence is monadically dependent. 
\end{proof}
By~\cref{thm:mixed-number-gen2}, in the previous theorem, the \emph{$f(k)$-sequence of $G$} can be replaced by an adjacency matrix of $G$ of grid rank at most $f(k)$. 

\medskip

On the contrary, the class of subcubic graphs is \emph{not} delineated.
Indeed the whole class is monadically dependent~(see for instance \cite{Podewski78}), even monadically stable, but has unbounded twin-width~\cite{twin-width2}.
We will see that the classes of segment graphs (even with some further restrictions) and visibility graphs of simple polygons are also \emph{not} delineated.
In some sense, what we do is to reduce to the easy case of subcubic graphs.

\begin{lemma}\label{lem:not-delineated}
If $\mathcal C$ admits a subclass which is transduction equivalent to $\mathcal G_{\leqslant 3}$ or to $\mathcal B_{\leqslant 3}$, then $\mathcal C$ is not delineated.
\end{lemma}
\begin{proof}
  It is easy to check that $\mathcal G_{\leqslant 3}$ and $\mathcal B_{\leqslant 3}$ are transduction equivalent, so we assume without loss of generality that $\mathcal C$ has a subclass $\mathcal D$ transduction equivalent to the class of all subcubic graphs.
  As $\mathcal D$, and a fortiori its hereditary closure $\mathcal D'$, transduces $\mathcal G_{\leqslant 3}$, the class $\mathcal D'$ cannot have bounded twin-width by~\cref{thm:transduction,thm:unbounded-tww}.
  As $\mathcal G_{\leqslant 3}$ transduces $\mathcal D$ and thus, its hereditary closure $\mathcal D'$, the class $\mathcal D'$ is monadically independent, and even monadically stable~\cite{Podewski78}.
\end{proof}

We will prove that the above mentioned classes are transduction equivalent to the class of all (bipartite) subcubic graphs and invoke~\cref{lem:not-delineated}.
Again, that shows something stronger: The considered classes contain a subclass whose hereditary closure is monadically stable and of unbounded twin-width. 

\medskip

\textbf{Organization.}
In the rest of the paper, we prove the results we mentioned concerning intersection graphs with a path or tree model (in~\cref{sec:paths-trees}), segment graphs (in~\cref{sec:segments}), and visibility graphs (in~\cref{sec:visibility}).
We end with a handful of open questions.

\section{Intersection graphs on trees and paths}\label{sec:paths-trees}

In this section, we show that interval graphs, and more generally rooted directed path graphs, are delineated. A formal statement is deferred 
till necessary definitions are ready.

\subsection{Interval graphs}\label{subsec:ig}

Let $G$ be an interval graph and $\I_G=\{I_v~:~v\in V(G)\}$ be an interval representation of~$G$ over the real line, 
where the interval $I_v$ is of the form $[i,j]$ for some integers $1 \leq i\leq j$. 
When the graph~$G$ is clear from the context, we omit the subscript of $\I_G$.
For every $v\in V(G)$ and the associated interval $I_v$, the first and the second entries of $I_v$ 
are called the left and the right endpoint of~$v$, and denoted as $\ell_v$ and $r_v$, respectively. 

This subsection is devoted to proving that interval graphs are delineated and intended as a preamble 
before considering the rooted directed path graphs. Moreover, interval graphs 
appear as a subcase in the proof for rooted directed path graphs.

\begin{proposition}\label{prop:delineation-interval}
  The class of all interval graphs is effectively delineated.
\end{proposition}

We recall that every class of bounded twin-width is monadically dependent.
To show that the hereditary closure of every monadically dependent subclass of interval graphs has bounded twin-width, we shall henceforth find a semi-induced transversal pair of half-graphs $T_t$ in any interval graph $G$ of twin-width at least $f(t)$. 
The delineation of interval graphs then follows from \cref{thm:tk-mi}.

For each interval graph $G$ with an interval representation $\I,$ 
we associate a total order $\prec$, following a lexicographic order on $\I$, defined in the following way.
We say $u\prec v$ if and only if $\ell_u<\ell_v$ or, $\ell_u=\ell_v$ and $r_u<r_v$. 
We assume that $G$ does not contain twins, which does not 
impede the generality of the proof. We also assume that $\I$ is a minimal representation 
in the sense that the sum of interval lengths is as small as possible. A consequence of 
the minimal representation assumption is that if $\ell_u < \ell_w$ for vertices $u,w\in V(G)$,
there exists a vertex $v\in V(G)$ such that $\ell_u \leq r_v <\ell_w$ since otherwise, 
the interval $I_u$ can be shortened by setting the new value for $\ell_u$ to $\ell_w$. This is the only property 
we use from the minimality assumption on $\I.$

Let $\C$ be a class of interval graphs of unbounded twin-width. 
Fix an arbitrary integer $t$, and let $G\in \C$ be an interval graph such that  the adjacency matrix $\adj{\prec}{G}$ of $G$ with rows and columns ordered by $\prec$ has a  rank-$4t^2$ division. 
Such a graph exists in $\mathcal{C}$ by \cref{thm:mixed-number-gen}. Let $\I$ be an interval representation of $G$. 
 A rank-$8t^2$ division of $\adj{\prec}{G}$ can be viewed as a $(2,2)$-division of $\adj{\prec}{G}$, each cell having a rank-$4t^2$ division. 
By choosing a suitable cell of the $(2,2)$-division, one can extract two disjoint vertex sets $A$ and $B$ of $V(G)$ such that 
(i) $A\prec B$ or $B\prec A$,  (ii) $A$ admits a division  $\mathcal A=\{A_1,\ldots, A_{2t^2}\}$ into $2t^2$ (non-empty) sets 
with $A_1\prec \dots\prec A_{2t^2}$, and (iii) $B$ admits a division $\mathcal B=\{B_1,\ldots, B_{4t^2}\}$ into $4t^2$ (non-empty) sets 
with $B_1\prec \dots \prec B_{4t^2}$. Without loss of generality, we assume $A\prec B$.

\bigskip

\noindent {\bf Exclusive region covering left endpoints of each $A_i$ and $B_j$.} 
For  vertices $u,v,w\in V(G)$ with $u\prec v\prec w$, observe that $(u,w)\in E(G)$ implies $(u,v)\in E(G)$. 
This means that in $\adj{\prec}{G}$ the zone $A\cap B$, i.e., the submatrix of $\adj{\prec}{G}$ obtained by taking  
rows indexed by $A$ and columns indexed by $B$,  the 1 entry at $(u,w)$ imposes 1 at  all the preceding entries in the same row. 
Since each zone $A_i\cap B_j$ is mixed, we infer that there exists a row containing $10$ as  
two consecutive entries.   

For each part $P\in \mathcal A \cup \mathcal B$, let ${\sf start}(P)$ be the minimal interval on the real line which covers the left endpoints $\{\ell_v:v\in P\}$, 
and we call ${\sf start}(P)$ the \emph{start interval} for $P$. Furthermore, for $P,P'\in \mathcal A \cup \mathcal B$ we write ${\sf start}(P)\leq {\sf start}(P')$ if $p\leq p'$ for every two points $p\in P$ and $p'\in P'$ (and equivalently for $<$). 
While it holds that ${\sf start}(B_1)\leq  \cdots \leq {\sf start}(B_j)  \leq \cdots \leq {\sf start}(B_{4t^2})$ due to the lexicographic order on $V(G)$, 
two consecutive start intervals share an endpoint. 
To extract a sub-collection with  disjoint start intervals, 
we argue that ${\sf start}(B_j) < {\sf start}(B_{j+2}).$ 
Consider two consecutive (w.r.t $\prec$) vertices $w,w'$ of $ B_{j+1}$ 
such that for some $u\in A$ the entry $(u,w)$ equals 1 and $(u,w')$ equals 0. The existence of such $w$ and $w'$ 
was argued in the previous paragraph. The mismatch of the two entries means that $\ell_w < \ell_{w'}$. 
It remains to observe that ${\sf start}(B_j)\leq \ell_w < \ell_{w'} \leq  {\sf start}(B_{j+2}).$

To see that ${\sf start}(A_i) < {\sf start}(A_{i+2}),$ 
consider the zones $A_i\cap B_{4t^2-2i+1}$ for all $i\in [2t^2]$. For every $i\in [2t^2],$ consider a vertex $a_i\in A_i$ 
for which "$10$" appears as two consecutive entries in the row indexed by $a_i$ within the zone $A_i\cap B_{4t^2-2i+1};$ 
the existence of such $a_i$ and the entries "10" on the row of $a_i$ were argued previously. Note that the entries  
"10" represent that $a_i$ is adjacent with some vertex of $B_{4t^2-2i+1}$ and non-adjacent with another vertex of $B_{4t^2-2i+1}$. 
Therefore, the right endpoint $r_{a_i}$ of the interval for $a_i$ is contained in ${\sf start}(B_{4t^2-2i+1})$. 
As we have ${\sf start}(B_{4t^2-2j+1}) < {\sf start}(B_{4t^2-2i+1})$ for every $1\leq i< j \leq 2t^2,$ we 
conclude $r_{a_j} < r_{a_i}$, which implies  $\ell_{a_i} < \ell_{a_j}$ due to the construction of $\prec$. 
Now, observe that ${\sf start}(A_i)\leq \ell_{a_i} < \ell_{a_{i+1}}\leq {\sf start}(A_{i+2})$, as claimed.

Therefore, after removing vertices associated with every other index, we assume that the collections $\mathcal A$ and $\mathcal B$ 
consist of $t^2$ sets each, and their start intervals satisfy 
$${\sf start}(A_1) < \cdots < {\sf start}(A_i) <\cdots {\sf start}(A_{t^2})\leq {\sf start}(B_1)< \cdots < {\sf start}(B_j) <\cdots {\sf start}(B_{t^2}).$$

\bigskip

\noindent {\bf Finding a transversal pair of half-graphs.} 
The following statement implies that the graph $G$ contains a semi-induced $T_t$ for our setting and finishes the proof of delineation of interval graphs. 
Lemma~\ref{lem:transversalPairInIntervalGraphs} will be reused in the later section on rooted directed path graphs.

\begin{lemma}\label{lem:transversalPairInIntervalGraphs}
	Let $G$ be an interval graph 
	and $\adj{}{G}$ an adjacency matrix of $G$. 
	Let $\mathcal{A}_1=\{A_1^1,\dots,A_{t^2}^1\}$ and $\mathcal{A}_2=\{A_1^2,\dots,A_{t^2}^2\}$ be two classes of pairwise disjoint vertex sets such that the following properties hold.
	\begin{enumerate}
		\item ${\sf start}(A_1^1) < \cdots  {\sf start}(A_{t^2}^1)\leq {\sf start}(A_1^2)< \cdots < {\sf start}(A^2_{t^2})$,
		\item The submatrix of $\adj{}{G}$ induced by rows $A_{i}^1$ and columns $A_{j}^2$ has rank at least $2$,
	\end{enumerate}
	where the starting interval ${\sf start}(A_i^j)$ is define with respect to a minimal interval representation of $G$.  
	Then $G$ contains a semi-induced transversal pair $T_t$.
\end{lemma}
\begin{proof} 
First note that ${\sf start}(A_1^1) < \cdots  {\sf start}(A_{t^2}^1)\leq {\sf start}(A_1^2)< \cdots < {\sf start}(A^2_{t^2})$ implies the following property
\begin{enumerate}
	\item[(P1)] If for some $b\in  A_i^1$ and $a\in \bigcup_{k\in [t^2]} A_k^2$ the entry $(b,a)$ of $\adj{}{G}$ is $1$ then the entry $(b,a')$ is also $1$ for every $a' \in A_j^2$ for $j<k$.
	\item[(P2)] If for some $b\in  A_i^1$ and $a\in \bigcup_{k\in [t^2]} A_k^2$ the entry $(b,a)$ of $\adj{}{G}$ is $0$ then the entry $(b,a')$ is also $0$ for every $a' \in A_j^2$ for $j>k$.
\end{enumerate}  

For each $i,j\in [t]$, choose a vertex $b_{i,j}\in A_{(j-1)t+i}^1$ and vertex $a_{i,j}\in A_{(i-1)t+j}^2$ such that the entry $(b_{i,j}, a_{i,j})$ 
equals 1 and the entry $(b_{i,j},a)$ equals $0$ for a vertex $a$  indexing a neighboring column of $a_{i,j}$.  
Note that this choice is possible  because each zone $A_i^1\cap A_{j}^2$ must contain $10$ or $01$ as consecutive entries (follows from the rank $2$ assumption). Since for $i,j\in [t]$ the vertex $a_{i,j}\in A_{(i-1)t+j}^2$  the occurrence of $01$ or $10$ in the row $b_{i',j'}$ and columns from $ A_{(i'-1)t+j'}^2$  implies that the interval $I_{b_{i',j'}}$ intersects interval $I_{a_{i,j}}$  precisely when $(i,j)\leqslant_{\text{lex}}(i',j')$ by property (P1) and (P2).

To choose set $C$ consider the left endpoints $\ell_{b_{i,j}}$ for all $i,j\in [t].$ Note that by the disjointness of the start intervals 
$\ell_{b_{i,j}}<\ell_{b_{i',j'}}$ if and only if $(j,i)\leqslant_{\text{lex}} (j',i')$ as $b_{i,j}\in A_{(j-1)t+i}$ and $b_{i',j'}\in A_{(j'-1)t+i'}$ 
and $(j-1)t+i\leq (j'-1)t+i'$ if and only if $(j,i)\leqslant_{\text{lex}} (j',i')$.
By the minimality assumption on the representation $\I$, 
there exists a vertex $c_{i,j}$ for each $i,j\in [t]$ which is adjacent with $b_{i,j}$ and not adjacent with every $b$ for which $\ell_b>\ell_{b_{i,j}}$  
that is, $\ell_{b_{i,j}}\leq r_{c_{i,j}} < \ell_{b}.$ Hence $b_{i,j}$ is adjacent to $c_{i',j'}$ if and only if $(j,i)\leqslant_{\text{lex}} (j',i')$.

We obtain that for $A := \{a_{i,j}~:~i, j \in [t]\}$, $B := \{b_{i,j}~:~i, j \in [t]\}$, and $C := \{c_{i,j}~:~i, j \in [t]\}$ the triple $(A,B,C)$ is a transversal pair of half-graphs $T_t$. For an illustration see \cref{fig:findingTransversalPairGraph}.
\end{proof}
\begin{figure}[h!]
	\centering
	\begin{tikzpicture}[scale=1.1,point/.style={circle,inner sep=0.06cm, fill}]
	\def\dist{0.35}
	\begin{scope}[xshift=0 cm,yshift=1 cm,rotate=270]
	\node (P-start) at (0.7,10.3) {};
	\node[] (P-B) at (0.7,-1.5) {};
	\draw[lightgray] (0.7,9)--(-9.5*\dist,9);
	\foreach \i in {0,...,8}{
		\begin{scope}
		\clip[postaction={fill=white,fill opacity=0.2}] (0.7,\i+0.5)--(-9.5*\dist,\i+0.5)--(-9.5*\dist,\i+1)--(0.7,\i+1)--(0.7,\i+0.5);		
		\foreach \x in {-10.1,-10,...,12.2}%
		\draw[lightgray!50](\x, -5+\i)--+(12,14.4+\i);
		\end{scope}
		\draw[lightgray] (0.7,\i)--(-9.5*\dist,\i);
		\draw[lightgray] (0.7,\i+0.5)--(-9.5*\dist,\i+0.5);}
	
	\foreach \i/\ash/\bsh in {0/0.1/0.25,1/0.05/0.4,2/0.4/0.3,3/0.3/0.15,4/0.35/0.4,5/0.2/0.3,6/0.2/0.1,7/0.45/0.3,8/0.3/0.4}{
		\pgfmathsetmacro{\even}{mod(\i,2)}
		\pgfmathsetmacro{\Ireg}{4.5+ 0.5*(mod((8-\i),3)*3+floor((8-\i)/3))}
		\draw[-|,red,thick] (-\dist*\i+0.17, 9.7)--(-\dist*\i+0.17, \Ireg+\ash);
		\draw[|-|,color=black!20!blue,thick] (-\dist*\i, \Ireg+\ash+0.025)--(-\dist*\i,0.5*\i+\bsh-0.025);
		\draw[|-,olive,thick] (-\dist*\i+0.17,0.5*\i+\bsh)--(-\dist*\i+0.17,-0.8);
	}
	
	\draw[<->, thick] (P-start) -- (P-B) node  {};
	\node[red] at (1,9.3) {$A$};
	\node[color=black!20!blue] at (1,4.4) {$B$};
	\node[olive] at (1,-0.5) {$C$};
	
	\end{scope}	
	\end{tikzpicture}
\caption{Schematic illustration of the interval representation of a semi-induced transversal pair. }
\label{fig:findingTransversalPairGraph}
\end{figure}
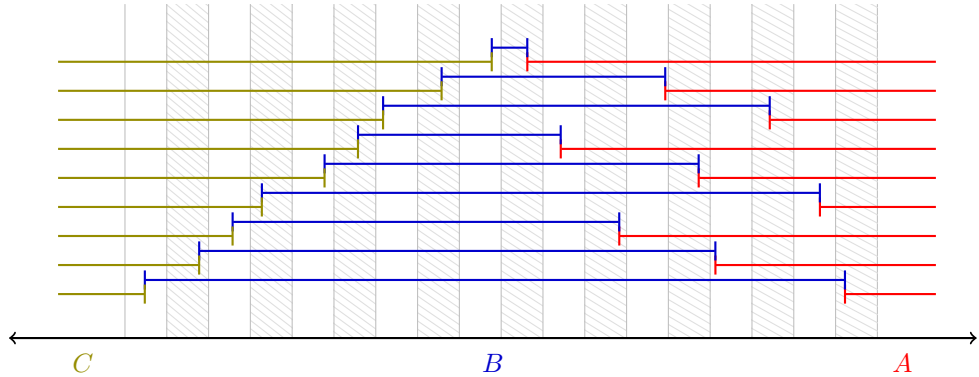	

\subsection{Rooted directed path graphs}\label{subsec:rdpg}


A \emph{tree model} of a graph $G$ consists of a tree $T$ together with a collection  $\{P_v : v \in V(G)\}$ of subtrees of $T$ such that $(u,w) \in E(G)$ if and only if $V(P_u) \cap V(P_w) \neq \emptyset$.
Gavril~\cite{Gavril.74,Gavril.78} showed that a graph $G$ is chordal if and only if $G$ can be represented by a tree model, and that a tree model of $G$ can be constructed in polynomial time.
To avoid ambiguities, we refer to $V(T)$ as the \emph{nodes of $T$}.

If every $P_v$ in a tree model of a graph $G$ is a path, we say that $G$ is an \emph{undirected path graph}.
If $T$ is oriented and every $P_v$ is a directed path, we say that $G$ is a \emph{directed path graph}.
Finally, a directed path graph with $T$ being a rooted tree in which every node is reachable from the root is called a \emph{rooted directed path graph}.
Notice that interval graphs form a subclass of rooted directed path graphs for which $T$ is a directed path. 

In the remainder of the section, we prove the following result.

\begin{theorem}\label{thm:delineation}
  The class of all rooted directed path graphs is effectively delineated.
\end{theorem}

Directed path graphs are \textsl{not} delineated.
Given a bipartite graph $G$ with parts $A,B$, by subdividing each edge of $G$ and making a clique of the set $X$ of newly added vertices, we create a directed path graph $\Pi(G)$ that encodes $G$.
A tree model of $\Pi(G)$ is easy to construct.
We start with a star $S$ with central vertex $s$ and leaves $A \cup B$, orient edges touching $A$ towards $s$, and edges touching $B$ away from $s$.
For $u \in A \cup B$ we let $P_u$ be the trivial path containing only $u$ in $S$.
Then for each edge $ab \in E(G)$ we associate the path in $S$ from $a$ to $b$ with the vertex used to subdivide the edge in $G$.
Then for each $x \in X$ with neighbors $a,b$ in $\Pi(G)$ we let $P_x$ be the path from $a$ to $b$ in $S$.
Now $(S, \{P_u : u \in A \cup B\} \cup \{P_x : x \in X\})$ is a tree model of $\Pi(G)$.
See Figure~\ref{fig:directed_path_graphs_not_delineated} for an example of this construction.

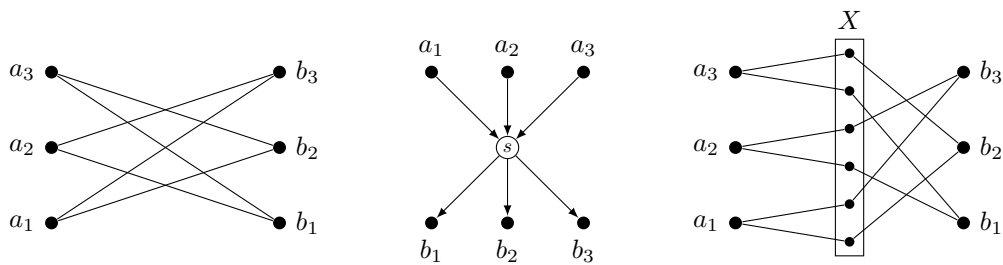
\begin{figure}[h]
\centering
\begin{tikzpicture}[point/.style={circle,inner sep=0.06cm, fill}, vertex/.style={circle, inner sep=0.06cm}]
\begin{scope}
\foreach \i/\j in {1/0,2/1,3/2}{
  \node[point, label = 180:{${a_\i}$}] (a\i) at (0, \i) {};
  \node[point, label = 0:{${b_\i}$}] (b\i) at (3, \i) {};
}

\foreach \i/\j in {1/2,1/3,2/1,2/3, 3/1, 3/2}{
  \draw (a\i) -- (b\j);
}
\end{scope}

\begin{scope}[xshift=9cm]
\foreach \i/\j in {1/0,2/1,3/2}{
  \node[point, label = 180:{$a_\i$}] (a\i) at (0, \i) {};
  \node[point, label = 0:{$b_\i$}] (b\i) at (3, \i) {};
}
\foreach \i/\j in {1/0.75, 2/1.25, 4/1.75, 5/2.25, 7/2.75, 8/3.25}{
  \node[point, scale=.75] (x\i) at (1.5, \j) {};
}


\draw (a1) -- (x1);
\draw (a1) -- (x2);

\draw (x1) -- (b2);
\draw (x2) -- (b3);

\draw (x4) -- (b1);
\draw (x5) -- (b3);
\draw (x7) -- (b1);
\draw (x8) -- (b2);

\draw (a2) -- (x4);
\draw (a2) -- (x5);
\draw (a3) -- (x7);
\draw (a3) -- (x8);

\node[draw, fit = (x1) (x8), label=${X}$] {};
\end{scope}%

\begin{scope}[xshift=4cm, yshift=1cm]
\node[draw,vertex, scale=.75] (x) at (2,1) {$s$};
\foreach \i in {1,2,3}{
  \node[point, label={$a_\i$}] (a\i) at (\i,2) {};
  \node[point, label=-90:{$b_\i$}] (b\i) at (\i,0) {};
  \draw[-latex] (a\i) -- (x);
  \draw[-latex] (x) -- (b\i);
}
\end{scope}
\end{tikzpicture}%
\caption{From left to right: a subcubic bipartite graph $G$, the star $S$, and the directed path graph $\Pi(G)$ (the edges in $X$ are omitted).}
\label{fig:directed_path_graphs_not_delineated}
\end{figure}

The class $\{\Pi(G) : G \in \mathcal{B}_{\leq 3}\}$ is transduction equivalent to $\mathcal{B}_{\leq 3}$ and thus by \cref{lem:not-delineated} we conclude that directed path graphs are not delineated.

We remark that this construction generates graphs that are modeled on trees with many roots (nodes with in-degree zero).
Thus delineation for any subclass of directed path graphs can be expected only if such a class can be modeled by trees with few roots.
We show that this is indeed the case for rooted directed path graphs.

In this section we adopt standard notation related to trees to describe relationships between nodes of a tree $T$.
Let $p,p'\in V(T)$.
We denote by $\leq_T$ the tree order on the nodes of $T$.
If there is a directed path from $p$ to $p'$ in $T$, we say that $p$ is an \emph{ancestor} of $p'$, that $p'$ is a \emph{descendant} of $p$, and write $p \leq_T p'$.
If $p \neq p'$ we also say that $p$ (resp.~$p'$) is a \emph{proper ancestor} (resp.~\emph{proper descendant}) of $p'$ (resp.~$p$) and write $p <_T p'$.
If $p <_T p'$ and there is an edge from $p$ to $p'$ in $T$, then we say that $p'$ is a \emph{child} of $p$ and that $p$ is the \emph{parent} of $p'$.
If $p$ and $p'$ share the same parent, then we say that they are \emph{siblings}.
The \emph{least common ancestor} of $p$ and $p'$ is the node $s \in V(T)$ that is furthest from the root of $T$ and such that $s \leq_T p$ and $s \leq_T p'$.

Given a tree model $(T, \{P_v : v \in V(G)\})$ of a graph $G$, we denote by $V_p$ the set $\{u \in V(G) : p \in V(P_u)\}$, and we say that the model is \emph{minimal} if there are no adjacent $p,p'\in V(T)$ such that $V_p \subseteq V_{p'}$.
It is known that one can always generate a minimal tree model of rooted directed path graph $G$ whenever a tree model of $G$ is given~\cite{Gavril.78}.

For the remaining of this section, and unless stated otherwise, we assume that $G$ is a rooted directed path and that $(T, \{P_v : v \in V(G)\})$ is a minimal tree model of $G$, where every $P_v$ is a directed path of the out-tree $T$. 
For $v \in V(G)$ we denote by $\high(v)$ and $\low(v)$ the nodes of $P_v$ that are closest and furthest from the root, respectively.
Notice they are not necessarily distinct.
We extend this notation to sets of vertices by defining $\high(X) = \{\high(v) : v \in X\}$ and $\low(X) = \{\low(v) : v \in X\}$ for $X \subseteq V(G)$.
The minimality of the tree model implies that, for every node $p \in V(T)$ that is not the root of $T$, there are vertices of $G$ distinguishing $p$ from its parent $p'$. 
\begin{lemma}\label{lem:tree_model_starting_paths}
Let $(T, \{P_v : v \in V(G)\})$ be a minimal tree model of a rooted directed path graph $G$. 
Then for every $p \in V(T)$ that is not the root of $T$,
\begin{romanenumerate}
\item there is  $u \in V(G)$ such that $p = \high(u)$; and 
\item there is  $w \in V(G)$ such that $V(P_w)$ contains the parent $p'$ of $p$, but not $p$.
\end{romanenumerate}
\end{lemma}
\begin{proof}
Let $p$ be a node of $T$ other than the root, and $p'$ be the in-neighbor of $p$ in $T$.
Clearly both $V_p$ and $V_p'$ are not empty since the tree model is minimal.
Now, if no path $P_v$ starts in $p$, then every $P_v$ containing $p$ also contains $p'$ and thus $V_p \subseteq V_{p'}$, contradicting the minimality of the tree model.
Thus $(i)$ follows.

Similarly, if every path $P_v$ containing $p'$ also contains $p$, then $V_{p'} \subseteq V_P$, contradicting the minimality of the tree model.
Thus there is $P_w$ containing $p'$ but not $p$ and the result follows.
\end{proof}

For $q \in V(T)$, let $T_q$ be the maximal subtree of $T$ rooted at $q$.
From the minimal tree model of $G$ we obtain a total order $\prec$ of its vertex set by a lex-DFS style search on $T$ which we describe in the following. 
We traverse  $T$ with depth-first search and \emph{process} a node $q$ when every node in $V(T_q)$ other than $q$ has been visited by the search.
We define a total order $\prec_{\operatorname{DFS}}$ on $V(T)$ which orders the nodes of $T$ by the time in the depth-first search they were processed, \ie{}, $x\prec_{\operatorname{DFS}}y$ if there is a point in the depth-first search when  $T_x$ has been fully explored but there is a child $z$ of $y$ such that $T_z$ has not yet been explored. In this depth-first search we additionally use the following tie-breaking rule. If $x,y$ are siblings then we explore $T_x$ before $T_y$ if the symmetric difference of the two sets $\{\high(v): \low(v)\in V(T_x)\}$ and  $\{\high(v):\low(v)\in V(T_y)\}$ has a minimum $\high(v)$ with respect to $\leq_{T}$ and $\low(v)\in V(T_x)$ (in case the symmetric difference of $\{\high(v): \low(v)\in V(T_x)\}$ and  $\{\high(v):\low(v)\in V(T_y)\}$ has no minimum we choose arbitrarily). From the order $\prec_{\operatorname{DFS}}$ we obtain $\prec$ using the following rules. $u\prec v$ if 
\begin{compactitem}
	\item $\low(u)\prec_{\operatorname{DFS}} \low(v)$ or
	\item $\low(u)=\low(v)$ and $\high(u)<_T \high(v)$.
\end{compactitem}
In case $\low(u)=\low(v)$ and $\high(u)=\high(v)$ we choose an arbitrary order of $u$ and $v$. The order $\prec$ has the following basic properties.
By following a DFS on $T$, using the aforementioned tie-breaking rule to decide which direction to follow when searching for a new leaf, it is not hard to see that $\prec$ can be generated in polynomial time.
Each time a node $q$ is processed, we add to a stack $S$ every $v \in V(G)$ with $\low(v) = p$, following the tie-breaking rules for this decision when needed.
In the end the order that the nodes appear in $S$ is the reverse order of $\prec$.
From $\leq_T$ and $\prec$ we immediately obtain the following.
\begin{observation}\label{obs:nozigzag1}
	Let $x$ be an internal node of $T$ and $y,z$ be two distinct children of $x$.
	For any vertices $v^1_y, v^2_y, v_z\in V(G)$, if $\low(v^i_y)\geq_T y$ and $\low(v_z)\geq_T z$ for $i=1,2$,
	then either $v^i_y \prec v_z$ for both $i=1,2$ or $v^i_y \succ v_z$ for both $i=1,2.$
\end{observation}

\begin{observation}\label{obs:nozigzag2}
	Let $u,v,w\in V(G)$ with $u\prec v \prec w.$ 
	Then the least common ancestor of $\low(u)$ and $\low(w)$ is an ancestor of the least common ancestor of $\low(v)$ and $\low(w).$
	Likewise, it is also an ancestor of the least common ancestor of $\low(u)$ and $\low(v).$
\end{observation}
\begin{proof}
	Suppose not, i.e., the former is a strict descendant of the latter. Then by the construction of $\prec,$ $w$ is discovered before starting 
	to explore the subtree rooted at the least common ancestor of $\low(v)$ and $\low(w)$ containing $v$; a contradiction. One can also simply apply~\cref{obs:nozigzag1} to suitable nodes. 
\end{proof}

\textbf{Summary of the upcoming proof.} We now show that rooted directed paths are effectively delineated.
We start by splitting the vertices of $G$ into two collections associated with rows and columns of a rank-$f(t)$ division of $\adj{\prec}{G}$: we let $\mathcal{A} = \{A_i : i \in f(t)/2\}$ be the first $f(t)/2$ parts of the rows partition and $\mathcal{B} = \{B_i : i \in f(t)/2\}$ to be the last $f(t)/2$ parts of the columns partition.
Then, for each $A_i$ and $B_i$ we take vertices $a_i,b_i$ to represent the sets, respectively, define $A^o$ to contain all $a_i$ and $B^o$ to contain all $b_i$.
The goal is to use $A^o$ and $B^o$ to distinguish between two cases in the proof: one where we can organize the relevant vertices appearing in sets in $\mathcal{A}$ and $\mathcal{B}$ (that is, the vertices defining adjacency in zones $A_i \cap B_j$) in such way that a large interval graph appears and then apply \cref{lem:transversalPairInIntervalGraphs}, and one where we have a good interval-like organization for one of $\mathcal{A}, \mathcal{B}$ but not for the other.
In the latter, we show that we either obtain a semi-induced $T_{t,2}$ or fall back to the first case.
Obtaining the desired $T_{t,2}$ in this case is far more challenging than finding a transversal pair in the interval case.

The first thing we observe is that there is a path $P$ of $T$ containing all $\high(u)$ where $u$ is a vertex defining some adjacency between sets of $\mathcal{A}$ and $\mathcal{B}$ and observe that $P$ defines an order $<_P$ on both $A^o$ and $B^o$.
We denote by $p(u)$ the node in $V(P_u) \cap V(P)$ that is closer to $\low(u)$ and say that $u \leq_P v$ if and only if $p(u) \leq_T p(v)$.
From this point, we prove a series of claims to show that, to organize large parts of $\mathcal{A}$ and $\mathcal{B}$ in a desirable way, we can focus on organizing large parts of $A^o$ and $B^o$.

The easier case is when both $A^o$ and $B^o$ contain sufficiently large strictly increasing chains with respect to $<_P$.
Since $<_P$ may not agree with $\prec$, we apply Erd\H{o}s-Szekeres' theorem to extract a large monotone sequence of both chains, and keep only the vertices appearing in those sequences in $A^o$ and $B^o$.
We then use those sequences to define, for each $A_i$ associated with a vertex in the new $A^o$, an exclusive subpath $I_i$ of $P$ that contains $p(a)$ for every $a \in A_i$, and do the same for each $B_i$.
This is done by observing that no $p(a)$ can be ``very far away'' from $p(a_i)$ with respect to the monotone sequence. 
With these subpaths, we construct an interval graph $G_I$ that is then given as input to~\cref{lem:transversalPairInIntervalGraphs} to solve this case.

If only $A^o$, for instance, contains a large strictly increasing chain with respect to $<_P$ then there must be a node $p \in P$ concentrating a large set of $\{p(b_i) : b_i \in B^o\}$.
Although we can, to some extent, predict the behavior of the paths associated with vertices in $B^o$ after they leave $P$ through $p$, we cannot apply \cref{lem:tree_model_starting_paths} to find in $G$ vertices distinguishing each of the parts of $\mathcal{B}$ associated with vertices of $B^o$.
This is the crucial difference that makes finding a semi-induced $T_{t,2}$ in this configuration much harder than in the first one.

\begin{theorem}\label{thm:rooted_directed_path_graphs_delineation}
Let $g_1(t) = (2t^2)^{2t^2}+2$, $g_2(t) = (2t^2)^2+1$, and $g(t) = 4(g_1(t))^{g_2(t)}$.
For every natural $t$, if $G$ is a rooted directed path graph and the grid rank of $\adj{\prec}{G}$ is at least $2 \cdot g(t)$ then $G$ has a semi-induced $T_{t}$ or $T_{t,2}$.
\end{theorem}
\begin{proof}
  Let $(T, \{P_v : v \in V(G)\})$ be a minimal tree model of $G$, where each $P_v$ is a path, and $M = \adj{\prec}{G}$.
  Notice that each row and each column of $M$ is associated with a vertex of $G$ and so we interchangeably refer to vertices of $G$ as rows and/or columns of $M$, and vice-versa.

  Assume that $M$ contains a rank-$(2 \cdot g(t))$ division $(\mathcal{D}^R, \mathcal{D}^C)$ with $\mathcal{D}^R = \{R_1, \ldots, R_{2g(t)}\}$ and $\mathcal{D}^C = \{C_1, \ldots, C_{2g(t)}\}$.
  Now, let $A_1, \ldots, A_{g(t)}$ be the first $g(t)$ parts of $\mathcal{D}^R$ and $B_1, \ldots, B_{g(t)}$ be the last $g(t)$ parts of $\mathcal{D}^C$.
  By definition of rank divisions, for every $i,j \in [g(t)]$ there are $u \in A_i$ and $v \in B_j$ such that the vertices $u$ and $v$ are adjacent in $G$.
  We denote by $A'_i$ the set of vertices of $A_i$ having a neighbor in some $B_j$, and by $B'_i$ the set of vertices of $B_i$ having a neighbor in some $A_j$.

	\begin{claim}\label{claim:path_P_hitting_all_high_nodes}
  Let $X = \bigcup_{i \in [g(t)]} A'_i \cup B'_i$.
	There is a directed path $P$ of $T$ from the root of $T$ to a leaf of $T$ such that $V(P) \supseteq \high(X)$.
	\end{claim}
	\begin{claimproof}
    Let $P$ be a directed path of $T$ maximizing $|V(P) \cap \high(X)|$ from the root of $T$ to one of its leaves.
    By contradiction, assume that there is a $v \in X$ such that $\high(v) \not \in V(P)$.
    By increasing $P$ if necessary, we can assume that $V(P)$ contains the root of $T$.
    Clearly the root of $T$ is a proper ancestor of $\high(v)$ 
    Let $q \in V(P)$ be the proper ancestor of $\high(v)$ that is closer to $\high(v)$ in $T$.
    If $q$ has no children in $V(P)$, $P$ then we have a contradiction since we can extend $P$ to contain $\high(v)$, and thus we assume that this is not the case.
    Now let $P'$ be the subpath of $P$ from the child $q'$ of $q$ in $V(P)$ to the leaf of $T$ that is contained in $V(P)$.
    Notice that there must be $u \in X$ with $q' <_T \high(u)$ since otherwise we can extend $P$ to contain $\high(v)$, contradicting our choice of $P$.

    If $v \in A'_i$ for some $i \in [g(t)]$ then, by the choice of such sets, $v$ has a neighbor $w \in B'_j$ for some $j \in [g(t)]$.
    Similarly, if $v \in B'_j$ then it has a neighbor $w$ in some $A'_i$.
    In both cases, $q <_T \low(w)$ since $q <_T \low(v)$, and thus there are $a \in A'_i$ and $b \in B'_j$ with $q <_T \low(a)$ and $q <_T \low(b)$ (with $v = a$ or $v = b$).
    Repeating this analysis with relation to $u$ and $q'$, we conclude that there are $a'\in A'_{i'}$ and $b' \in B'_{j'}$ satisfying $q' <_T \low(a')$ and $q' <_T \low(b')$ (with $u = a'$ or $u = b'$).
    This contradicts \cref{obs:nozigzag1}, we conclude that no $v \in X$ with $\high(v) \not \in V(P)$ can exists, and the claim follows.
	\end{claimproof}
	
  For the remaining of this proof, we fix $P$ to be a maximal path of $T$ containing $\high(X)$ for $X = \bigcup_{i \in [g(t)]} A'_i \cup B'_i$.
  For every vertex $v \in X$, we denote by $p(v)$ the node in $V(P_v) \cap V(P)$ that is closer to $\low(v)$ in $T$.
  We extend this notation to sets by writing $p(Y)$ to denote the set $\{p(y) : y \in Y\}$.
  Note that $\leq_T$ defines a canonical linear quasi-order $\leq_P$ on $X$. 
  Namely, for $x,x'\in X$ we have $x\leq_P x'$ if and only if $p(x)\leq_T p(x')$. Clearly it is a quasi-order (reflexive and transitive). 
  Note that $x \leq_P x'$ and $x\geq_P x'$ if and only if $p(x)=p(x')$. If $x >_P x'$, $p(x)$ is a strict descendant of $p(x')$.
  For each $i \in [g(t)]$ we fix $a_i \in A'_i$ such that $a_i \leq_P a$ for every $a \in A'_i$, and $b_i \in B'_i$ such that $b_i \leq_P b$ for every $b \in B'_j$.
  Let $A^o=\{a_i: i\in [g(t)]\}$ and $B^o=\{b_i: i\in [g(t)]\}$.
  There are three possible cases concerning $\leq_P$ and the vertices in $A^o$ and $B^o$.
  
%

\begin{enumerate}
  \item[Case A.] There is a strictly increasing chain of size at least $g(t)/4$ on both $A^o$ and $B^o$.
  \item[Case B.] There is a strictly increasing chain size at least $g(t)/4$ on $A^o$ but not on $B^o$ (or vice-versa, the cases are symmetric).
  \item[Case C.] There is no strictly increasing chain of size $g(t)/4$ on neither $A^o$ nor $B^o$.
\end{enumerate}

If Case A occurs, we show how to find a transversal pair.
If Case B occurs, we show that we either find a semi-induced $T_{t,2}$ or fall back to Case A (thus finding a transversal pair).
Finally, we show that Case C cannot occur without bounding from above the rank of at least one zone $A_i \cap B_j$ of $M$, thus contradicting our assumption on $(\mathcal{D}^R, \mathcal{D}^C)$.
We use the following claims.
%
\begin{claim}\label{lem:territory}
  Let $k$ be a natural.
  For $i \in [k]$,
  \begin{enumerate}[(i)]
    \item if $a_1 <_P \cdots <_P a_{k+1}$ then $a <_P a_{i+2}$ for every $a \in A'_i$.
    \item if $b_1 <_P \cdots <_P b_{k+1}$ then $b <_P b_{i+2}$ for every $b \in B'_i$.
  \end{enumerate}
\end{claim}
\begin{claimproof}
  Suppose item (i) does not hold, i.e., $p(a_{i+2})\leq_P p(a)$. 
	Applying~\cref{obs:nozigzag2} to the three vertices $a\prec a_{i+1}\prec a_{i+2}$ yields the result immediately.
  The proof of (ii) follows similarly.
\end{claimproof}
	\begin{claim}\label{lem:unranked}
	Suppose that there exists some node $p$ on $P$ such that $p <_T \low(a_i)$ for every $i\in [0,t+1].$ 
	Then for every $i\in [t]$ and $a\in A_i$, 
	we have $p \leq_T \low(a)$.
\end{claim}
\begin{claimproof}
	As $a_0\prec a \prec a_{t+1}$ for every $a\in A_i$, $i\in [t]$ the statement follows directly from~\cref{obs:nozigzag2}.
\end{claimproof}

\begin{claim}\label{claim:only_two_sharing_low}
  For every node $s \in V(T)$, there are at most two distinct $u,w \in A^o \cup B^o$ with $\low(u) = \low(w) = s$.
  \end{claim}
  \begin{claimproof}
  By contradiction assume that $\low(u) = \low(v) = \low(w) = s$ with $u,v,w \in A^o \cup B^o$.
  Without loss of generality, we can assume that $u \prec v \prec w$ since all three are in distinct parts of the rank division $(\mathcal{D}^R, \mathcal{D}^C)$.
  By the choice of $\prec$, the set $Y = \{y \in V(G) : \low(y) = s\}$ appears as a consecutive sequence in $\prec$, and $Y$ forms a clique of $G$.
  Since $u \prec v \prec w$, $Y$ contains all vertices associated with rows and columns of at least one zone of $(\mathcal{D}^R, \mathcal{D}^C)$, a contradiction since such a zone would be constant, and the result follows.
\end{claimproof}

\begin{claim}\label{claim:p_concentrating_is_last}
If there is a node $p \in V(P)$ satisfying $\abs{\{b \in B^o : p(b) = p\}} \geq 5$, then there is no other node $p' \in V(P)$ with $p' >_T p$ and such that $p' = p(a_i)$ for some $a_i \in A^o$.
\end{claim}
\begin{claimproof}
Assume that the statement of the claim does not hold and that this is witnessed by some pair $p,p'$ with $p' >_T p$.
By \cref{claim:only_two_sharing_low} at most two distinct $b,b' \in B^o$ satisfy $\low(b) = \low(b') = p$.
Hence there are distinct $b_i, b_\ell, b_j \in B^o$ with $\low(b_i), \low(b_\ell), \low(b_j) \not \in  V(P)$, with $i,\ell,j \in [g(t)]$.
Without loss of generality we can assume that $b_i \prec b_\ell \prec b_j$.
Thus for every $a \in A'_i$ the path $P_a$ contains $p$ if and only if $a$ has a neighbor in $B_\ell$.
Under this configuration, the zone $A_i \cap B_\ell$ can have at least two distinct types of rows, depending on the relationship of each $a \in A_i$ with each $b \in B_\ell$, and the grid rank of the submatrix associated with this zone is at most two.
This contradicts our choice of $(\mathcal{D}^R, \mathcal{D}^C)$ and the result follows.%
\end{claimproof}
It is not hard to see that we can swap the roles of $A^o$ and $B^o$ in the statement of \cref{claim:p_concentrating_is_last}.
That is, the implication also holds if we ask that $p$ concentrates many of the nodes in $p(A^o)$ and $p' = p(b_i)$ for some $b_i \in B^o$.
The crucial observation is that a vertex in the set associated with $p'$ cannot form complex rows (resp.~columns) with vertices appearing in $A'_2$ (resp.~$B'_2$).
Hence we have the following.

\begin{claim}\label{claim:p_concentrating_is_last_swap}
If there is a node $p \in V(P)$ satisfying $\abs{\{a \in A^o : p(a) = p\}} \geq 5$, then there is no other node $p' \in V(P)$ with $p' >_T p$ and such that $p' = p(b_i)$ for some $b_i \in B^o$.
\end{claim}

If Case C occurs there are $p_A$ and $p_B$ satisfying $\abs{\{a \in A^o : \low(a) = p_A\}} \geq 5$ and $\abs{\{b \in B^o : \low(b) = p_B\}} \geq 5$, respectively.
If $p_A = p_B$ then by the choice of $A^o$ and $B^o$ the vertices in $\{v \in A'_i : p(a_i) = p_A\} \cup \{v \in B'_i : p(b_i) = p_B\}$ form a clique in $G$, implying that each zone $A_i \cap B_j$ with $a_i \in A^o$, $b_i \in B^o$, and $p(a_i) \in p_(b_i) = p_A = p_B$ cannot have large rank, a contradiction.
If $p_A \neq p_B$ and $g(t) \geq 5$ then we obtain a contradiction with either \cref{claim:p_concentrating_is_last} or \cref{claim:p_concentrating_is_last_swap}, depending on whether $p_A <_T p_B$ or $p_B < p_A$.
We now focus on solving Case A and Case B.

%
	\bigskip
	\noindent {\bf Case A.} 
  Here we need $g(t)/4 \geq (4t^2+1)^2+1$.
  This case is similar to the interval graph case.
  By discarding some vertices and renaming appropriately, we can assume without loss of generality that $A^o = \{a_1, \ldots , a_{(4t^2+1)^2+1}\}$ and $p(a_i) \neq p(a_j)$ for every $a_i,a_j \in A^o$.
  Likewise, we pick vertices $b_1, \ldots, b_{(4t^2+1)^2+1}$ of $B^o$ satisfying the same property.
	
	After applying the Erd\H{o}s–Szekeres theorem, we 
	assume that the total order on $A^o$ by $\leq_P$ is monotone, i.e.,  
	either $p(a_0) <_T \cdots <_T p(a_{4t^2+2})$ or $p(a_{4t^2+2}) <_T \cdots <_T p(a_0)$ holds. Likewise, either $p(b_0) <_T \cdots <_T p(b_{4t^2+2})$ 
	or $p(b_{4t^2+2}) <_T \cdots <_T p(b_0)$ holds. Note that we discard all $A_i$'s and $B_j$'s for which $p(a_i)$, $p(b_j)$ is not in the decreasing or increasing subsequence obtained by the Erd\"{o}s–Szekeres theorem and then we rename appropriately. 
	
	Let $p(A^o):=\{p(a_i):i\in [4t^2+2]\}$ and $p(B^o):=\{p(b_i):i\in [4t^2+2]\}$.
	We can  further discard half of each of the sets $A_{1}',\dots, A_{4t^2+2}'$ and half of $B_1',\dots,B_{4t^2+2}'$ so that either $p(A^o) <_T p(B^o)$ or $p(B^o) <_T p(A^o)$ holds after discarding the corresponding $a_i$ and $b_j$ from $A^0$ and $B^0$. 
	This can be done as follows. Choose an edge $e_A$ which divides $P$ into two subpaths so that each subpath contains half of $p(A^o)$. 
	Choose an edge $e_B$ similarly. If $e_A\leq_T e_B$, then we discard every $A_i$ for which $p(a_i)\leq_T e_A$ and $B_j$ for which $e_B\leq_T p(b_i)$. We end up in the case $p(A^o) <_T p(B^o)$. Otherwise, we discard in such a way that we end up in the case $p(B^o) <_T p(A^o)$. We assume that after discarding we  rename so that we are left with $A_1',\dots, A_{2t^2+1}'$ and $B_1',\dots, B_{2t^2+1}'$.
	Since $\leq_P$ is monotone on both $A^0$ and $B^0$ we can apply~\cref{lem:territory} to both sets. Hence we can find $t^2$ pairwise node disjoint subpaths $I_1,\dots,I_{t^2}$ of $P$ such that for every $i\in \{2i-1: i\in [t]\}$  we have that $p(a)$ is in $I_{i}$ for every $a\in A_{2i-1}'$. Likewise we choose pairwise node disjoint subpaths $J_1,\dots, J_{t^2}$ containing the $p(b)$'s. Note that the $I_i$'s and $J_j$'s are aligned on $P$ following the total order on $A^o$ and $B^o$. We  discard all sets $A_i'$ and $B_i'$ with $i$ even and rename such that  the remaining blocks are called $A_1',\dots, A_{t^2}',B_1,\dots,B_{t^2}'$. Note that the subgraph $G_I$ of $G$ induced by vertices in $\bigcup_{i \in [t^2]}A_i'\cup B_i'$ is an interval graph. Further note that the interval representation of $G_I$ that can be obtained by looking at $P$ in reverse satisfies that ${\sf start}(A_i')\subseteq I_i$ and ${\sf start}(B_j')\subseteq J_j$. Further note that the zone $A_i'\cup B_j'$  in the adjacency matrix of $G_I$ has rank at least $2$ since we obtain this zone by deleting rows and columns that are constant $0$ from the zone $A_i'\cup B_j'$ in $M$ which has rank at least $3$ by choice. Therefore \cref{lem:transversalPairInIntervalGraphs} applies to $G_I$ which implies that $G$ contains a semi-induced transversal pair $T_t$.

 \bigskip

	\noindent {\bf Case B.} 
  Let $g_1(t) = (2t^2)^{2t^2}+2$, $g_2(t) = (2t^2)^2+1$, and $g_B(t) = (g_1(t))^{g_2(t)}$.
  For this case we ask that $g(t)/4 \geq g_B(t) + 4$ and assume that there exists a node $p \in V(P)$ such that $p=p(b_i)$ for at least $g(t)/4$ distinct vertices $b_i \in B^o$.
  By assumption, we also have a strictly increase chain of size at least $g(t)/4$ on $A^o$.
  By \cref{claim:p_concentrating_is_last} we know that $p(a) <_P p$ for every~$a$ appearing in this chain.
  We repeat the steps done in Case A with regard to $A^o$ and this chain to obtain the exclusive ``territories'' for the sets $p(A'_i)$ in $P$.
  Without loss of generality, we assume that $A^o = \{a_i : i \in [(8t^2)^2+1]\}$
  by discarding and renaming vertices appropriately if needed.
  Then, we invoke the Erd\H{o}s-Szekeres theorem to obtain a monotone subsequence of size $(8t^2 + 1)$ of $A^o$ with regard to $<_P$.
  We now keep in $A^o$ only the vertices appearing in this monotone sequence.
  Finally, we apply \cref{lem:territory} to obtain the pairwise node disjoint subpaths $I_1, \ldots, I_{4t^2}$ of $P$ where $V(I_i)$ contains $p(a)$ for every~$a \in A'_i$ and $i \in [4t^2]$.
  


  \medskip
  \noindent\textbf{Cleaning $B^o$ and organizing the sets $B_i$.}
  By \cref{claim:only_two_sharing_low} at most two distinct $b,b'\in B^o$ satisfy $\low(b) = \low(b') = p$.
  We drop those two vertices from $B^o$ to guarantee that every~$b$ in the set has $\low(b) \neq p$.
  By discarding other vertices of $B^o$ and renaming when necessary, we assume that $B^o = \{b_i : i \in [0,g_B(t)+1]\}$. 
  

  
	Let $i \in [g_B(t)]$. 
  Observe that in this setting $p(b) = p$ and $p <_T \low(b)$ for  every $b \in B'_i$ since otherwise either $b_{g_B(t)+1} \prec b$ or $b \prec b_0$ and both cases contradict the choice of $\prec$. 
  We conclude that $p <_T \low(b)$ and that there is a child $q$ of $p$ with $q \not \in V(P)$ such that $p'\leq_T \low(b)$.
  Thus for every $i \in [g(t)]$ there is a set $Q_i$ of children of $p$, all of which are not in $V(P)$, such that every $b \in B_i$ satisfies $q \leq_T \low(b)$ for some $q \in Q_i$.
  In other words, the path $P_b$ associated with $b$ ends in a subtree of $T$ rooted in a child of $p$ that is not in $V(P)$ whenever $b \in B'_i$.
  

  We say that a node $q$ of $T$ \emph{hosts} a vertex $v$ if $q \leq_T \low(v).$
  As an immediate consequence of~\cref{obs:nozigzag1}, the children of any node $s$ in $T$ admits a total order by $\prec$. 
  That is, for any pair of children $q,q'$ of $s$ we write $q \prec q'$ if every $v$ with $\low(v) \geq_T q$ satisfies $v \prec v'$ whenever $\low(v') \geq_T q'$.
  \cref{obs:nozigzag1} says that any pair of children of $p$ is comparable, and clearly the transitivity and antisymmetry of $\prec$ on $V(G)$ extends to 
  to that on the children of $s$.
  This observation is later used to organize a large subset $Y$ of $B^o$ in such way that each $B'_i$ associated with a $b_i \in Y$ is hosted by an exclusive collection of children of $p$.

  
  Let $q_1, q_2, \ldots, q_\ell$ be the children of $p$ hosting at least one vertex of $B^o$, where the indexing follows $\prec$; that is, $q_1 \prec q_2 \prec \cdots \prec q_\ell$.
  For each $i \in [g_B(t)]$ we denote by $J_i$ the subset of $[\ell]$ such that at least one vertex of $B_i$ is hosted in $q_j$ for some $j \in J_i$.
  The sets $J_i$ may not be well-behaved: they may intersect since some $q_j$ can host vertices from distinct $B_i$ and $B_{i'}$.
  If we can find a collection of sets $J_i$ with $i \in [t]$ that do not overlap, then we can use this configuration, together with the previously defined subpaths $I_i$ of $P$, to find a semi-induced $T_{t,2}$ and conclude Case B.
  See~\cref{fig:Case_B_part_1} for the configuration achieved until this point.
  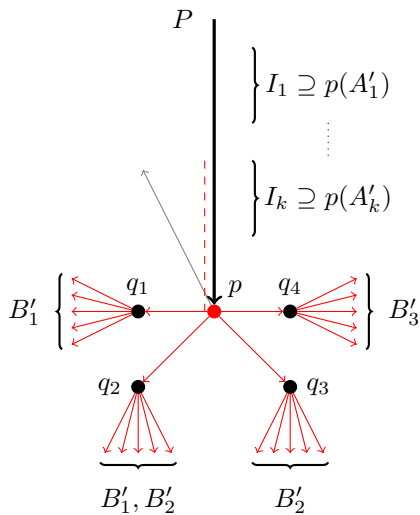
\begin{figure}[h!]
  \centering
  \begin{tikzpicture}[point/.style={circle,inner sep=0.06cm, fill}]
  \centering
  \node (P-start) at (0,4) {};
  \node[draw, point, red, label=45:{$p$}] (P-B) at (0,0) {};
  \node (P-end) at (-1,2) {};
  \draw[->, very thick] (P-start) -- (P-B) node [pos = 0, label=180:{$P$}] {};
  \draw[->, gray] (P-B) -- (P-end) {}; 
  \foreach \i/\j in {0/1,1.5/k}{
    \draw [decorate,thick,decoration={brace,amplitude=2pt},xshift=0pt,yshift=0]
    ($(P-start) + (0.5,-0.5-\i)$) -- ($(P-start) + (0.5,-1.5-\i)$) node
    [midway, xshift = 1cm] {$I_\j \supseteq p(A'_\j)$}  node[midway] (y\j) {};
  }
  \draw[dotted] ($(y1) + (1, -.5)$) -- ($(yk) + (1, .5)$);

  \node[draw,point, label=90:{$q_1$}] (q1) at ($(P-B) + (-1,0)$) {};
  \node[draw,point, label=180:{$q_2$}] (q2) at ($(P-B) + (-1,-1)$) {};
  \node[draw,point, label=0:{$q_3$}] (q3) at ($(P-B) + (1,-1)$) {};
  \node[draw,point, label=90:{$q_4$}] (q4) at ($(P-B) + (1,0)$) {};

  \foreach \i in {1,2,3,4}{
    \draw[->,red] (P-B) -- (q\i);
  }
  \foreach \j in {-2,-1,0,1,2}{
      \node (x1\j) at ($(q1) + (-1, \j/4)$) {};
      \draw[->,red] (q1) -- (x1\j);
  }
  \draw [decorate,thick,decoration={brace,amplitude=2pt},xshift=0pt,yshift=0]
    ($(x1-2) + (0,0)$) -- ($(x12) + (0,0)$) node
    [midway, xshift = -.5cm] {$B'_1$}; 
  \foreach \j in {-2,-1,0,1,2}{
      \node (x4\j) at ($(q4) + (1, \j/4)$) {};
      \draw[->,red] (q4) -- (x4\j);
  }
  \draw [decorate,thick,decoration={brace,amplitude=2pt, mirror},xshift=0pt,yshift=0]
    ($(x4-2) + (0,0)$) -- ($(x42) + (0,0)$) node
    [midway, xshift = .5cm] {$B'_3$}; 

  \foreach \j in {-2,-1,0,1,2}{
      \node (x2\j) at ($(q2) + (\j/4, -1)$) {};
      \draw[->,red] (q2) -- (x2\j);
  }
  \draw [decorate,thick,decoration={brace,amplitude=2pt, mirror},xshift=0pt,yshift=0]
    ($(x2-2) + (0,0)$) -- ($(x22) + (0,0)$) node
    [midway, yshift = -.5cm] {$B'_1, B'_2$}; 

  \foreach \j in {-2,-1,0,1,2}{
      \node (x3\j) at ($(q3) + (\j/4, -1)$) {};
      \draw[->,red] (q3) -- (x3\j);
  }
  \draw [decorate,thick,decoration={brace,amplitude=2pt, mirror},xshift=0pt,yshift=0]
    ($(x3-2) + (0,0)$) -- ($(x32) + (0,0)$) node
    [midway, yshift = -.5cm] {$B'_2$}; 
  \draw[red,dashed] ($(P-start) + (-.125, -2)$) -- ($(P-B) + (-.125, 0)$);
  \end{tikzpicture}%
  \caption{Illustration of the configuration achieved in Case B up to this point. In the example, $k = 2t^2$ and the red edges are denoting the direction the paths associated with vertices in $B'_i$ with $i \in [4]$ are taking after $p$. Here we have $J_1 = \{1,2\}, J_2 = \{2,3\}, J_3 = \{4\}$. The light gray path denotes the continuation of $P$ after $p$.}
  \label{fig:Case_B_part_1}
  \end{figure}
  Next, we show that we either find such a collection or fall back to Case A, finding a transversal pair.

  

	
  \medskip
  \noindent\textbf{Eliminating the overlaps between sets $J_i$.}	
  We begin with two simple observations, formalized by the following claim.
  For $q \in V(T)$ we denote by $J^q \subseteq [g_B(t)]$ be the set of indices $i$ such that there exists $a \in B_i$ that is hosted by $q$.
	
	\begin{claim}\label{lem:unrankedneat}
    For each $q \in \{q_1, \ldots, q_\ell\}$, the following properties hold.
		\begin{enumerate}[$(i)$]
			\item If $q$ hosts some $b \in B_i$ and some $b'\in B_j$ with $i \neq j$ then no other child $q'$ of $p$ does, and
			\item $J^q$ forms an interval of $[g_B(t)]$, i.e., consists of consecutive integers.
		\end{enumerate}
	\end{claim}
	\begin{claimproof}
  To prove $(i)$ by contradiction, assume that $q$ hosts $x, y$ with $x \in B_i$ and $y \in B_j$, and that $q'$ hosts $x',y'$ with $x' \in B_i$ and $y'\in B_j$.
  Without loss of generality we may assume that $j > i$.
  Now either $x \prec x' \prec y$ or $x' \prec x \prec y'$.
  In either case, we obtain a contradiction with \cref{obs:nozigzag2} since $q$ is an ancestor of the least common ancestor of $\low(x)$ and $\low(y)$, $q'$ is an ancestor of the least common ancestor of $\low(x')$ and $\low(y')$, and $p$ is the least common ancestor of both pairs $\low(x),\low(x')$ and $\low(y),\low(y')$.

  We prove $(ii)$ similarly.
  If there are $x \in B_i$, $y \in B_j$, and $z \in B_\ell$ with $i < j < \ell$ such that $q$ hosts $x,z$ and $q'$ hosts $y$, we obtain a contradiction by applying~\cref{obs:nozigzag2} to $x,y$ and $z$ when $q\neq q'  $, and the result follows.
	\end{claimproof}

By the definition of the order $\prec$ on $\{q_1,\ldots, q_{\ell}\}$, each $J_i$ with $i \in [g_B(t)]$ forms an interval (i.e., a set of contiguous integers) of $[\ell].$ 
Moreover, \cref{lem:unrankedneat} implies that for $i < j$, we have $q_{i'} \prec q_{j'}$ whenever $i' \in J_i$ and $j'\in J_j$.
In particular, $J_i$ and $J_j$ can overlap on at most one index.	
Our goal is now to ensure that there is no overlap between any $J_i$ and $J_j$. 
More precisely, for sufficiently large $k$ we want to extract a subset $B^\star$ of $B^o$ with size $k$ and a minimal ordered (with respect to $\prec$) subset $\{q'_1, \ldots, q'_{\ell'}\}$ of the children $\{q_1,\ldots, q_{\ell}\}$ hosting every vertex of $\{v \in B'_j : b_j \in B^\star\}$ so that the following property holds. 
	
\begin{center}
	Property $(\star)$: $J_1,\cdots , J_k$ forms a division of $[\ell']$. That is, each $J_i$ is an interval of $[\ell']$ and they are pairwise disjoint.
\end{center}

Let $\mathcal{B} = \{B'_1, \ldots, B'_{g_B(t)}\}$.
The only obstruction is the possibility that some $q_i$ hosts the vertices of arbitrarily many sets of $\mathcal{B}$, potentially all of them. 
If each $q_i$ hosts at most $g_1(t)$ sets of $\mathcal{B}$ then, assuming that $\abs{\mathcal{B}} \geq g_q(t) \cdot h(t)$, one can greedily extract $g_1(t)-2$ sets of this collection meeting Property $(\star)$ by finding a large independent set of the interval graph associated with the $J_1, \ldots, J_{\abs{\mathcal{B}}}$.
\begin{claim}\label{claim:interval_independent_set}
Let $J \subseteq [g_B(t)]$ and $q$ be a node of $T$  such that every $b_i$ with $i \in J$ is hosted by $q$.
If there is a function $h(t)$ such that every child of $q$ hosts at most $h(t)$ vertices $b_i$ and $\abs{J} \geq g_1(t) \cdot h(t)$, then there are $\{q'_1, \ldots, q'_{\ell'}\}$, $B^\star \subseteq B^o$, and $J_1, \ldots, J_{g_1(t)-2}$ satisfying Property~$(\star)$.
\end{claim}
\begin{claimproof}
By \cref{claim:only_two_sharing_low} we know that there are at most two distinct $i,j \in J$ with $\low(b_i) = \low(b_j) = q$.
Let $J' = J - \{i,j\}$ and  $\{q_1, \ldots, q_\ell\}$ be the set of children of $q$, ordered by $\prec$.
For $i \in J'$ let $J_i$ be the subset of $[\ell]$ such that at least one vertex of $B'_i$ is hosted by $q_j$ for some $j \in \ell$, and consider the interval graph $H$ with vertex set $\{J_i : i \in J'\}$
If each~$q_i$ hosts at most $h(t)$ sets of $\mathcal{B}$ then each index in $J'$ appears in at most $h(t)$ of the intervals forming the vertices of $H$ and thus $\omega(H) \leq h(t)$.
Interval graphs are known to be~\emph{perfect}\footnote{A graph $G$ is \emph{perfect} if $\chi(G') = \omega(G')$ for every induced subgraph $G'$ of $G$.}
Hence, $\chi(H) = \omega(H)$.
Finally, since $\alpha(H) \cdot \omega(H) \geq \abs{V(H)}$ we get that $\alpha(H) \geq g_1(t)-2$, as desired.

We construct the sets satisfying Property $(\star)$ from an independent set $Y$ of $H$ with $|Y| = g_1(t)$ by keeping in $B^\star$ only the vertices $b_j$ with $j$ appearing in some $J_i \in Y$ and choosing the minimal subset of children of $q$ hosting every vertex of $\{v \in B'_j : b_j \in B^\star\}$.
\end{claimproof}
Next, we show how to use those sets to find a semi-induced $T_{t,2}$ and then show that with $h(t) = g_B(t)$ we either obtain the sets satisfying Property~$(\star)$ with $k = g_1(t)-2$ or fall back to Case A.

\medskip
\noindent\textbf{Using Property $(\star)$ to find a semi-induced $T_{t,2}$.} 
In this section we aim to find a semi-induced $T_{k,2}$ $(A,B,C,D,E)$. We set $h'(t)=(2t^2)^{2t^2}$ and observe that we obtained at least $h'(t)$ sets $B_i'$ with property $(\star)$ and at least $4t^2$ different $A_i'$ with distinct regions.
The following claim is a consequence of the tie-breaking rule used to obtain $\sigma$ and it will provide us with sets $D$ and $E$ of the semi-induced $T_{k,2}$.
\begin{claim}\label{claim:consequenceOfTieBrakingRule}
	There exist a set $S\subseteq [h'(t)]$ of size $2t^2$ and a set $\{\tilde{s}_i:i\in S\}$  such that 
	\begin{enumerate}
		\item[(i)]  $\tilde{s}_i$ is hosted by some $q_{j}$ for $j\in J_i$ for every $i\in S$,
		\item[(ii)] $\high(\tilde{s}_i)<_T \high(\tilde{s}_j)$ for $i<j$,
		\item[(iii)] every $q_j$ with $j\in J_i$ hosts some path $s$ with $\high(\tilde{s}_{i})\leq_T \high(s)\leq_{T}\high(\tilde{s}_{i'})$ for $i,i'\in S$, $ i<i'$ and
		\item[(iv)] $\high(\tilde{s}_i)\leq_T p(a)$ for some $a\in A_{2t^2}'$.
	\end{enumerate}    
\end{claim}
\begin{claimproof}
Let $s_i^1$ be the vertex such that $\high(s_i^1)$ is the minimum of the set $\{\high(v): \low(v)\in T_{q_j}, j\in J_i\}$ with respect to $\leq_T$.
Since the set $\{q_j: j\in J_i\}$ hosts all vertices from $B_i'$ and for $i<i'$ the vertices in $B_i'$ get discovered before the vertices in $B_{i'}'$, 
the tie-breaking rule for obtaining $\prec$ implies that $\high(s_1^1)\leq_T \dots\leq_T \high(s_{h'(t)}^1)$. There are two cases to consider. Either there is a 
subset $S\subseteq [h'(t)]$ of size $2t^2$ such that $\high(s_i^1)<_T \high(s_j^1)$ for $i,j\in S$, $i<j$  or there is a node $p^1$ on $P$ such that there is a subset $S^1\subseteq [h'(t)]$ of size at least $h'(t)/(2t^2)$ such that $\high(s_i^1)=p^1$ for every $i\in S^1$.  In the first case we terminate. In the second case we define $s_i^2$ such that $\high(s_i^2)$ is the minimum of the set $\{\high(v): \low(v)\in T_{q_j}, j\in J_i, \}\setminus \{p^1\}$ with respect to $\leq_T$ and repeat the argument iteratively  no more than $2t^2$ times. That is in the $k$-th iteration for every $i\in S^{k-1}$ we select $s_i^k$ such that $\high(s_i^k)$ is the minimum of the set $\{\high(v): \low(v)\in T_{q_j}, j\in J_i, \}\setminus \{p^j: j<k\}$ with respect to $\leq_T$.
Again by the tie-breaking rule we used to obtain $\prec$ we infer that $\high(s_i^{k})\leq_T  \high(s_{i'}^k)$ for every $i,i'\in S^{k-1}$ for which $i<i'$. Again we either find a subset $S\subset [h[t']]$ of size $2t^2$ such that $\high(s_i^{k})<_T\high(s_j^k)$ for $i,j\in S$, $i<j$  or a node $p^k$ and a set $S^k\subseteq S^{k-1}$ of size $h'(t)/(2t^2)^k$ such that $\high(s_i^{k})=p^k$ for every $i\in S^k$.  We repeat this procedure until we obtain either 
\begin{enumerate}
	\item[(C1)] a set $S\subseteq [h'(t)]$ of size $2t^2$ and vertices $\tilde{s}_i$, $i\in S$ such that $\tilde{s}_i$ is hosted by some $q_j$ for $j\in J_i$ and $\high(\tilde{s}_i)<_T\high(\tilde{s}_j)$ for $i,j\in S$, $i<j$  or
	\item[(C2)] $2t^2$ points $p^1,\dots,p^{2t^2}$, a set $S^{2t^2}$ and for every $i\in S^{2t^2}$, $j\in [t]$ a vertex $s_i^j$ such that $\high(s_i^j)=p_j$ for every $i\in S^{2t^2}$.
\end{enumerate} 
Note that each group $q_j$, $j\in J_i$ hosts at least $2t^2$ vertices with different higher endpoints and hence the set $\{\high(v): \low(v)\in T_{q_j}, j\in J_i, \}\setminus \{\high(s_i^j): j<k\}$ from which we choose $s_i^k$ in the $k$-th iteration of the above procedure is never empty.  

In case we obtain case (C1) the properties (i),(ii) are true by construction. To prove property (iii) observe that there is a $k$ such that $\tilde{s}_i=s_i^k$ for every $i\in S$.  Since every branch $T_{q_j}$ for $j\in J_i$ is explored by the lex-DFS before $T_{q_{j'}}$ for $j'\in J_{i'}$ and $i'>i$ the tie-breaking rule implies that  $q_j$ must host some path $s$ with $p^{k'}\high(s)\leq_{T}\high(\tilde{s}_{i'})$ for $k'<k$. Since we further picked $s_i^k$  in the above procedure to be the minimum of $\{\high(v): \low(v)\in T_{q_j}, j\in J_i, \}\setminus \{p^j: j<k\}$ we obtain (iii). Hence assume we obtained case (C2). We now set $S=S^{2t^2}$ and $\tilde{s}_i:= s_i^{2t^2-i}$. Since $s_i^{2t^2-i}$ is hosted by some $q_j$ for $j\in J_i$ and by construction $\high(s_1^t)=p^t<_T\dots <_T \high(s_t^1)=p^1$ these choices also satisfy the conditions of the statements (i) and (ii). Note that in this case, (iii) follows from the observation that if in the $k$-th iteration of the above procedure we find $p^k$ such that $\high(s_i^k)=p^k$ for every $i\in S^k$ 
then the tie-breaking rule implies that every $q_j$ for $j\in \bigcup_{i\in S}J_i$ hosts a path $s$ with $\high(s)=p^k$.

The condition (iv), \ie $\high(\tilde{s}_i)\leq_P p(a)$ for some $a\in A_{2t^2}'$, follows trivially from the tie-breaking rule  as the group $\{q_i:i\in J_{2t^2}\}$ hosts an element $b\in A'_i$ for every $i\in [2t^2]$ by construction. %
\end{claimproof}
Let $S\subseteq [h'(t)]$ and $\tilde{s}_i\in V(G)$, $i\in S$ as in  \cref{claim:consequenceOfTieBrakingRule}. After renaming respectful of the order $\sigma$ we assume that $S=[2t^2]$.  For easier reading we define $A_i'':=A'_{4t^2-(i-1)}$ for every $i\in [2t^2]$ and $I''_i:=I_{4t^2-(i-1)}$. Note that this is just a renaming scheme to keep the indices in the following argument small. For $i,j \in [t]$ pick $b_{i,j}$ to be the vertex that is hosted in $q_k$ for $k\in J_{2(t^2-(j-1)t-(i-1))}$ and $\high(b_{i,j})\in I''_{2((i-1)t+j)-1}$. Now pick $a_{i,j}\in A_{2((i-1)t+j)}'$ arbitrarily. Note that by construction (and considering the renaming of the sets $A_i'$ for $i>2t^2$) $a_{i,j}$ is adjacent to $b_{i',j'}$ if and only if $(i,j)\leqslant_{\text{lex}} (i',j')$.   Set $c_{i,j}$ be a vertex with $\high(c_{i,j})=q$ where $q\in \{q_j:j\in J_j\}$ is the node hosting $b_{i,j}$. Note that $c_{i,j}$ exists by \cref{lem:tree_model_starting_paths}.
Set $d_{i,j}$ to be a vertex hosted in the same branch as $b_{i,j}$  with $\tilde{s}_{2(t^2-(j-1)t-(i-1))}\leq_P \high(d_{i,j})\leq_P \tilde{s}_{2(t^2-(j-1)t-(i-1))+1}$. The existence of $d_{i,j}$ follows from Claim~\ref{claim:consequenceOfTieBrakingRule}(iii) since $b_{i,j}$ and $d_{i,j}$ are hosted by some $q_k$ for $k\in J_{2(t^2-(j-1)t-(i-1))}$. Furthermore $d_{i,j}$ is distinct from $b_{i'j'}$, $i',j'\in [t]$ as $\high(d_{i,j})\leq_P a <_P \high(b_{i',j'})$ for some element $a\in A'_{2t^2}$ by choice of $b_{i,j}$ and Claim~\ref{claim:consequenceOfTieBrakingRule}(iv). Observe that if $(j,i)<_{\text{lex}}(j',i')$ then $\high(d_{i',j'})\leq_P \tilde{s}_{2(t^2-(j'-1)t-(i'-1))+1}<_P  \tilde{s}_{2(t^2-(j-1)t-(i-1))}\leq_P \high(d_{(i,j)})$ where the strict inequality follows from Claim~\ref{claim:consequenceOfTieBrakingRule}(ii).  Hence we can pick $e_{i,j}$ to be a vertex such that $p(e_{i,j})=\high(d_{i,j})$ and obtain that $d_{i,j}$ is adjacent to $e_{i',j'}$ if and only if $(j,i)\leqslant_{\text{lex}}(j',i')$. Note that $e_{i,j}$ exists by \cref{lem:tree_model_starting_paths}. Set $A:= \{a_{i,j}~:~i,j \in [k]\}$, $B := \{b_{i,j}~:~i,j \in [k]\}$,  $C := \{c_{i,j}~:~i,j \in [k]\}$, $D:= \{d_{i,j}~:~i,j \in [k]\}$ and $E := \{e_{i,j}~:~i,j \in [k]\}$ and observe that indeed $(A,B,C,D,E)$ is a semi-induced $T_{k,2}$. For illustration see Figure~\ref{fig:findingTk2}.

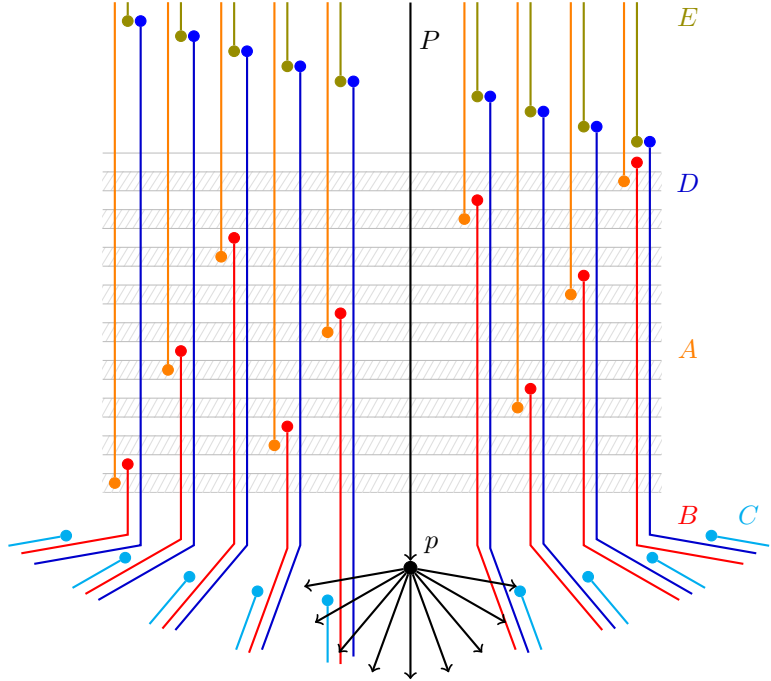
\begin{figure}[h!]
	\centering
	\begin{tikzpicture}[rotate=270,scale=1,point/.style={circle,inner sep=0.06cm, fill}]
	\def\dist{0.7}
	\def\sh{0.125}
	\def\r{1.5}
	\begin{scope}[xshift=0 cm,yshift=1 cm,rotate=270]
	
	\draw[lightgray] (0.7,9)--(-9.5*\dist,9);
	\foreach \i in {9,...,17}{
		\begin{scope}
		\clip[postaction={fill=white,fill opacity=0.2}] (0.7,0.5*\i+0.25)--(-9.5*\dist,0.5*\i+0.25)--(-9.5*\dist,0.5*\i+0.5)--(0.7,0.5*\i+0.5)--(0.7,0.5*\i+0.25);		
		\foreach \x in {-10.1,-10,...,12.2}%
		\draw[lightgray!50](\x, -5+0.5*\i)--+(12,14.4+0.5*\i);
		\end{scope}
		\draw[lightgray] (0.7,0.5*\i)--(-9.5*\dist,0.5*\i);
		\draw[lightgray] (0.7,0.5*\i+0.25)--(-9.5*\dist,0.5*\i+0.25);}
	\begin{scope}[xshift=0.2 cm,yshift=0 cm]
	\foreach \i in {0,...,4}{
		\pgfmathsetmacro{\even}{mod(\i,2)}
		\pgfmathsetmacro{\Ireg}{4.5+ 0.5*(mod((8-\i),3)*3+floor((8-\i)/3))}
		\pgfmathsetmacro{\deg}{10+20*\i}
		\node[inner sep=0pt, minimum width=4pt] [shift={(\dist*\i-0.17, -9.56-0.06*\i)}](red_\i) at (\deg:\r) {};
		\node[inner sep=0pt, minimum width=4pt] [shift={(\dist*\i, -9.7)}](blue_\i) at (\deg:\r) {};
		\node[draw,cyan,circle, fill=cyan, inner sep=0pt, minimum width=4pt] [shift={(\dist*\i-0.34, -9.46-0.08*\i)}](c_\i) at (\deg:0.65) {};
		\node[ inner sep=0pt, minimum width=4pt] [shift={(\dist*\i+0.17, -9.6)}](cend_\i) at (\deg:\r) {};
		\node[draw,red,circle, fill=red, inner sep=0pt, minimum width=4pt](redend_\i) at (-\dist*\i+0.17, \Ireg+\sh) {};
		\node[draw,blue,circle, fill=blue, inner sep=0pt, minimum width=4pt](blueend_\i) at (-\dist*\i,2.65+0.2*\i+\sh-0.025) {};
		\node[draw,olive,circle, fill=olive, inner sep=0pt, minimum width=4pt](oliveend_\i) at (-\dist*\i+0.17,2.65+0.2*\i+\sh-0.025) {};
		\node[draw,orange,circle, fill=orange, inner sep=0pt, minimum width=4pt](orangeend_\i) at (-\dist*\i+0.34, \Ireg+3*\sh) {};
		\node[ inner sep=0pt, minimum width=4pt] [shift={(\dist*\i-0.34, -9.46-0.08*\i)}](cend_\i) at (\deg:\r) {};
		\draw[cyan,thick] (c_\i)--(cend_\i);
		\draw[-,red,thick] (red_\i)--(-\dist*\i+0.17, 9.56+0.06*\i)--(redend_\i);
		\draw[-,color=black!20!blue,thick](blue_\i)-- (-\dist*\i, 9.7)--(blueend_\i);
		\draw[-,olive,thick] (oliveend_\i)--(-\dist*\i+0.17,2.5);
		\draw[-,color=orange,thick] (orangeend_\i)--(-\dist*\i+0.34,2.5);
	}
	\end{scope}
	\begin{scope}[xshift=-0.9 cm,yshift=0 cm]
	\foreach \i in {5,...,8}{
		\pgfmathsetmacro{\even}{mod(\i,2)}
		\pgfmathsetmacro{\Ireg}{4.5+ 0.5*(mod((8-\i),3)*3+floor((8-\i)/3))}
		\pgfmathsetmacro{\deg}{10+20*\i}
		\node[inner sep=0pt, minimum width=4pt] [shift={(\dist*\i-0.17, -9.7)}](red_\i) at (\deg:\r) {};
		\node[inner sep=0pt, minimum width=4pt] [shift={(\dist*\i, -10.04+0.06*\i)}](blue_\i) at (\deg:\r) {};
		\node[draw,cyan,circle, fill=cyan, inner sep=0pt, minimum width=4pt] [shift={(\dist*\i+0.17, -10.1+0.08*\i)}](c_\i) at (\deg:0.65) {};
		\node[draw,red,circle, fill=red, inner sep=0pt, minimum width=4pt](redend_\i) at (-\dist*\i+0.17, \Ireg+\sh) {};
		\node[draw,blue,circle, fill=blue, inner sep=0pt, minimum width=4pt](blueend_\i) at (-\dist*\i,2.65+0.2*\i+\sh-0.025) {};
		\node[draw,olive,circle, fill=olive, inner sep=0pt, minimum width=4pt](oliveend_\i) at (-\dist*\i+0.17,2.65+0.2*\i+\sh-0.025) {};
		\node[draw,orange,circle, fill=orange, inner sep=0pt, minimum width=4pt](orangeend_\i) at (-\dist*\i+0.34, \Ireg+3*\sh) {};
		\node[ inner sep=0pt, minimum width=4pt] [shift={(\dist*\i+0.17, -10.1+0.08*\i)}](cend_\i) at (\deg:\r) {};
		\draw[cyan,thick] (c_\i)--(cend_\i);
		\draw[-,red,thick] (red_\i)--(-\dist*\i+0.17, 9.7)--(redend_\i);
		\draw[-,color=black!20!blue,thick](blue_\i)-- (-\dist*\i, 10.04-0.06*\i)--(blueend_\i);
		\draw[-,olive,thick] (oliveend_\i)--(-\dist*\i+0.17,2.5);
		\draw[-,color=orange,thick] (orangeend_\i)--(-\dist*\i+0.34,2.5);
	}
	\end{scope}
	\begin{scope}[xshift=0.0001 cm,yshift=0 cm]
	\node[draw, point, label=45:{$p$}] (P-start) at (-3.35,10) {};
	\node (P-B) at (-3.35,2.38) {};
	\foreach \i in {10,30,...,170}{
		\node[inner sep=0pt, minimum width=4pt] [shift={(3.35,-10)}](q_\i) at (\i:\r) {};
		\draw[->, thick] (P-start) -- (q_\i);}
	\draw[<-, thick] (P-start) -- (P-B);
	\node at (-3.6,3) {$P$};
	\end{scope}
	\begin{scope}[xshift=-8 cm,yshift=0 cm]
	\node[red] at (1,9.3) {$B$};
	\node[cyan] at (0.2,9.3) {$C$};
	\node[orange] at (1,7.1) {$A$};
	\node[color=black!20!blue] at (1,4.9) {$D$};
	\node[olive] at (1,2.7) {$E$};
	\end{scope}
	
	\end{scope}	
	\end{tikzpicture}
	\caption{Schematic illustration of the tree model of a $T_{3,2}$. Here lines depict paths of unknown length while vertices emphasize known start or end points of the paths. }
	\label{fig:findingTk2}
\end{figure}

\medskip
\noindent\textbf{Bounding $h(t)$.}
In this part, we assume that there is a $q \in \{q_1, \ldots, q_\ell\}$ such that $q$ hosts more than $g_B(t)$ vertices $b_j \in B^o$.
For $i \in [0,g_2(t)]$ let $h_i(t) = (g_1(t))^{g_2(t) - i}$ and define $s_0 = q$.
Notice that $h_0(t) = g_B(t)$.

Let $J^0$ be the indices associated with the vertices of $B^o$ that are hosted by $s_0$.
If all children of $s_0$ host at most $h_1(t)$ vertices $b_i$ with $i \in J^0$ then we apply \cref{claim:interval_independent_set} with input $J^0, s_0$, and $h_0(t)$, since $g_1(t) \cdot h_1(t) = h_0(t)$, to obtain the sets satisfying Property $(\star)$. 
Otherwise, some children of $s_0$ hosts at least $h_1(t)$ vertices of $B^o$ and we choose $s_1$ to be the furthest descendant of $s_0$ hosting a set of vertices $B^1 \subseteq B^o$ with $\abs{B^1} \geq h_1(t)$.
With this choice we can guarantee that $s_0$ hosts at least one $b \in B^o$ that is not hosted by $s_1$: if this is not the case then $s_1$ hosts the same $h_0(t)$ vertices of $B^o$ as $s_0$, but no children of $s_1$ hosts more than $h_1(t)$ vertices of $B^o$ (since otherwise we can move $s_1$ further away from $s_0$) and we can apply \cref{claim:interval_independent_set}.
Let $J^1 = \{i \in J^0 : b_i \in B^1\}$. 

Assume now that $i$ vertices $\{s_1, \ldots, s_i\}$ have been chosen this way.
If $i < g_2(t)$ then we can repeat the procedure described in the previous paragraph.
If all such children host at most $h_{i+1}(t)$ we apply \cref{claim:interval_independent_set} with input $J^i, s_i$, and $h_{i}(t)$ to obtain the sets satisfying Property $(\star)$. 
Otherwise, we proceed to the next iteration.

Let $t' = g_2(t)$.
Since $h_i(t) = g_1(t) \cdot h_{i+1}(t)$, this procedure does not end before $t'$ vertices $s_1, \ldots, s_{t'}$ are found.
If this is the case, we fall back to Case A.
We form a new path $P'$ by deleting every descendant of $p$ in $P$ and appending the path from $s_1$ to $s_{t'}$ to the resulting path.
For every $v \in X$ (we remind the reader that $X = \bigcup_{i \in [g(t)]} A'_i \cup B'_i$) we denote by $p'(v)$ the node in $V(P_v) \cap V(P')$ that is closer to $\low(v)$ in $T$, and consider the order $<_{P'}$ in which such vertices appear in $P'$.
By the choice of the vertices $s_i$, we know that for each of them we can pick one $b \in B^o$ that is hosted by $s_i$ but not by any other $s_j$ with $j > i$.
Let $B^*$ be the subset of $B^o$ formed by those vertices.
Clearly $<_{P'}$ defines a total order on $A^o$ (and thus we can keep the intervals $I_1, \ldots, I_{4t^2}$ defined in the beginning of Case B) and total order on $\{b : b \in B^*\}$.
Hence  we have the necessary conditions to enter Case A with $B^*$ instead of $B^o$ and $<_{P'}$ instead of $<_P$, skipping the steps necessary to organize the nodes $p(a_i)$ and to ensure that all such nodes appear in $P'$ before every node of $B^*$, and the result follows.
\end{proof}

By \cref{thm:rooted_directed_path_graphs_delineation} and \cref{thm:seq-or-tk} we conclude that the class of rooted directed path graphs is effective delineated.

\section{Segment graphs}\label{sec:segments}

In this section, we explore the twin-width and delineation of (subclasses of) segment graphs.
Pure (hence triangle-free) axis-parallel unit segment graphs were shown to have unbounded twin-width~\cite{twin-width2}, by constructing a family of such graphs with super-exponential growth.
This family contains arbitrarily large bicliques (see~\cite[Figure 4]{twin-width2}). 
We will show that bicliques are necessary to make the twin-width large, even when we lift the requirements that the segments are axis-parallel and unit.

\begin{theorem}\label{thm:ktt-free-segments}
  $K_{t,t}$-free segment graphs have bounded twin-width.
\end{theorem}

In the previous theorem, one cannot relax the $K_{t,t}$-freeness assumption to $H_t$-freeness.
Let $B_n$ be the graph obtained from the 2-subdivision of a biclique $K_{n,n}$ by adding back the edges of the original biclique.
The left part of~\cref{fig:ap-unit-seg-utww} shows that, for every $n \in \mathbb N$, the graph $B_n$ is realizable with axis-parallel segments of two different lengths.
Note however that $B_n$ has no semi-induced $H_4$, and that $\lim_{n \to \infty}\tww(B_n) = \infty$.

To establish the latter claim, one can for instance ``remove'' the edges of the biclique by means of an FO transduction, and invoke~\cref{thm:transduction} and the third item of~\cref{thm:unbounded-tww}.
The transduction first marks the long horizontal segments by unary relation $U_1$ (color 1), and the long vertical segments, by unary relation $U_2$ (color 2), and interpret the new edges as $\varphi(x,y) \equiv E(x,y) \land \neg (U_1(x) \land U_2(y)) \land \neg (U_1(y) \land U_2(x))$. 

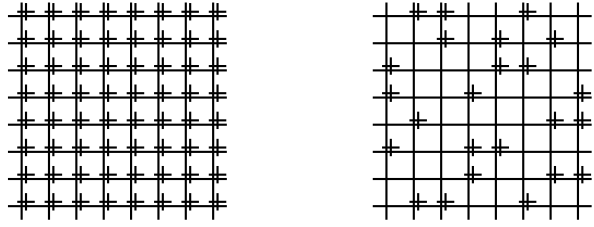
\begin{figure}[h!]
  \centering
  \begin{tikzpicture}[scale=0.6]
    \def\t{8}
    \def\s{0.6}
    \def\z{0.16}
    \pgfmathsetmacro\b{0.5 * \s}
    \pgfmathsetmacro\e{\s * \t + 0.5 * \s}
    \foreach \i in {1,...,\t}{
      \draw[thick] (\s * \i,\b) -- (\s * \i,\e) ;
      \draw[thick] (\b,\s * \i) -- (\e,\s * \i) ;
    }
    \foreach \i in {1,...,\t}{
      \foreach \j in {1,...,\t}{
        \draw[thick] (\s * \i + \z * \s,\s * \j - \z * \s) -- (\s * \i + \z * \s,\s * \j + 3 * \z * \s) ;
        \draw[thick] (\s * \i - \z * \s,\s * \j + \z * \s) -- (\s * \i + 3 * \z * \s,\s * \j + \z * \s) ;
      }
    }

    \begin{scope}[xshift=8cm]
      \foreach \i in {1,...,\t}{
      \draw[thick] (\s * \i,\b) -- (\s * \i,\e) ;
      \draw[thick] (\b,\s * \i) -- (\e,\s * \i) ;
    }
    \foreach \i/\j in {1/3,1/5,1/6, 2/1,2/4,2/8, 3/1,3/7,3/8, 4/2,4/3,4/5, 5/3,5/6,5/7, 6/1,6/6,6/8, 7/2,7/4,7/7, 8/2,8/4,8/5}{
        \draw[thick] (\s * \i + \z * \s,\s * \j - \z * \s) -- (\s * \i + \z * \s,\s * \j + 3 * \z * \s) ;
        \draw[thick] (\s * \i - \z * \s,\s * \j + \z * \s) -- (\s * \i + 3 * \z * \s,\s * \j + \z * \s) ;
      }
    \end{scope}
  \end{tikzpicture}
  \caption{Left: An axis-parallel $H_4$-free two-lengthed segment graph realizing $B_n$ (here drawn with $n=8$), whose twin-width grows with $n$. Right: Axis-parallel segment graphs are transduction equivalent to $\mathcal B_{\leqslant 3}$, thus not delineated.}
  \label{fig:ap-unit-seg-utww}
\end{figure}

Similarly the 2-subdivision of any subcubic bipartite graph, augmented with the biclique between its two partite sets, is realizable with axis-parallel segments of two different lengths (see right-hand side of~\cref{fig:ap-unit-seg-utww}).
Indeed those graphs --let us denote by $\mathcal C$ the class they form-- are induced subgraphs of some $B_n$.
We claim that $\mathcal C$ and $\mathcal B_{\leqslant 3}$ are transduction equivalent, and by~\cref{lem:not-delineated}, two-lengthed axis-parallel segments are not delineated.

To see that $\mathcal C$ transduces $\mathcal B_{\leqslant 3}$, use three extra unary relations $A, B, C$ partitioning the vertex set of some $G \in \mathcal C$ into the long vertical segments, the long horizontal segments, and the short segments, respectively.
Then the transduction fixes the new domain by $\nu(x) \equiv A(x) \lor B(x)$ and edge set by $\varphi(x,y) \equiv A(x)~\land~B(y)~\land~\exists z_1 \exists z_2~(C(z_1)~\land~C(z_2)~\land~E(x,z_1)~\land~E(z_1,z_2)$ $\land~E(z_2,y))$.

Conversely, $\mathcal B_{\leqslant 3}$ transduces $\mathcal C$.
Just observe that graphs $G$ of $\mathcal C$ are bipartite, and admits a partition in $(A_1,A_2,B_1,B_2)$ such that $(A_1 \cup A_2, B_1 \cup B_2)$ is their bipartition, $A_1 \cup A_2$ induces a biclique, and every vertex of $A_1 \cup A_2$ (resp.~$B_1 \cup B_2$) has at most three neighbors in $B_2$ (resp.~$A_2$).
Thus if one removes the biclique between $A_1$ and $B_1$, one obtains a graph $G' \in B_{\leqslant 3}$.
Given $G'$, the transduction then guesses $A_1$ and $B_1$ with two unary relations, and adds all the edges in between them. 

In the construction of~\cref{fig:ap-unit-seg-utww}, we use two different lengths for the segments.
We show that with a unique length (unit segments), axis-parallel $H_t$-freeness implies bounded twin-width.

\begin{theorem}\label{thm:ap-ht-free-unit-segments}
  Axis-parallel $H_t$-free unit segment graphs have bounded twin-width.
\end{theorem}

Actually we prove a stronger statement than the previous theorem, where segments are not necessarily unit, but the ratio between the largest and the smallest lengths is bounded.
Again it shows that the fact that this ratio is unbounded in~\cref{fig:ap-unit-seg-utww} (left) is unavoidable.

%
%

\subsection{$K_{t,t}$-free segment graphs have bounded twin-width}

The main result of this subsection is summarized in the next theorem.

\begin{theorem}\label{thm:biclique-segment}
There is a function $h\colon \mathbb{N} \rightarrow \mathbb{N}$, such that $K_{t,t}$-free  segment graphs have twin-width at most $h(t)$.
\end{theorem}

For each fixed $t$ we provide an FO transduction ${\mathsf T}_t$ from a class of 2-edge-colored graphs of bounded twin-width 
to a class containing all $K_{t,t}$-free segment graphs. Once this is done, it follows immediately from~\cref{thm:transduction} that $K_{t,t}$-free segment 
graphs have twin-width at most $h(t)$.

Our target property of 2-edge-colored graphs is \emph{matching-augmented planar graphs}.\footnote{Actually we shall  enhance such graphs with {\sl facial vertices} and transduce from 
  these enhanced graphs.}
We say that a planar graph is \emph{matching-augmented} if it can be obtained from a plane graph~$G$
and a collection $\C$ of facial cycles which are pairwise vertex-disjoint faces by adding a~matching $M$ 
such that the endpoints of any matching edge lie on a facial cycle in~$\C$. 
Let $M_C\subseteq M$ be the set of matching edges both of whose endpoints lie on a facial cycle $C\in \C$.  
The matching $M$ is \emph{nicely aligned} with respect to $\C$  
if for every $C\in \C$, there exists two vertices $s$ and $t$ such that 
$M_C$ is of the form $\{u_iv_i:i\in [\ell]\}$ or $\{st\}\cup \{u_iv_i:i\in [\ell]\}$ for some $\ell$, 
where $u_i$'s are entirely contained in one path of $C-\{s,t\}$, $v_i$'s in another path of $C-\{s,t\}$, 
and $u_1\prec \cdots u_\ell\prec v_\ell \prec \cdots \prec v_1$ for a cyclic order $\prec$ on $V(C)$.

Clearly, a matching-augmented planar graph can be seen as a relational structure with two binary relations $E$ and $M$, the former  
being the edges of the underlying planar graph and the latter being the  matching edges. 
It will be shown later that the class of matching-augmented planar graphs with nice alignment property have bounded twin-width. 
Given that, we first describe how to obtain a matching-augmented planar graph $P_G$ with the nice alignment property from any $K_{t,t}$-free segment graph $G$ 
and how the original graph $G$ can be recovered by an FO transduction ${\mathsf T}_t$.

\subsubsection{Transducing $K_{t,t}$-free segment graphs from matching-augmented planar graphs}

Let us begin by explaining how to obtain a matching-augmented planar graph from a $K_{t,t}$-free segment graph. 
We use the bounded degeneracy of $K_{t,t}$-free segment graphs.
 
\begin{theorem}[\cite{Lee17}]\label{thm:lee}
    There is a constant $c$ such that every $K_{t,t}$-free segment graph is $c t \log t$-degenerate.
\end{theorem}

Let $d:=c t \log t$ with the constant $c$ from~\cref{thm:lee}. 
Let $G$ be a $K_{t,t}$-free segment graph  and $<$ be a $d$-degeneracy order of $V(G)$ 
so that $N^<(v)=\{u: uv\in E(G), u<v\}$ is of size at most~$d$.
Fix a segment representation $\mathcal S=\{S_v \subseteq \mathbb{R}^2:v\in V(G)\}$. 
Without loss of generality, the intersection of any two segments is internal to both segments.

\bigskip

\noindent {\bf Constructing a planar graph on sub-segments of $\mathcal S$.} 
Split each segment $S_v \in \mathcal S$ into at most $d+1$ parts, by breaking it at every intersection with the segment $S_u$ with $u<v$. 
For notational simplicity, let us assume that each segment $S_v\in \mathcal S$ except the segment of the first vertex in $<$ is broken into exactly $d+1$ parts, which we refer to 
as $S_v^1,\ldots , S_v^{d+1}$ in the order of appearance in $S$; this does not affect the validity of the argument henceforth. 
Let $\mathcal S'$ be the at most $1+(d+1)(n-1)$ newly created segments. 

Note that for each $v\in V(G)$ and $i\in [d+1]$, 
there are at most two neighbors of $v$, preceding it in $<$, whose segments break $S_v$ to yield $S_v^i$. Let 
$N^<(S_v^i)$ be the set of such neighbors of $v$.
Let $N^>(S_v^i)$ be the set of neighbors $u$ of $v$ such that $S_v^i$ splits $S_u$ into new segments put in $\mathcal S$.  
Notice that any $u\in N^>(S_v^i)$ comes after $v$ in the order $<$, 
and $v$ may have arbitrarily many such neighbors. Clearly, we have $N^>(v)=\bigcup_{i\in [d+1]} N^>(S_v^i).$ 
We remark that for the first vertex $v$ in the order $<$, $N^<(v)$ is empty. It is clear that 
$u\in N^>(S_v^i)$ for some $i\in [d]$ if and only if $v\in N^<(S_u^j)$ for some $j\in [d]$.


To construct a matching-augmented planar graph $P_G$, let us shrink 
 each newly created segment in $\mathcal S'$ by $\delta$ at the endpoints, 
where $\delta$ is an arbitrarily small constant. 
Now, around every segment $S\in \mathcal S'$ place a cell $F(S)\subseteq \mathbb{R}^2$ whose boundary is away 
from $S$ by $\delta/4$. Note that the cells obtained from $\mathcal S'$ are pairwise disjoint. 

We construct a plane graph $P_G$, in which   the cells $\{F(S):S\in \mathcal S'\}$ are faces, as follows. 
For every $S:=S_v^i\in \mathcal S'$, the boundary of each cell $F(S)$ becomes a facial cycle of $P_G$ 
on $\abs{N^<(S)}+2\abs{N^>(S)}$ vertices; one vertex $x^S_u$ for each vertex of $u\in N^<(S)$ and two vertices $x^S_{u,1},x^S_{u,2}$ for each of $u\in N^>(S)$. 
These  vertices are placed around the facial boundary in a canonical way governed by the segment representation $\mathcal S$, namely:
\begin{itemize}
\item for each $u\in N^<(S)$, consider two intersection points of $S_v\in \mathcal S$ with the boundary of the cell $F(S)$. 
The vertex $x^S_u$ is located at the intersection point closer to $S_u\cap S_v,$
\item for each $u\in N^>(S)$, the vertices $x^S_{u,1}$ and $x^S_{u,2}$ are respectively located at the two intersection points 
of $S_u$ with the boundary of the cell $F(S),$
\item the vertex sets $\{x^S_{u,1}:u\in N^>(v)\}$ and  $\{x^S_{u,2}:u\in N^>(v)\}$ are consecutive, respectively. 
\end{itemize}

The vertex set $V(P_G)$ of $P_G$ consists of the union of sets $V(S):=\{x^S_u: u\in N^<(S)\}\cup \{x^S_{u,1}, x^S_{u,2}: u\in N^>(S)\}$ over 
all segments $S\in \mathcal S'.$ The edge set $E(P_G)$ of $P_G$ contains the edge set on the facial cycle of $F(S)$ for every $S\in \mathcal S'$. 
Additionally, $P_G$ contains two edges for every edge $uv\in E(G)$ with $u<v$: 
recall that $S_u$ breaks $S_v$, and let $S_v^j$ and $S_v^{j+1}$ be the two resulting segments in $\mathcal S'$. 
Let $i\in [d+1]$ be such that $v\in N^>(S_u^i)$. Notice that $u$ belongs to both $N^<(S_v^j)$ and $N^<(S_v^{j+1}),$ 
and thus the boundaries of $F(S_v^j)$ and $F(S_v^{j+1})$ contain $x^{S_v^j}_u$ and $x^{S_v^{j+1}}_u$, respectively.
We also note that the boundary of $F(S_u^i)$ contains $x^{S_u^i}_{v,1}$ and $x^{S_u^i}_{v,2}$ as $v$ belongs 
to $N^>(S_u^i)$. Now we connect $x^{S_u^i}_{v,1}$ with $x^{S_v^j}_u$, and $x^{S_u^i}_{v,2}$ with $x^{S_v^{j+1}}_u$ 
(or change the order of $j$ and $j+1$ so that the edges are disjoint on the plane). It is straightforward to verify that 
$P_G=(V(P_G),E(P_G))$ with the current embedding is a plane graph.
The facial cycles $\C$ of the cells $\{F(S):S\in \mathcal S'\}$ form a collection of pairwise vertex-disjoint facial cycles. 

On top of the edges added so far, we define a matching $M$ as follows: for every $S\in \mathcal S'$ 
and for every $u\in N^>(S)$, 
$M$ contains the edge between $x^{S}_{u,1}$, and $x^{S}_{u,2}$. If $\abs{N^<(S)}=2$, 
then $M$ contains an edge connecting the two vertices on the boundary of $F(S)$ corresponding to the vertices of $N^<(S).$ 
It is straightforward to see that $M$ is nicely aligned with respect to~$\C$. This completes the description of the matching-augmented planar graph $P_G$ with nice alignment property. 

\bigskip

\noindent {\bf Reconstructing $G$ from $P_G$ via FO transduction.} Consider the 2-edge-colored graph on $V(P_G)$ 
with the 2-colored edge set $E(P_G)\cup M$. 
Notice that the adjacency of a vertex $v$ of $G$ is encoded in $P_G$ as the adjacency of the vertex sets
$V(S):=\{x^S_u: u\in N^<(S)\}\cup \{x^S_{u,1}, x^S_{u,2}: u\in N^>(S)\}$, where $S=S_v^i$ over all $i\in [d].$ 
To ease the FO transduction, we enhance $P_G$ with a facial vertex $x_S$ for each $S\in \mathcal S'$ 
which is adjacent with every vertex in the facial cycle of $F(S)$ and with no other vertices. It is clear that the enhanced graph $P^*_G$ 
is a plane graph (before adding $M$). 

The transduction ${\sf T}$ first colors all facial vertices $x_S$ with a unary relation $F$. 
Among the facial vertices, those of the form $x_{S_v^1}$ for all $v\in V(G)$ will be distinguished with a unary relation $F^o$; 
these vertices will be used as the vertex set on which $G$ is defined. 
Due to $d$-degeneracy of $G$, we know that $G$ is $d+1$-colorable. Fix a $(d+1)$-coloring $c:V(G)\rightarrow [d+1]$ 
and color the facial vertices and some additional vertices accordingly. That is, for each $S\in \mathcal S'$, the associated facial vertex $x_S$ and 
at most two vertices in $\{x^S_u: u\in N^<(S)\}$ have color $c(v)$ 
if $S\subseteq S_v\in \mathcal S$. For every $i\in [d+1]$, let $F_i$  be the unary relation that distinguishes vertices of color $i$ in $V(P^*_G)$.

First, the local adjacency between segments of $ \mathcal S'$ is defined with the FO formula. 
\[
\varphi_{\mathcal S'}(x,y) \equiv x,y\in F \wedge  {\sf dist}_3(x,y),
\]
where ${\sf dist}_k(x,y)$ is a binary relation on $V(P^*_G)$ for vertices of distance exactly $k$ in $P^*_G$. Clearly, 
the relation ${\sf dist}_k$ is FO definable for every constant $k$. From the construction of $P^*_G$, 
it is clear that two segments of $\mathcal S'$ have facial vertices of distance three if and only if 
they are locally adjacent, i.e., $v'\in N^<(S)$ or $v\in N^<(S')$, where $S\subseteq S_v$ and $S'\subseteq S_{v'}$.

Now, we want to aggregate the adjacency of all facial vertex $x_S$ for all $S\in \mathcal S'$ contained in $S_v$ 
as the adjacency of the facial vertex $x_{S_v^1}\in F^o$. For this, we use the path connecting 
$x_{S_v^1}$ with $x_S$ via which traverses a matching edge when crossing a face originating from a segment of $\mathcal S$ other than $S_v$ itself. 
Specifically, for any $S,S'\in \mathcal S'$ from the same segment $S_v\in \mathcal S$,
there exists a path $Q$ in $(V(P^*_G),E(P^*_G)\cup M)$ connecting $x_S$ and $x_{S'}$ which satisfies the following conditions: 
$(i)$ $Q$ has length $3k-1$ for some $k\in [d]$, $(ii)$ precisely every third edge on $Q$ belongs to $M$, and 
$(iii)$ every vertex on $Q$ which is not the endpoint of a matching edge in $E(Q)\cap M$ has color $c(v)$. Moreover, no such path exists 
if $S$ and $S'$ are obtained from distinct segments of $\mathcal S$.

Let ${\sf alt}(x,y)$ be a formula which states that between $x\in F^o$ and $y\in F$, there exists a path satisfying the  conditions $(i)$-$(iii)$. 
It is easy to see that such a formula ${\sf alt}(x,y)$ can be written using the additional unary relations $F_i$'s and the binary relation $M$. 
Now, the formula 
\[
\varphi_{\mathcal S}(x,y) \equiv x,y\in F^o  \wedge  \exists w_x,w_y\in F \big ( {\sf alt}(x,w_x) \wedge {\sf alt}(y,w_y)  \wedge \varphi_{\mathcal S'}(w_x,w_y) \big ).
\]
defines the adjacency on $F^o$ which is isomorphic to the edge relation of $G$. 
%

It remains to see that the class of matching-augmented planar graphs that can be obtained from $K_{t,t}$-free segment graphs 
have bounded twin-width. This will be the focus in the remainder of this subsection.

\subsubsection{Planar graph with a packing of facial cycles}

Let $G$ be a planar graph with a fixed embedding $\Sigma$ on the plane 
and $\C$ be a packing of facial cycles which are pairwise vertex-disjoint. 
We assume that $G$ is connected and without parallel edges. We also assume 
that a (counter-clockwise) cyclic order of  edges incident with each vertex $v$ under the embedding $\Sigma$ is given. 
To minimize verbosity, we shall talk about cyclic order of the neighbors of $v$, which means the cyclic order of the edges incident with the neighbors 
whose other endpoints meet at $v$.  
Therefore, we can query the ternary relations such as the following: for three neighbors $w_1,w_2,w_3$ of $v$, 
$w_2$ is {\sf after} $w_1$ and {\sl before} $w_3$ in the counter-clockwise transversal around $v$. 
We write $(w_1,w_2,w_3)_v$ if this query is returned in positive. 
For a facial cycle $C$ and a vertex $v\in V(C)$, denote by ${\sf succ}_C(v)$ (resp. ${\sf pred}_C(v)$) the vertex on $C$ which immediately succeeds (resp. precedes) $v$ 
in the counter-clockwise transversal of $C$ under the embedding $\Sigma$.
%
%

Below, we prove that a planar graph with $\C$ admits a vertex ordering which induces a counter-clockwise total order 
around each facial cycle of $\C$ and have bounded twin-width.

\begin{proposition}\label{prop:counterplanar}
Let $G$ be a planar graph with a fixed planar embedding $\Sigma$ and $\C$ be a collection of pairwise vertex-disjoint facial cycles. 
Then there exists a vertex ordering $\prec$ on $V(G)$ with the following properties.
\begin{enumerate}[$(i)$]
\item $\adj{\prec}{G}$ has bounded twin-width, and
\item for every $C\in \C$, the vertex ordering $\prec$ restricted to $V(C)$ is a counter-clockwise transversal order along the facial cycle $C$ under $\Sigma$.
\end{enumerate}
\end{proposition}

Most of this subsection is devoted to proving~\cref{prop:counterplanar}. For this, we first construct a vertex ordering $\prec$ which naturally satisfies the property $(ii)$, 
and then establish $(i)$. Recall that the actual planar graph we use for the FO transduction is one which is enhanced by 
adding a facial vertex on each $C\in \C$ and making the facial vertex adjacent with every vertex on $C$. 
For such enhanced graph, essentially the same ordering with a trivial modification works to bound the twin-width. We discuss how to 
handle the enhanced graphs at the end of this subsection.

\bigskip

\noindent {\bf Constructing a vertex ordering $\prec$.}
We aim to present a vertex ordering $\prec$ of $G$ which has the two properties indicated in~\cref{prop:counterplanar}. 
For simplicity, we assume that $\C$ partitions the vertex set of $G$ by treating each vertex as a trivial facial cycle.
For each vertex $v$ of $G$, let $C_v\in \C$ denote the facial cycle containing $v$. 

Along with a vertex ordering $\prec$, we construct a spanning tree $T$ of $V(G)$ akin to BFS tree of $G$ and a layered decomposition 
$\mathcal L=L_0\uplus \cdots \uplus L_p$. 
Essentially, our tree $T$ will be a natural extension of a BFS tree in which the neighbors of a discovered vertex $v$ are 
visited in the counter-clockwise order around $v$ under the  embedding $\Sigma$, starting from (the neighbor after) the parent of $v$ in the BFS tree constructed thus far. 
The associated order is the discovery order of vertices during the BFS tree construction.
The crucial difference of our desired $T$ from the aforementioned BFS tree is twofold. First, 
when a vertex $v$ from a facial cycle $C$ of $\C$ is discovered, the remaining vertices of $C$ will be discovered only by $v$ 
by temporarily deviating from BFS. Second, when $v$ of a facial cycle $C$ in $\C$ 
is explored, the neighbors of $v$ in $C_v$ are ignored except for ${\sf succ}_{C}(v)$. Moreover, once ${\sf succ}_C(v)$ is discovered (after discovering 
the neighbors of $v$ before ${\sf succ}_C(v)$), we fully discover the vertices of $C$ in the counter-clockwise order of $C$ before we continue exploring $v$. 
In other words, once a single vertex $v$ from $C$ is discovered, all the remaining $V(C)\setminus v$ are deemed {\sl privately discovered} by~$v$ and thus  
quarantined against further discovery by other vertices. But instead of inserting $V(C)$ 
in the queue at this point, only $v$ is placed in the queue and the remaining $V(C)\setminus v$ will wait to be put in the queue 
till ${\sf succ}_{C_v}(v)$ is (publicly) discovered by $v$ during its exploration phase. 
The layered decomposition  represents the partition of $V(G)$ according to the the distance of each vertex to the root $r$ in $T$, where the distance between 
${\sf succ}_C(v)$ to any vertex of $V(C)\setminus \{v,{\sf succ}_C(v)\}$ is deemed zero. 

Let us give a formal description on how to construct of $T$, a total order $\prec$ on $V(G)$ and the layered decomposition $\mathcal L$ as a partition of $V(G)$.
As usual BFS, there are three states: \white\ means undiscovered, which is the initial state for all vertices of $G$, 
\gray\ means discovered but not explored (it is the state of vertices in the queue), and \black\ is the state of fully explored vertex. 
On top of these three states, there is one more state, \reserved($v$). When a vertex $u$ is marked \reserved($v$), it means that 
$v$ is in the same facial cycle of $\C$ with $u$ and $v$ is already discovered. Additionally, we give to each vertex $v$ an attribute $\parent(v)$, 
to construct $T$. The initial value of $\parent(v)$ is null for all $v$. 

Choose an arbitrary vertex $r$ of $G$, mark $r$ \gray\ and mark all vertices of $V(C_r)\setminus \{r\}$ by \reserved($r$). 
In general, we explore the first vertex $v$, marked \gray, in the queue and do the following. Examine the 
edges incident with $v$ in the counter-clockwise order starting from the edge connecting $v$ with $\parent(v)$; if null, then choose an arbitrary edge. 
Let $u$ be the other endpoint of the edge at hand. 
\begin{enumerate}[(a)]
\item If $u$ is marked by \black, \gray,  \reserved($v'$) with $v'\neq v$, or \reserved($v$) but $u\neq {\sf succ}_{C_v}(v)$, we do nothing and continue with 
the counter-clockwise exploration of $v$.
\item If $u$ is marked by \white, then change $u$'s label to \gray\ and mark all $V(C_u)\setminus u$ by \reserved($u$). Insert $u$ in the queue with the attribute 
$\parent(u)=v$.
\item If $u$ is marked by \reserved($v$) and $u= {\sf succ}_{C_v}(v)$, 
then mark $V(C)\setminus v$ \gray\ and let $\parent(z)={\sf pred}_{C_v}(z)$ for every $z\in V(C)\setminus v.$ 
Put $V(C_v)\setminus v$ in the queue in the counter-clockwise order. 
\end{enumerate}

Once $v$ is done with exploration, then mark it \black. 
Because save $u= r$, the attribute $\parent(u)$ changes to a non-null value, i.e., to the vertex which contributed to the discovery of $u$ 
exactly once, namely when $u$ is put in the queue,  it defines a spanning tree $T$ rooted at~$r$. The total order $\prec$ is the discovery order. 
It is clear from the construction that $\prec$ satisfies the condition $(ii)$ of~\cref{prop:counterplanar}, namely 
the vertices of each facial cycle $C\in \C$ are ordered under $\prec$ counter-clockwise.

To define the layered decomposition $\mathcal L$, we distinguish the tree edges of $T$ into two types. When $u$ is discovered 
during the exploration of $v$ and $\parent(u)=v$, then the tree edge $uv$ is called \emph{explored}. If $\parent(u)\neq v$, 
then we call the tree edge between $u$ and $\parent(u)$ a \emph{fast-track} edge. Note that 
if $u$ is discovered in the case (b) above, the tree edge between $u$ and $\parent(u)$ is an explored edge. If $u$ is discovered in the case (c) above,
then the tree edge between $u$ and $\parent(u)$ is a fast-track edge except for the first $u$ discovered as a successor of $v$ on $C_v$. 
Now, the layer $L_i$ consists of all vertices $v$ for which the path connecting $v$ and $r$ in $T$ has exactly $i$ explored edges.

\bigskip

\noindent {\bf Properties of $\prec$, $T$ and $\mathcal L$.}
We begin with some basic properties of $\mathcal L=L_0\uplus \cdots \uplus L_p.$

\begin{lemma}
The layered decomposition $\mathcal L$ satisfies the following properties.
\begin{enumerate}[$(i)$]
\item $L_0\prec L_1 \prec \cdots \prec L_p$.
\item There is no edge of $G$ connecting a vertex in $L_i$ and a vertex in $L_{i+3}$ for any $i$.
\end{enumerate}
\end{lemma}
\begin{proof} 
To see $(i)$, recall that the attribute $\parent(u)$ is a pointer to the vertex during whose (fast-track) exploration 
$u$ was discovered. Clearly, before the discovery of~$u$ its ancestor should be discovered. Therefore, the discovery order $\prec$ 
is a linear extension of the tree order imposed by~$T$. The property $(i)$ follows by the construction of layers. 

To see $(ii)$, suppose not and let $uv$ be an edge of $G$ with $u\in L_i$ and $v\in L_{i+3}$.
Clearly, we have $C_u\neq C_v$ because any facial cycle of $\C$ has exactly one explored tree edge. 
On the other hand, the first discovered vertex $u_o$ of $C_v$ is in $L_{i+2}$. However, the existence of the edge $uv$ 
indicates that $v$ could have been discovered during the exploration of $u$, therefore before $u_o$ is discovered, a contradiction.
\end{proof}

For a vertex $x$ in the layer $L_i$, let $T_x$ be the connected component of $T$ containing $x$ obtained by 
removing all edges whose both endpoints are contained in $L_{\leq i}$. The next lemma follows easily  
from the first-come-first-served nature of the queue (discovery order is the exploration order), 
and the fact that the vertices of a facial cycle $C\in \C$ is put in the queue as a batch once the successor of 
the first discovered vertex of $C$ is discovered.

\begin{lemma}\label{lem:order2}
Let  $x, y$ be two vertices in the layer $L_i$. If  $x\prec y$,  
then $V(T_x)\cap L_k \prec V(T_y)\cap L_k$ for every $k\geq i$.
\end{lemma}

The next lemma is immediate from~\cref{lem:order2}.

\begin{lemma}\label{lem:order1}
Let $v$ be a vertex in the layer $L_i$ and $x,y$ be two children of $v$ in $T$ with $x\prec y$. 
If $x,y\in L_{i+1}$, then $V(T_x)\cap L_k \prec V(T_y)\cap L_k$ for every $k\geq i$.
\end{lemma}

The above two lemmas lead to the statement below, which is the key for proving that $\adj{\prec}{G}$ has bounded twin-width. 
It essentially says that if there are vertices $v_1 \prec \cdots \prec v_\ell$ from the same layer, 
then in the boundary of the disk covering $T$ truncated after the said layer, the order $\prec$ coincides with the cyclic order on 
the boundary vertices in the counter-clockwise order.  

\begin{lemma}\label{lem:niceorder}
Let $v_1 \prec \cdots \prec v_\ell$ be vertices in the same layer of $\mathcal L$, and let $T'$ 
be a minimal subtree of $T$ containing all $v_i$'s. Then the graph $H$ on the vertex set $V(T')$ 
with the edge set $E(T')\cup \{v_iv_{i+1}:i\in [\ell]\}$ is planar, where $v_{\ell+1}=v_1$. 
Furthermore, there exists an embedding $\Sigma'$ of $H$ 
in which the embedding of $T$ coincides with that in $\Sigma$ and all the vertices of $\{v_i:i\in [\ell]\}$ lie on the outer facial cycle. 
\end{lemma}
\begin{proof}
We assume that any pair of vertices in $\{v_i:i\in [\ell]\}$ are pairwise incomparable in $T$; the extension of the proof for the general case 
is tedious. Therefore, the set $\{v_i:i\in [\ell]\}$ is the set of the leaves of $T'$. 

Consider a closed disk on the Euclidean plane that contains $T'$ and whose boundary intersects $V(T')$ exactly on the leaves $\{v_i:i\in [\ell]\}$ of $T'$. 
Such a disk can be obtained in the usual way of splitting the edges of $T'$ into two half-edges to obtained an Eulerian tour of $T'$ along half-edges, 
then bypassing all vertices except for the leaves of $T'$. 
Let $u_1\prec_{T'} u_2 \prec_{T'} \cdots \prec_{T'} u_{\ell}$ be a counter-clockwise transversal of $\{v_i:i\in [\ell]\}$ along the boundary 
with $u_1=v_1$. We claim that for every $2\leq i<j \leq \ell$, it holds that $v_i \prec_{T'} v_j$ and thus $\prec_{T'}$ coincides with $\prec.$ 
Choose arbitrary $v_i$ and $v_j$ with $i<j$ and let $T''$ be the minimal subtree of $T'$ containing $v_1, v_i$ and $v_j.$ 
Note that $T''$ is a subdivided star with three leaves 
and let $v$ be the vertex of $T''$ of degree three. Let $v'_1, v'_j$ and $v'_i$ be the three neighbors of $v$ 
which is closer to $v_1,v_j$ and $v_i$ respectively. 

Let $L_{k'}$ be the  minimum layer among all layers of $\mathcal L$ with non-empty intersection with $V(T'')$. 
Let $L_{k}$ be the layer containing $v$. We consider a few cases.  

\smallskip

\noindent {\bf Case A. $L_k=L_{k'}$}. 
Note that at most one of the three edges $vv'_1, vv'_i$ and $vv'_j$ can be fully contained in $L_k$. 
Indeed, any tree edge fully contained in a single layer is an edge of a facial cycle in $\C$. 
Therefore, if any of these edges fully contained in $L_k$, then $C_v$ must contain such an edge. 
This implies that if all three edges fully contained in $L_k$, then two of then coincides and thus $v$ cannot have degree three in $T''$.

Suppose none of these three edges is fully contained in $L_k$, and thus 
all three vertices $v'_1,v'_j$ and $v'_i$ are in the layer $L_{k+1}$. By~\cref{lem:order1}, 
and the condition $v_1\prec v_i \prec v_j$, we have $v'_1\prec v'_i \prec v'_j$. 
As $\prec$ is the discovery order and we discover the neighbors of $v$ in the counter-clockwise order, 
this imposes $v_1 \prec_{T'} v_i \prec_{T'} v_j$. 

Therefore, we may assume that at least one of $vv'_1, vv'_i$ and $vv'_j$ is an edge of $C_v$. 
If $vv'_1$ is an edge on $C_v$ and the others are not, observe that $v'_1={\sf pred}_{C_v}(v)$ and $v'_j$ and $v'_i$ are both in $L_{k+1}$. 
Note that~\cref{lem:order1} applied to $v$ and $v'_j,v'_i$ and the fact $v_i\prec v_j$ imply that $v'_i\prec v'_j.$
As we construct $T$ and explore the neighbors of $v$, we discover them in the counter-clockwise order starting from the parent of $v$, 
which is $v'_1.$ Therefore, the ternary relation $(v'_1,v'_i,v'_j)_v$ holds and thus  
$v_1 \prec_{T'} v_i \prec_{T'} v_j$. 
If $vv'_i$ is an edge on $C_v$, then both $v'_1$ and $v'_j$ are in $L_{k+1}$. No matter whether $v'_i$ is a predecessor or successor of $v$ on $C$, 
~\cref{lem:order2} yields a contradiction either with $v'_1$ or $v'_3$. Lastly, if $vv'_j$ is an edge on $C$,~\cref{lem:order2} implies that 
$v'_j={\sf succ}_{C_v}(v)$. Observe that the ternary relation $(v'_1,v'_i,v'_j)_v$.

Suppose that two edges of $vv'_1, vv'_i$ and $vv'_j$ are fully contained in $L_k$. Using~\cref{lem:order2}, it is easy to verify 
that $v'_1={\sf pred}(v)$, $v'_j={\sf succ}(v)$ and $v'_i$ is in $L_{k+1}$. Therefore, again  we observe the ternary relation $(v'_1,v'_i,v'_j)_v$. 

\smallskip

\noindent {\bf Case B. $L_k\neq L_{k'}$}. For the sake of contradiction, suppose that we observe $v_1\prec_{T'}v_j\prec_{T'} v_i$ 
on the boundary of the disk. This impose the ternary relation around $v$ as $(v'_1,v'_j,v'_i)_v$. Note 
that exactly one of the paths connecting $v$ and $\{v_1,v_i,v_j\}$ in $T''$ intersects with $L_{k-1}$ due to the assumption $L_k\neq L_{k'}$ 
and the fact that $v\in L_k$ has degree three.  

First, suppose that the path connecting $v$ and $v_1$ intersects $L_{k-1}$ and note that $v'_1$ is the parent of $v$ in $T$. Therefore, 
if both $v'_j$ and $v'_i$ are in $L_{k+1}$, the ternary relation around $v$ mandates $v'_j\prec v'_i$, which is a contradiction with $v_i\prec v_j$ 
by~\cref{lem:order1}. If this is not the case, then we have $(v'_1,v'_j,v'_i)_v$ implies $v'_i={\sf succ}_{C_v}(v)$. Now~\cref{lem:order2} applied 
to $v$ and the first vertex in $L_k$ which is met by the path from $v_i$ to $v$ in $T''$ 
demands $v_j\prec v_i$, a contradiction. 

Second, suppose that the path connecting $v$ and $v_i$ intersects $L_{k-1}$ and thus $v'_i$ is the parent of $v$ in $T$.
Now either~\cref{lem:order1} or~\cref{lem:order2} applied to suitable vertices at $L_{k'}$ imply that 
$v_i$ cannot be between $v_1$ and $v_j$, a contradiction. 

Lastly, suppose that the path connecting $v$ and $v_j$ intersects $L_{k-1}$ and thus $v'_j$ is the parent of $v$ in $T$. 
In this case, the same argument as in the first case applies with the ternary relation $(v'_j,v'_i,v_1)_v$. This completes the proof. 
%
%
%
\end{proof}

\bigskip

\noindent {\bf $\adj{\prec}{G}$ has bounded twin-width.} 
Now we turn to show that $\adj{\prec}{G}$ has bounded twin-width. 
Specifically, we will show that for some $t$ whose value will be determined later, one cannot extract any universal pattern $\mat{t}{s}$ as an off-diagonal submatrix of $\adj{\prec}{G}$ 
for any $s\in \{0,1,\uparrow,\downarrow,\leftarrow,\rightarrow\}.$ 
The proof of this claim immediately implies that $\adj{\prec}{G}$ has constant twin-width by~\cref{thm:canonicaler}. 
For $s\neq 0$, a universal pattern $\mat{4}{s}$ already entails 
the existence of $K_{3,3}$ as a semi-induced subgraph of $G$, which is impossible. 

So, fix $t=3$ and suppose that $\adj{\prec}{G}$ contains $\mat{3t}{0}$ as an off-diagonal submatrix of $\adj{\prec}{G}.$ 
Let $\mathcal A=\{A_1,\ldots , A_{3t}\}$ be the row parts and $\mathcal B=\{B_1,\ldots , B_{3t}\}$ be the column parts in the $3t\times 3t$ subdivision of $\mat{3t}{0}$. 
Without loss of generality, we assume $A_i \prec B_j$ for any $i,j\in [3t]$ and thus 
$A_1\prec \cdots \prec A_{3t} \prec B_1 \prec \cdots B_{3t}$. 

Recall that there is no edge between $L_j$ and $L_{j+3}.$ Therefore, because every vertex of $A:=\bigcup_{i\in [3t]}A_i$ 
is adjacent with some vertex of $B:=\bigcup_{j\in [3t]} B_j$ and vice versa, 
the union $A\cup B$ is contained in three consecutive layers of $\mathcal L$. Without loss of generality, 
we assume $A\cup B\subseteq L_{p-2}\cup L_{p-1}\cup L_p$. As at most two part from $\mathcal A \cup \mathcal B$ can intersect two layers,  
we may assume that each $A_i$, respectively $B_j$, is entirely contained either in $L_{p-2}$, $L_{p-1}$ or $L_p$. 

Notice that at least one of the three layers which contains at least $t$ parts of $\mathcal A$, and the same holds with $\mathcal B$. 
Let $L_A$ and $L_B$ be such layers. 
Let $\mathcal A'=\{A_i:i\in [t]\}$ (respectively, $\mathcal B'=\{B_j: j\in [t]\}$)
be the $t$ parts of $\mathcal A$ in $L_A$ (respectively, $t$ parts of $\mathcal B$ in $L_B$).
There are two cases to consider: $L_a=L_b$ and $L_a\prec L_b.$

Let us first consider the case when $L_a=L_b$. As the vertex set $\bigcup_{i\in [t]}A_i \cup \bigcup_{i\in [t]}B_i$ 
is fully contained in the same layer, we can apply~\cref{lem:niceorder} and obtained a new planar graph $H$ 
on the vertex set $\bigcup_{i\in [t]}A_i \cup \bigcup_{i\in [t]}B_i$ and with the edge set $E(A,B)$ as 
well as the edges connecting two consecutive (under $\prec$) vertices. 
Now we are allowed to contract each part $A_i$ and $B_j$ with the newly added edges. This creates $K_{t,t}$ as a minor with $t=3$, a contradiction 
to the planarity of $H$.

Suppose $L_a\prec L_b$. We create two copies of $G$, and from the first copy of $G$ we create a new planar graph as above $H_a$ 
whose vertex set is $\bigcup_{i\in [t]}A_i$. Likewise, from the second copy of $G$ we create $H_b$ whose vertex set is $\bigcup_{i\in [t]}B_i$. 
By flipping the embedding of $H_b$, one can add back the edge set $E(A,B)$ between the vertex sets of $H_a$ and $H_b$ while 
maintaining the planarity. Again, the facial cycles superposed on $\bigcup_{i\in [t]}A_i$ and $\bigcup_{i\in [t]}B_i$ allow us to 
contract each part $A_i$ and $B_j$ and obtained $K_{3,3}$ as a minor, a contradiction. This proves that $\adj{\prec}{G}$ has bounded twin-width, 
and complete the proof of~\cref{prop:counterplanar}.

We remark on how to extend $\prec$ when $G$ is enhanced by adding a universal facial vertex for each facial cycle $C\in \C$. 
Let $v_C$ be a facial vertex associated with $C\in \C$. While constructing~$T$, when vertices of $C$ are marked \reserved($v$) 
we also mark $v_C$ \reserved($v$). When $v$ is explored and $V(C)\setminus v$ are discovered, $v_C$ is discovered in the same batch 
and put to the queue while setting $\parent(v_C)=v.$ Therefore in $T$, all facial vertices will be a leaf whose parent is the first discovered vertex 
in the corresponding facial cycle. It is easy to see that this modification does not affect any part of the subsequent arguments. 
This modification is summarized in the statement below.

\begin{proposition}\label{prop:augplanar}
Let $G$ be a planar graph with a fixed planar embedding $\Sigma$ and $\C$ be a collection of pairwise vertex-disjoint facial cycles of $G$. 
Let $G^\C$ be a planar graph obtained from $G$ by adding a vertex $v_C$ for each $C\in \C$ and making it adjacent with all vertices of $V(C)$.
Then there exists a vertex ordering $\prec$ on $V(G^\C)$ with the following properties.
\begin{enumerate}[(i)]
\item $\adj{\prec}{G^\C}$ has bounded twin-width, and
\item for every $C\in \C$, the vertex ordering $\prec$ restricted to $V(C)$ is a counter-clockwise transversal order along the facial cycle $C$ under $\Sigma$.
\end{enumerate}
\end{proposition}


\bigskip

\noindent{\bf Augmenting chosen faces with nicely aligned matchings.} 
We have so far shown that $(V, E ,\prec)$ have bounded twin-width, where $G^\C$ is a planar graph enhanced by facial vertices as in~\cref{prop:augplanar}, $V=V(G^\C)$ and $E=E(G^\C)$. 
We turn to prove that the matching-augmented planar graphs (or their enhancements)
have bounded twin-width with the vertex ordering $\prec$ if the matching is nicely aligned with respect to some facial cycle packing $\C$. 

\begin{lemma}
Let $G$ be a plane graph and $\C$ be a collection of pairwise vertex-disjoint facial cycles of $G$.
Suppose that $M\subseteq V(G)\times V(G)$ is a pairwise disjoint vertex pairs 
which is nicely aligned with respect to $\C$. 
Then there exists a vertex ordering $\prec$ on $V:=V(\G^\C)$ such that 
$(V,E\cup M, \prec)$ has bounded twin-width, where $E:=E(G^\C)$.
\end{lemma} 
\begin{proof}
We use the vertex ordering $\prec$ from~\cref{prop:augplanar} for which $(V,E,\prec)$ has bounded twin-width. It suffices to show 
that  $(V,M,\prec)$ has bounded twin-width because then the 2-edge-colored graph $(V, E\cup M, \prec)$ have bounded twin-width by~\cref{lem:union}. 
Let $\adj{\prec}{H}$ denote the adjacency matrix of the ordered graph $H=(V,M,\prec)$. 
From the construction of~$\prec$, note that for every $C\in \C$ and the first discovered vertex $v$ of $C$, 
the vertices of $V(C)\setminus v$ appears continuously in $\prec$ while $v$, which precedes $V(C)\setminus v$ 
may not immediately precede $V(C)\setminus v$. Moreover, if $u$ and $v$ are the first discovered vertices of 
facial cycles $C_u$ and $C_v$ respectively and $u\prec v$, then $V(C_u)\setminus u \prec V(C_v)\setminus v$ 
due to the first-come-first-served nature of 
the discovery order $\prec$. 

Consider an arbitrary off-diagonal submatrix $M'$ of $\adj{\prec}{H}$ which is above the diagonal of $\adj{\prec}{H}$. 
Let us call a principal submatrix of $\adj{\prec}{H}$ \emph{diagonal matching block} if its rows and columns correspond 
to the vertex set $V(C)\setminus v$ for some $C\in \C$ and the first discovered vertex $v$ of $C$. 
It is clear that $M'$ can overlap with at most one diagonal matching block. 
Non-zero entries of $M'$ which are not contained in the intersection with a diagonal matching block 
correspond to edges of the form $vw$, where $v$ is the first discovered vertex of some $C\in \C$ 
and $w\in V(C)\setminus v$. Moreover, the $1\times \abs{V(C)\setminus v}$ submatrix of $M'$ 
contains at most one non-zero entry. Now, due to the property mentioned in the first paragraph, 
the non-zero entries of $M'$ which is not contained in the intersection with a diagonal matching block occupies a monotone sequence of positions of $M'$. 
Therefore, $M'$ is a union of two ordered binary relational structures (over the same universe and with the same order) encoded as matrices, 
one corresponding to the intersection with a diagonal matching block, 
another corresponding to a monotone sequence of non-zero entries. 
Both of them trivially have bounded twin-width, and thus $M'$ has bounded twin-width as well by~\cref{lem:union}. 
Then~\cref{thm:canonicaler} implies that $\adj{\prec}{H}$ has bounded twin-width, thus we conclude that $\adj{\prec}{H}$ has bounded twin-width. 
\end{proof}

This Lemma combined with the FO-transduction from matching-augmented planar graphs to $K_{t,t}$-free segment graphs complete the proof of~\cref{thm:biclique-segment}.

\subsection{Axis-parallel $H_t$-free unit segment graphs have bounded twin-width}

To show what the naming of the section claims, we face, like in the previous section, the challenging task of finding a  ``good'' linear order on objects from a two-dimensional space.

We place a virtual grid whose cells are of size $1 \times 1$, and cut the segments along this grid, adding some junction vertices in between the cut pieces corresponding to the same segment.
We first prove by FO transduction that if this new graph has bounded twin-width, then the original segment graph has bounded twin-width.
For the newly built graph, a natural order consists of locally enumerating the segments counter-clockwise according to where they cross the grid, and globally enumerating the cells of the grid row by row.
Note that the dimension of the grid cells imposes that every segment crosses the grid.

The crux is then to argue that the circular order along the boundary of a cell yields adjacency matrices with bounded grid rank.
Somewhat surprisingly this part leverages the same argument as we will later use for $H_t$-free visibility graphs of terrains; a forbidden pattern like the Order Claim (see~\cref{subsec:terrains}).
We write the lemmas in greater generality that what we actually need, for them to be possibly used in future work (in different settings).

\medskip

A~\emph{$t$-splitting of a graph $G$} replaces each of its vertices $v$, by a path $P_v: v_1v_2 \ldots v_{t'}$ on at least one and at most~$t$ vertices.
The edges incident to $v$ are then distributed between the~$v_i$'s.
(Note that $P_v$ may consist of a single vertex, in case no modification occurs at $v$.)
A~\emph{1-subdivided $t$-splitting} further subdivides exactly once each edge of each substituted $P_v$.
The~\emph{length} of a (1-subdivided) $t$-splitting is simply $t$.

The following lemma says that 1-subdivided splittings of bounded length can be undone by a~transduction.
\begin{lemma}
  Let $t$ be a fixed positive integer.
  There is an FO transduction $\mathsf T_t$ such that for every graph $G$, and every $G'$, a~1-subdivided $t$-splitting of $G$, it holds that $G \in \mathsf T_t(G')$. 
\end{lemma}
\begin{proof}
  Let $G'$ be a~1-subdivided $t$-splitting of $G$.
  For every $P_v: v_1v_2 \ldots v_{t'}$ in $G'$, the vertex $v_1$ is marked by a unary relation $U$, while every degree-2 vertices resulting from the 1-subdivision of $P_v$ are marked with a unary relation $D$.
  Transduction $\mathsf T_t$ defines a new domain by $\nu(x) \equiv U(x)$ and a new edge set by $$\varphi(x,y) \equiv U(x)~\land~U(y)~\land~\exists x' \exists y'~E(x',y')~\land~P(x,x')~\land~P(y,y')$$
  where $P(x,y)$ is $$x=y~\lor~\bigvee_{s \in [t-1]} \exists d_1 \exists x_1 \exists d_2 \exists x_2 \ldots \exists d_s~E(x,d_1)~\land~E(d_s,y)~\land~$$ $$\bigwedge_{i \in [s-1]} D(d_i)~\land~E(d_i,x_i)~\land~E(x_i,d_{i+1}).$$
  As the only paths in $G'$ alternating between vertices in and out of $D$ are subpaths of some~$P_v$, we get that $G \in \mathsf T_t(G')$. 
\end{proof}

In particular, using~\cref{thm:transduction}, if a class of 1-subdivided $O(1)$-splittings have bounded twin-width, then the class of the original graphs have bounded twin-width, too.
\begin{lemma}\label{lem:splitting-bdtww}
  Let $t$ be a fixed positive integer, and $\mathcal C$ be a graph class.
  Let $\mathcal D$ be a class which contains for every $G \in \mathcal C$, a~1-subdivided $t$-splitting of $G$.
  Then $\mathcal C$ has bounded twin-width if $\mathcal D$ has bounded twin-width. 
\end{lemma}

A \emph{$t$-grid} is an infinite planar square grid (or square tessellation) whose every cell dimension is $t \times t$.
We may call 1-grids \emph{unit grids}.
A~grid is a $t$-grid for some $t$, and we may refer to the cells of $\Gamma$ as the \emph{$\Gamma$-cells}.
A~\emph{clip} of a family $\mathcal F$ of objects in the plane by a grid $\Gamma$, or \emph{$\Gamma$-clip} of $\mathcal F$, is any multiset $C_f = \{\{x \cap f~:~x \in \mathcal F\}\}$ where $f$ is a face (or cell) of $\Gamma$.
We denote by $y^{\mathcal F}$ the object of $\mathcal F$ giving rise to (the particular copy of) $y \in C_f$. 
We refer to $C_f$ as the \emph{$f$-clip} of~$\mathcal F$.

Let $\partial f$ be the boundary of $f$.
For every $x, y \in C_f$, denote $x \prec_\circ y$ if a point of $x \cap \partial f$ is met before a point of $y \cap \partial f$ when starting at the top-left corner of $f$ and going counter-clockwise along $\partial f$.
Note that if $(1)$ every object of $C_f$ intersects $\partial f$, and $(2)$ every pair of objects of $C_f$ intersect $\partial f$ at disjoint sets, then $\prec_\circ$ is a total order.
For our purposes, we will manage to always enforce $(1)$, but not always $(2)$.
As we need to build total orders, we will break ties arbitrarily (but consistently) with any total order $\prec_{\mathcal F}$ over $\mathcal F$.
Then the \emph{circular order of $f$ over $C_f$}, denoted by $\prec_f$, is a total order defined by: $x \prec_f y$ if and only if $x \prec_\circ y \lor (y \not\prec_\circ x \land x' \prec_{\mathcal F} y')$ where $x'=x^{\mathcal F}$ and $y'=y^{\mathcal F}$.
We say that $x \in C_f$ is \emph{rooted at a side $B$} of $f$, if the first point $p$ of $x \cap \partial f$ along the circular order is on $B$.
We may also write that \emph{$x$ is rooted at $p$}.

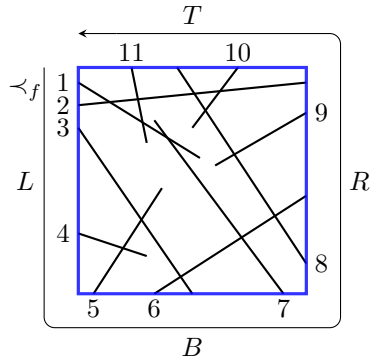
\begin{figure}[h!]
  \centering
  \begin{tikzpicture}
    \foreach \i/\j in {{0,2.8}/{1.6,1.8}, {0,2.5}/{3,2.8}, {0,2.2}/{1.5,0}, {0,0.8}/{0.9,0.5}, {0.2,0}/{1.1,1.4}, {1,0}/{3,1.3}, {2.7,0}/{1,2.3}, {3,0.4}/{1.3,3},
      {3,2.4}/{1.8,1.7}, {2.1,3}/{1.5,2.2}, {0.7,3}/{0.9,2}}{
      \draw[thick] (\i) -- (\j) ;
    }
    \foreach \i/\j/\l in {-0.2/2.8/1, -0.2/2.5/2, -0.2/2.2/3, -0.2/0.8/4, 0.2/-0.2/5, 1/-0.2/6, 2.7/-0.2/7, 3.2/0.4/8, 3.2/2.4/9, 2.1/3.2/10, 0.7/3.2/11}{
      \node at (\i,\j) {\l} ;
    }
    \draw[opacity=0.8,very thick,blue] (0,3) -- (0,0) -- (3,0) -- (3,3) -- cycle ;

    \draw[rounded corners] (-0.45,3) -- (-0.45,-0.45) -- (3.45,-0.45) -- (3.45,3.45) -- (0.5,3.45) ;
    \draw[-stealth] (0.5,3.45) -- (0,3.45) ;

    \foreach \i/\j/\l in {-0.7/1.5/L, 1.5/-0.7/B, 3.7/1.5/R, 1.5/3.7/T, -0.7/2.7/{\prec_f}}{
      \node at (\i,\j) {$\l$} ;
    }
  \end{tikzpicture}
  \caption{The $f$-clip of a set $S$ of segments, with the circular order $\prec_f$. Four segments are rooted at $L$ (1 to 4), three at $B$ (5 to 7), two at $R$ (8 and 9), and two at $T$ (10 and 11).}
  \label{fig:circular-order}
\end{figure}

Grid $\Gamma$ is in \emph{general position with respect to $S$} if no segment endpoint or proper crossing within $S$ lies at the boundary of a $\Gamma$-cell.
Note that the general position still allows two or more segments and a $\Gamma$-cell boundary to have a common intersection, provided the segments are collinear.
A grid $\Gamma$ \emph{hits} a set of segments $S$ if every segment of $S$ intersects (the 1-skeleton of)~$\Gamma$. 

\begin{figure}[h!]
  \centering
  \begin{tikzpicture}
    \def\s{0.2}
    \foreach \i in {1,2}{
    \begin{scope}[yshift=1.5 * \s cm, xshift=-4.5 * \i cm]
     \foreach \i/\j in {{0.2,0.1}/{1.4,1.2},{0.3,2.1}/{1.5,3.5},{1.8,2.8}/{0.4,3.9},{2.3,1.1}/{1.8,3.1},{1.9,0.1}/{2.5,1.2},{3.1,2.4}/{3.8,0.7},
      {1.4,2.7}/{2.8,2.2},{1.3,2.5}/{3.5,2.8},{2.4,3.6}/{3.2,1.2},{3.8,2.2}/{2.9,3.6},{3.1,1.7}/{2.6,3.8},{3.7,1.3}/{2.1,0.8},
      {0.1,2.2}/{0.3,3.5},{1.1,1.5}/{1.5,2.9},{3.4,0.6}/{3.2,1.8},{1.6,0.4}/{2.9,0.5},{1.4,1.9}/{3.6,2.8},{0.7,1.7}/{2.4,1.2},
      {1.4,0.9}/{0.1,2.5}}{
      \draw[thick] (\i) -- (\j) ;
     }
    \end{scope}
    }
    
      \begin{scope}[yshift=1.5 * \s cm, xshift=-4.5 cm]
     \foreach \i in {0,...,4}{
      \draw[opacity=0.8,very thick,blue] (-0.1,\i) -- (4.1,\i) ;
      \draw[opacity=0.8,very thick,blue] (\i,-0.1) -- (\i,4.1) ;
     }
     \end{scope}

    \foreach \i [count = \ip] in {0,...,3}{
      \foreach \j [count = \jp] in {0,...,3}{
    \begin{scope}[xshift=\i * \s cm, yshift=\j * \s cm]
      \clip (\i,\j) rectangle (\ip,\jp) ;
    \foreach \i/\j in {{0.2,0.1}/{1.4,1.2},{0.3,2.1}/{1.5,3.5},{1.8,2.8}/{0.4,3.9},{2.3,1.1}/{1.8,3.1},{1.9,0.1}/{2.5,1.2},{3.1,2.4}/{3.8,0.7},
      {1.4,2.7}/{2.8,2.2},{1.3,2.5}/{3.5,2.8},{2.4,3.6}/{3.2,1.2},{3.8,2.2}/{2.9,3.6},{3.1,1.7}/{2.6,3.8},{3.7,1.3}/{2.1,0.8},
      {0.1,2.2}/{0.3,3.5},{1.1,1.5}/{1.5,2.9},{3.4,0.6}/{3.2,1.8},{1.6,0.4}/{2.9,0.5},{1.4,1.9}/{3.6,2.8},{0.7,1.7}/{2.4,1.2},
      {1.4,0.9}/{0.1,2.5}}{
      \draw[thick] (\i) -- (\j) ;
    }
    \foreach \i in {0,...,4}{
      \draw[opacity=0.8,very thick,blue] (-0.1,\i) -- (4.1,\i) ;
      \draw[opacity=0.8,very thick,blue] (\i,-0.1) -- (\i,4.1) ;
    }
    \end{scope}
      }
      }
    
  \end{tikzpicture}
  \caption{Left: A set $S$ of segments. Center: Same set with a grid $\Gamma$ hitting $S$ and in general position w.r.t. $S$. Right: The 16 non-empty $\Gamma$-clips of $S$.}
  \label{fig:clips}
\end{figure}
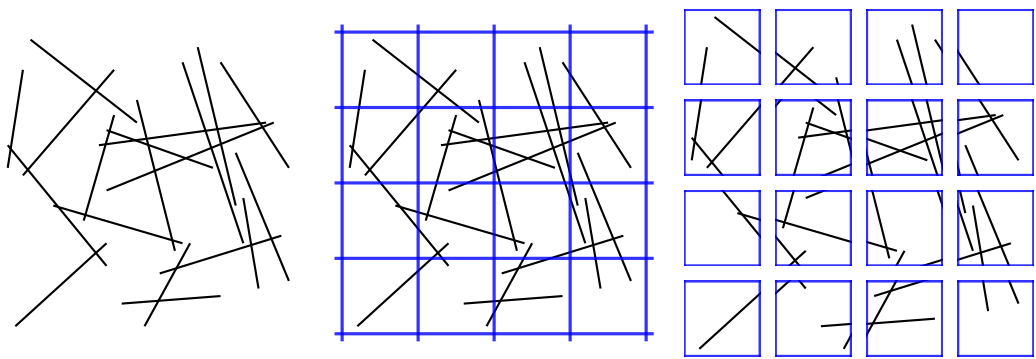

Let $S$ be a set of segments and $\Gamma$ be a grid hitting $S$, and in general position with respect to $S$.
The $\Gamma$-splitting of $S$ is the graph $H(S,\Gamma)$ constructed as follows.
First, every segment $s \in S$ is partitioned into $s_1,x_1,s_2,x_2,s_3,\ldots,s_{\ell-1},x_{\ell-1},s_\ell$ such that the $x_i$'s (with $i \in [\ell-1]$) are the intersection points of $s$ with (the skeleton of) $\Gamma$, for every $i \in [2,\ell-1]$, $s_i$ is the open segment between $x_{i-1}$ and $x_i$, and $s_1, s_\ell$ are the remaining segment parts, respectively incident to $x_1$ and $x_\ell$.
The graph $H(S,\Gamma)$ is then the intersection graph of the $s_i$'s for every $s \in S$, where one vertex $d_j$ is added per $x_j$, and made adjacent to the vertices of $s_j$ and $s_{j+1}$.
In particular, even if two segments $s, s' \in S$ intersect at the same point $x$ of the grid, the vertices $d_j$ and $d'_{j'}$ associated to $x$ remain of degree~2.
We call $s_1,s_2,\ldots$ \emph{$s$-vertices} or \emph{segment vertices}, and $d_1,d_2,\ldots$, \emph{$d$-vertices} or \emph{dividing vertices}. 

Note that $H(S,\Gamma)$ is a~1-subdivided splitting of the intersection graph of~$S$.
It is further a~1-subdivided $t$-splitting if every segment of $S$ is contained in at most $t$ $\Gamma$-cells.
Even if $H(S,\Gamma)$ is an abstract graph without a segment representation, it will help to often identify its vertices to segments within a cell ($s$-vertices) and points at the boundary of a cell ($d$-vertices). 

\begin{lemma}\label{lem:d-vs-s}
  Let $S$ be a set of segments and $\Gamma$ be a grid in general position.
  Let $B$ be a side of a $\Gamma$-cell $f$.
  Let $X$ be the $d$-vertices on $B$, and let $Y$ be the $s$-vertices of $f$ in $H = H(S,\Gamma)$ restricted to $f$.
  Assume that the intersection graph of the $f$-clip of $S$, along the circular order~$\prec_f$, has grid rank at most~$k$.
  Then $\adj{\prec}{H \langle X,Y \rangle}$ has grid rank bounded by a function of~$k$, where $\prec$ orders $Y$ by $\prec_f$, and $X$ by $\prec_f$ or its mirror.
  
  (The tie-break is imposed by any total order over $S$, thus is consistent between $X$ and $Y$.)
\end{lemma}
\begin{proof}
  Let $q$ be the maximum integer such that $\adj{\prec}{H \langle X,Y \rangle}$ has a rank-$q$ division, and let $\mathcal D$ be such a division.
  Let $Y_1 \subseteq Y$ be the $s$-vertices touching both $B$ and a side of $f$ that comes before $B$ along~$\prec_f$, and let $Y_2 = Y \setminus Y_1$.
  By definition of~$\prec_f$ (with the consistent tie-break), the 1 entries of $\adj{\prec}{H \langle X,Y_2 \rangle}$ form a strictly increasing sequence.
  Hence $\mathcal D$ contains a rank-$\lceil q/2 \rceil$ division $\mathcal D'$ without any 1 entry from $H\langle X,Y_2 \rangle$.

  The 1 entries of this division $\mathcal D'$ then comes from $s$-vertices \emph{not} rooted at $B$ adjacent to $d$-vertices on $B$.
  Besides, $\adj{\prec}{H \langle X,Y_1 \rangle}$ is the permutation matrix associated to the permutation graph $H[Y_1]$.
  It was shown that the permutation matrices $M$ associated to permutation graphs $P$ have bounded grid rank if and only if so do have the adjacency matrices of $P$ along the natural order (see~\cite[Section 6.1]{twin-width1} and \cite{twin-width4}).
  By assumption the permutation graph $H[Y_1]$ along $\prec_f$ has grid rank at most~$k$.
  Therefore $\gr(\adj{\prec}{H \langle X,Y_1 \rangle})$, and hence $q$, are bounded by a function of $k$.
\end{proof}

The $\Gamma$-bounding box of $S$ is the minimal rectangular region of $\Gamma$-cells containing $S$.
For any $\ell \geqslant 1$, a~\emph{$(1,\ell)$-segment} is a segment of length at least 1 and at most $\ell$.
We are now ready to show the main technical lemma of this section. 

\begin{lemma}\label{lem:rooted-to-unit}
  Let $\ell \geqslant 1$ and $k$ be fixed integers. 
  Let $\mathcal C$ be a class of graphs representable as the intersection graph of a set $S$ of $(1,\ell)$-segments such that every grid $\Gamma$ hitting $S$, and in general position with respect to $S$, satisfies the following.
  For every face $f$ of $\Gamma$, the adjacency matrix of the intersection graph of the $f$-clip of $S$ along the circular order of $f$ has grid rank at most~$k$.
  Then $\mathcal C$ has bounded twin-width.
\end{lemma}
\begin{proof}
  Let $H=H(\Gamma,S)$ be the $\Gamma$-splitting of $S$, where $\Gamma$ satisfies the preconditions of the lemma. 
  We show that $H$ has bounded twin-width (in $\ell$ and $k$), and conclude with~\cref{lem:splitting-bdtww}.

  Let $\{C_{i,j}~:~i \in [N],~j \in [M]\}$ be the set of cells of the $\Gamma$-bounding box of $S$, indexed by their relative $x$- and $y$-coordinates.
  Let $\prec$ be the following total order of $V(H)$, given by listing the vertices of $V(H)$ from smallest to largest. 
  For $i$ going from 1 to $N$, for $j$ going from 1 to $M$, we list the $s$-vertices of the $C_{i,j}$-clip of $S$ and the $d$-vertices (that were not previously listed) on $\partial C_{i,j}$ along the circular order $\prec_{C_{i,j}}$, with tie-break imposed by an arbitrary but fixed order $\prec_S$ of $S$.
  More specifically, as one goes counter-clockwise along $\partial C_{i,j}$, starting at the top-left corner of $C_{i,j}$, every time a $d$-vertex is met, it is appended to the order, and immediately followed by its corresponding $s$-vertex in $C_{i,j}$.
  When several $d$-vertices ``occupy'' the same point of $\partial C_{i,j}$, their relative order is imposed by $\prec_S$. 
  We denote by $Z_{i,j}$ the vertices that are listed at iteration $(i,j)$. 

  We now show that $M = \adj{\prec}{H}$ has bounded grid rank.
  For every $i \in [N],~j \in [M]$, let $X_{i,j}$ (resp~$Y_{i,j}$) be the row part (resp.~column part) of $M$ corresponding to $Z_{i,j}$.
  Let $\mathcal D$ be the corresponding division of $M$, that is, with cells $X_{i,j} \cap Y_{i',j'}$ for $i,i' \in [N]$ and $j,j' \in [M]$.

  We observe that only 5 ``diagonals'' of zones $X_{i,j} \cap Y_{i',j'}$ are not full 0: the main diagonal, the first superdiagonal and subdiagonal, and the $M$-th superdiagonal and subdiagonal.
  We thus claim that a rank-$11(q+1)$ division $\mathcal D'=(\mathcal D'^R,\mathcal D'^C)$ of $M$ induces a rank-$q$ of a submatrix $X_{i,j} \cap Y_{i',j'}$.

  Indeed assume the dividing lines of $\mathcal D'^R$ stabs 12 distinct row parts of $\mathcal D$ (see~\cref{subsec:rd} for the corresponding definitions).
  Then there is a coarsening of $\mathcal D'^R$ of six parts, each entirely containing a row part of $\mathcal D$.
  Now no part of $\mathcal D'^C$ can be entirely contained in a column part of $\mathcal D$.
  Otherwise its intersection with at least one of these six parts lies in a~full 0 zone, contradicting that $\mathcal D'$ is a rank division.
  Thus the dividing lines of $\mathcal D'^C$ stabs 12 distinct column parts.
  Therefore there is a coarsening of $\mathcal D'^C$ (resp.~$\mathcal D'^R$) of twelve column parts $C_1, \ldots, C_{12}$ (row parts $R_1, \ldots, R_{12}$) such that every dividing line is in a distinct part of~$\mathcal D$.
  Removing every other dividing line (in row and column), we get a 6-division where each cell fully contains a non full 0 cell of~$\mathcal D$.
  In particular, the anti-diagonal of this 6-division (displaying 6 non-empty cells of $\mathcal D$ in a strictly decreasing pattern) is a contradiction to $\mathcal D$ having only 5 non-empty diagonals.

  Thus at most 11 row parts of $\mathcal D$ can be stabbed by the dividing lines of~$\mathcal D'^R$.
  Symmetrically the same reasoning shows that at most 11 column parts of $\mathcal D$ can be stabbed by the dividing lines of~$\mathcal D'^C$.
  Hence there is a zone $X_{i,j} \cap Y_{i',j'}$ of $\mathcal D$ stabbed by at least $q+1$ dividing lines of $\mathcal D'^R$ and at least $q+1$ dividing lines of~$\mathcal D'^C$.
  In particular, $X_{i,j} \cap Y_{i',j'}$ has a rank-$q$ division.

  By~\cref{lem:d-vs-s}, every (non-empty) non-diagonal cell $X_{i,j} \cap Y_{i',j'}$ has bounded grid rank.
  Indeed, $X_{i,j} \cap Y_{i',j'}$ deprived of its full 0 rows and columns corresponds to a biadjacency matrix between $s$-vertices of a cell $f$, and $d$-vertices of an adjacent cell $f'$, where the $s$-vertices are ordered by $\prec_f$, and the $d$-vertices are ordered by $\prec_{f'}$.
  (This is where it is important that the arbitrary tie-breaker $\prec_S$ is consistent across the different $\Gamma$-cells.)
  In particular, $\prec_{f'}$ orders the $d$-vertices along $\prec_f$ or its mirror, hence we can use~\cref{lem:d-vs-s}.

  We finally show that every diagonal cell $X_{i,j} \cap Y_{i,j}$ has bounded grid rank.
  Let $L, B, R, T$ be a partition of $Z_{i,j}$ into $s$-vertices rooted at, and $d$-vertices lying on, the left, bottom, right, and top sides of $C_{i,j}$, respectively.
  We shall just bound the grid rank of $\adj{\prec}{H \langle X,Y \rangle}$ for every $X,Y \in \{L,B,R,T\}$.
  This holds for $X=Y$ since then, $H \langle X,Y \rangle$ is an induced subgraph of a circularly-ordered adjacency matrix of a $\Gamma$-clip, augmented by rows and columns (the $d$-vertices on $X$) with 1 entries forming a strictly monotone sequence.

  For $X \neq Y$, assume without loss of generality, that $X$ comes before $Y$ in the circular order.
  Then observe that there is no edge between the $d$-vertices of $X$ and the $s$-vertices of $Y$.  
  Let $Y_s$ be the $s$-vertices of $Y$, and $Y_d$ be its $d$-vertices.
  $\adj{\prec}{H \langle X,Y_s \rangle}$ has bounded grid rank by the hypotheses of the current lemma, and $\adj{\prec}{H \langle X,Y_d \rangle}$ has bounded grid rank by~\cref{lem:d-vs-s}.
  Thus, by~\cref{lem:x-vs-yz}, $H \langle X,Y \rangle$ has bounded grid rank.
  
  Hence overall, the grid rank of $\adj{\prec}{H}$ is bounded (by a function of $\ell$ and $k$), and we conclude by \cref{thm:mixed-number-gen2}. 
\end{proof}

We say that a collection $S$ of segments is \emph{square-rooted} if $S$ is an $f$-clip (for $f$, a square grid cell) such that every segment of $S$ intersects $\partial f$ at least once.
We say that \emph{square-rooted} segments are \emph{axis-aligned}, if both the segments and the square grid cell are axis-parallel. 
In the next lemma, we show that the circular order is a ``good order'' for $H_t$-free square-rooted axis-aligned segments.
By~\cref{lem:rooted-to-unit} this implies that $H_t$-free axis-parallel unit segment (even $(1,\ell)$-segment) graphs have bounded twin-width. 

\begin{theorem}\label{thm:ht-free-uap}
  Let $\ell$ be a fixed integer.
  $H_t$-free axis-parallel $(1,\ell)$-segment graphs have bounded twin-width.
  Furthermore, if a geometric representation is given, $O(1)$-sequences can be computed in polynomial time.
\end{theorem}
\begin{proof}
  Let $S'$ be any finite set of axis-parallel $(1,\ell)$-segment whose intersection graph is $H_t$-free. 
  Let $\Gamma$ be an axis-parallel unit grid (axis-aligned with the segments) in general position with respect to $S'$.
  This can always be obtained by slightly moving $\Gamma$.
  As the segments of $S'$ are axis-parallel, the intersection graph $G$ of any $\Gamma$-clip of $S'$ is also $H_t$-free.
  (Observe that this is not necessarily true if the segments can have other slopes.) 
  We show that the circularly-ordered adjacency matrix of $G$ has bounded grid rank, and conclude with~\cref{lem:rooted-to-unit}.

  Let $S$ be the $\Gamma$-clip of $S'$ of square-rooted axis-aligned segments representing $G$.
  We materialize by $(0,1)(0,0)(1,0)(1,1)$ the square where $S$ is rooted.
  As the segments of $S'$ are axis-aligned and of length at least 1, every segment of $S$ intersects $(0,1)(0,0)(1,0)(1,1)$ at least once (thus applying~\cref{lem:rooted-to-unit} is possible).
  
  Let $L, B, R, T$ be a partition of $S$, that we identify to $V(G)$, such that every segment in $L$ is rooted at the left side $(0,1)(0,0)$, in $B$, rooted at the bottom side $(0,0)(1,0)$, in $R$, rooted at the right side $(1,0)(1,1)$, and in $T$, rooted at the top side $(1,1)(0,1)$.
  
  Let $\prec$ be the circular order of the cell $(0,1)(0,0)(1,0)(1,1)$ on $S$ (starting at $(0,1)$), with tie-break imposed by any total order over $S'$. 
  We assume, for the sake of contradiction, that $\adj{\prec}{G}$ has grid rank at least $4 f(\lceil \sqrt t \rceil)$, where $f$ is the function of~\cref{thm:rd-to-up}.
  This implies that $\adj{\prec}{G\langle X,Y \rangle}$ has grid rank at least $f(\lceil \sqrt t \rceil)$ for $X,Y \in \{L,B,R,T\}$.
  We first observe that $\adj{\prec}{G\langle X,Y \rangle}$ has grid rank less than 2 for $\{X,Y\} \in \{\{L,L\},\{L,R\},\{R,R\},\{B,B\},\{B,T\},$ $\{T,T\}\}$.
  Indeed, in that case the 1 entries of $\adj{\prec}{G\langle X,Y \rangle}$ can be partitioned into rectangles whose projections (horizontal and vertical) are pairwise disjoint.
  (Recall that the segments are axis-aligned, hence the corresponding graphs are disjoint unions of bicliques.)

  By symmetry, we can thus assume that $\adj{\prec}{G\langle B,L \rangle}$ has grid rank at least $f(\lceil \sqrt t \rceil)$.
  Observe that for every $a,b \in L$ and $c,d \in B$ such that $a \prec b$ and $c \prec d$, if $ac, bd \in E(G)$ then $ad \in E(G)$; see~\cref{fig:forced-BL}.
  Note also that this property holds even when $a$ and $b$ overlap, and/or when $c$ and $d$ overlap, regardless how two overlapping horizontal (resp.~vertical) segments are relatively ordered.
  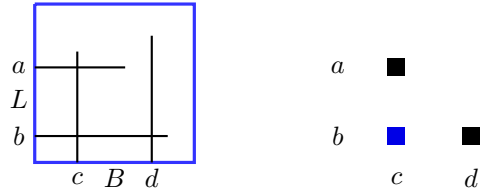
\begin{figure}[h!]
    \centering
    \begin{tikzpicture}[scale=0.7]
      \draw[opacity=0.8,very thick,blue] (0,3) -- (0,0) -- (3,0) -- (3,3) -- (0,3) ;
      \begin{scope}[thick]
      \draw (0.8,0) -- (0.8,2.1) ;
      \draw (2.2,0) -- (2.2,2.4) ;
      \draw (0,0.5) -- (2.5,0.5) ;
      \draw (0,1.8) -- (1.7,1.8) ;
      \end{scope}
      \node at (-0.3,1.8) {$a$} ;
      \node at (-0.3,0.5) {$b$} ;
      \node at (0.8,-0.3) {$c$} ;
      \node at (2.2,-0.3) {$d$} ;

      \node at (1.5,-0.3) {$B$} ;
      \node at (-0.3,1.2) {$L$} ;

      \begin{scope}[xshift=6cm]
      \node at (-0.3,1.8) {$a$} ;
      \node at (-0.3,0.5) {$b$} ;
      \node at (0.8,-0.3) {$c$} ;
      \node at (2.2,-0.3) {$d$} ;

      \node[fill] at (0.8,1.8) {} ;
      \node[fill] at (2.2,0.5) {} ;
      \node[fill,black!10!blue] at (0.8,0.5) {} ;
      \end{scope}
    \end{tikzpicture}
    \caption{Left: Four segments $a \prec b \in L, c \prec d \in B$ satisfy that if $a$ and $c$ cross, and $b$ and $d$ cross, then $b$ and $c$ also cross. Right: In $\adj{\prec}{G\langle B,L \rangle}$, the 1 entries at $(a,c)$ and $(b,d)$ imply a 1~entry at $(b,c)$.}
    \label{fig:forced-BL}
  \end{figure}

  For $s \in \{0,1,\uparrow,\rightarrow\}$, $\mat{2}{s}$ contains the pattern $\left( \begin{smallmatrix} 1&? \\ 0&1 \end{smallmatrix} \right)$ forbidden by~\cref{fig:forced-BL}, thus impossible in $\adj{\prec}{G\langle B,L \rangle}$. 
  
  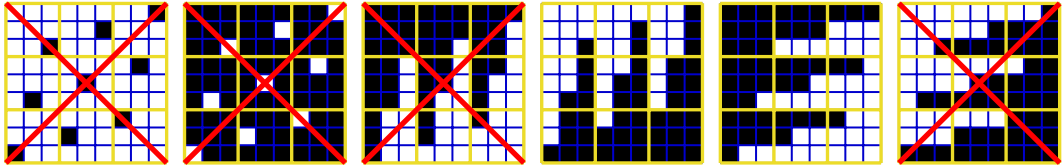
\begin{figure}[h!]
  \centering
    \begin{tikzpicture}[scale=.235]

   \foreach \symb/\b/\xsh/\ysh/\r in {{==}/0/-30/-10/100, {>}/0/-10/-10/100,{<}/0/0/-10/0,{<}/1/10/-10/0,{>}/1/20/-10/100}{
   \begin{scope}[xshift=\xsh cm,yshift=\ysh cm]     
    \foreach \i in {0,...,8}{
      \foreach \j in {0,...,8}{
        \pgfmathsetmacro{\ip}{\i+1}
        \pgfmathsetmacro{\jp}{\j+1}
        \pgfmathsetmacro{\col}{ifthenelse(\i == mod(\j,3)*3+floor(\j/3),"black",ifthenelse(\b==1,ifthenelse(\i \symb mod(\j,3)*3+floor(\j/3),"black","white"),ifthenelse(\j \symb mod(\i,3)*3+floor(\i/3),"black","white")))}
        \fill[\col] (\i,\j) -- (\i,\jp) -- (\ip,\jp) -- (\ip,\j) -- cycle;
      }
    }
    \draw[line width=0.75pt, scale=1, color=black!20!blue] (0, 0) grid (9, 9);
    \draw[line width=1.25pt, scale=3, color=black!10!yellow] (0, 0) grid (3, 3);

    \draw[red,line width=2pt,opacity=\r] (0,0) -- (9,9) ;
    \draw[red,line width=2pt,opacity=\r] (0,9) -- (9,0) ;
   \end{scope}
   }
   \foreach \symb/\b/\xsh/\ysh/\r in {{!=}/0/-20/-10/100}{
   \begin{scope}[xshift=\xsh cm,yshift=\ysh cm]     
    \foreach \i in {0,...,8}{
      \foreach \j in {0,...,8}{
        \pgfmathsetmacro{\ip}{\i+1}
        \pgfmathsetmacro{\jp}{\j+1}
        \pgfmathsetmacro{\col}{ifthenelse(\i == mod(\j,3)*3+floor(\j/3),"white","black")}
        \fill[\col] (\i,\j) -- (\i,\jp) -- (\ip,\jp) -- (\ip,\j) -- cycle;
      }
    }
    \draw[line width=0.75pt, scale=1, color=black!20!blue] (0, 0) grid (9, 9);
    \draw[line width=1.25pt, scale=3, color=black!10!yellow] (0, 0) grid (3, 3);

    \draw[red,line width=2pt,opacity=\r] (0,0) -- (9,9) ;
    \draw[red,line width=2pt,opacity=\r] (0,9) -- (9,0) ;
   \end{scope}
   }
    \end{tikzpicture}
    \caption{Four of the six universal patterns are not realizable with segments rooted at $B$ and $L$.}
    \label{fig:forb-BL-rooted}
  \end{figure}
  Thus, by~\cref{thm:rd-to-up}, it should be that $\mat{\lceil \sqrt k \rceil}{s}$ is a submatrix of $\adj{\prec}{G\langle B,L \rangle}$ for an $s \in \{\downarrow,\leftarrow\}$.
  This contradicts that $G$ is $H_t$-free.
  Finally one can observe that every step described in this section is effective, as well as \cref{thm:mixed-number-gen2,thm:transduction}.
\end{proof}

\section{Visibility graphs}\label{sec:visibility}

We now deal with visibility graphs.
These classes tend to \emph{not} be hereditary. 
This is why delineation was defined with ``hereditary closure of subclasses'' and not simply with ``hereditary subclasses.''

We show that visibility graphs of terrains without arbitrarily large ladders have bounded twin-width.
In stark contrast, we exhibit a subclass of visibility graphs of simple polygons whose hereditary closure has unbounded twin-width but is monadically dependent, and even monadically stable.
Finally, we show that the twin-width of simple polygons is bounded by a function of their independence number, which, combined with the FO model checking algorithm in~\cite{twin-width1}, generalizes a conjecture of Hliněný, Pokrývka, and Roy~\cite{hlinveny2019fo}, and shows in particular that \textsc{$k$-Independent Set} is FPT on visibility graphs of simple polygons.

\subsection{$H_t$-free visibility graphs of 1.5D terrains have bounded twin-width}\label{subsec:terrains}

Let $p_1, p_2, \ldots, p_n$ be the vertices of the terrain ordered by the $x$-monotone polygonal chain, i.e., such that $p_ip_{i+1}$ is an edge of the terrain for every $i \in [n-1]$.
We denote by $\prec$ the left-right order on the vertices of the terrain, i.e., $p_i \prec p_j$ holds whenever $i < j$. 
We extend the order to sets: $A \prec B$ holds when for every $a \in A$ and $b \in B$, $a \prec b$.
The following easy lemma is the sole required property to show that $H_t$-free terrain graphs have bounded twin-width.

\begin{lemma}[Order Claim~\cite{Ben-Moshe07}, see Figure~\ref{fig:orderClaim}]\label{lem:orderClaim}
If $a \prec b \prec c \prec d$, $a$ see $c$, and $b$ see $d$, then $a$ and $d$ also see each other. 
\end{lemma}

\begin{figure}[h!]
\centering
\begin{tikzpicture}
\draw[thick] (-0.3,0.8) -- (0.2,1) -- (0.7,0.6) -- (1,-0.2) -- (1.3,0.2) -- (2.1,0.2) -- (2.4,-0.2) -- (3,0) -- (4.4,-0.1) -- (5,0.4) -- (5.5,0.2) -- (6,-0.4) -- (6.3,-0.4) -- (7.1,-0.2) -- (7.7,-0.4) -- (8,0) -- (8.4,0.5) -- (9,1) -- (9.7,0.6);

\node[fill,circle,inner sep=-0.03cm] (a) at (0.2,1) {} ;
\node at (0.2,0.8) {$a$} ;

\node[fill,circle,inner sep=-0.03cm] (b) at (2.1,0.2) {} ;
\node at (2.1,0) {$b$} ;

\node[fill,circle,inner sep=-0.03cm] (c) at (4.4,-0.1) {} ;
\node at (4.4,-0.3) {$c$} ;

\node[fill,circle,inner sep=-0.03cm] (c) at (4.4,-0.1) {} ;
\node at (4.4,-0.3) {$c$} ;

\node[fill,circle,inner sep=-0.03cm] (d) at (9,1) {} ;
\node at (9,0.8) {$d$} ;

\draw[thin,dashed] (a) -- (c) ;
\draw[thin,dashed] (b) -- (d) ;
\draw[dashed,blue,thick] (a) -- (d) ;

\begin{scope}[xshift=11.5cm,yshift=-0.4cm]
\node at (0,0) {$a$} ;
\node at (0.5,0) {$b$} ;

\node at (-0.4,0.6) {$c$} ;
\node at (-0.4,1.4) {$d$} ;

\node[fill] at (0,0.6) {} ;
\node[fill] at (0.5,1.4) {} ;
\node[fill,black!10!blue] at (0,1.4) {} ;
\end{scope}

\end{tikzpicture}
\caption{Left: The Order Claim. The dashed black edges imply the dashed blue edge. Right: In the thus ordered adjacency matrix, the 1 entries at $(a,c)$ and $(b,d)$ implies the 1 entry at $(a,d)$.}
\label{fig:orderClaim}
\end{figure}
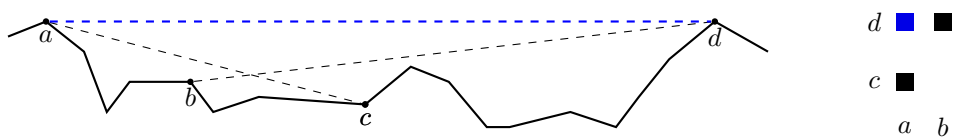

Let $G$ be the visibility graph of any terrain and $\prec$ the left-right order of its boundary.
The subsequent algorithm does not need a full geometric representation: the data of a left-right order in at least one representation is enough.

\begin{figure}[h!]
  \centering
    \begin{tikzpicture}[scale=.235]

   \foreach \symb/\b/\xsh/\ysh/\r in {{==}/0/-30/-10/100, {>}/0/-10/-10/0,{<}/0/0/-10/100,{<}/1/10/-10/0,{>}/1/20/-10/100}{
   \begin{scope}[xshift=\xsh cm,yshift=\ysh cm]     
    \foreach \i in {0,...,8}{
      \foreach \j in {0,...,8}{
        \pgfmathsetmacro{\ip}{\i+1}
        \pgfmathsetmacro{\jp}{\j+1}
        \pgfmathsetmacro{\col}{ifthenelse(\i == mod(\j,3)*3+floor(\j/3),"black",ifthenelse(\b==1,ifthenelse(\i \symb mod(\j,3)*3+floor(\j/3),"black","white"),ifthenelse(\j \symb mod(\i,3)*3+floor(\i/3),"black","white")))}
        \fill[\col] (\i,\j) -- (\i,\jp) -- (\ip,\jp) -- (\ip,\j) -- cycle;
      }
    }
    \draw[line width=0.75pt, scale=1, color=black!20!blue] (0, 0) grid (9, 9);
    \draw[line width=1.25pt, scale=3, color=black!10!yellow] (0, 0) grid (3, 3);

    \draw[red,line width=2pt,opacity=\r] (0,0) -- (9,9) ;
    \draw[red,line width=2pt,opacity=\r] (0,9) -- (9,0) ;
   \end{scope}
   }
   \foreach \symb/\b/\xsh/\ysh/\r in {{!=}/0/-20/-10/100}{
   \begin{scope}[xshift=\xsh cm,yshift=\ysh cm]     
    \foreach \i in {0,...,8}{
      \foreach \j in {0,...,8}{
        \pgfmathsetmacro{\ip}{\i+1}
        \pgfmathsetmacro{\jp}{\j+1}
        \pgfmathsetmacro{\col}{ifthenelse(\i == mod(\j,3)*3+floor(\j/3),"white","black")}
        \fill[\col] (\i,\j) -- (\i,\jp) -- (\ip,\jp) -- (\ip,\j) -- cycle;
      }
    }
    \draw[line width=0.75pt, scale=1, color=black!20!blue] (0, 0) grid (9, 9);
    \draw[line width=1.25pt, scale=3, color=black!10!yellow] (0, 0) grid (3, 3);

    \draw[red,line width=2pt,opacity=\r] (0,0) -- (9,9) ;
    \draw[red,line width=2pt,opacity=\r] (0,9) -- (9,0) ;
   \end{scope}
   }
    \end{tikzpicture}
    \caption{Four of the six universal patterns contradict the Order Claim.}
    \label{fig:forb-LNk-terrain}
  \end{figure}
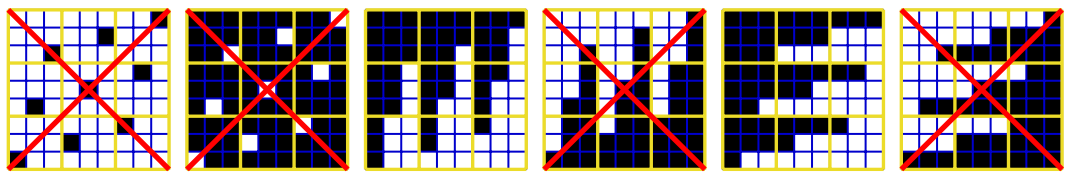

\begin{lemma}\label{lem:more-forb}
  Let $A, B \subseteq V(G)$ be such that $A \prec B$.
  $\adj{\prec}{G \langle A,B \rangle}$ cannot be equal to $\mat{2}{0}$ nor to $\mat{2}{1}$.
\end{lemma}
\begin{proof}
  These matrices have the pattern forbidden by the Order Claim; see~\cref{fig:forb-LNk-terrain}.
\end{proof}

\begin{theorem}\label{thm:ht-free-terrains}
  $H_t$-free visibility graphs of 1.5D terrains have bounded twin-width.
  Furthermore, $O(1)$-sequences can be computed in polynomial time on such graphs given with a~left-right ordering $\prec$.
\end{theorem}
\begin{proof}
  We run~\cref{thm:canonicaler} with $\adj{\prec}{G}$ (where $G$ is drawn from this class) and integer $k := \max(\lceil \sqrt t \rceil,2)$.
  We claim that the algorithm cannot return a universal pattern $\mat{k}{s}$ for some $s \in \{0,1,\uparrow,\downarrow,\leftarrow,\rightarrow\}$.
  Indeed by~\cref{lem:more-forb}, $s$ cannot be in $\{0,1\}$, and by $H_t$-freeness $s$~cannot be in $\{\uparrow,\downarrow,\leftarrow,\rightarrow\}$.
  The only option left is that a contraction sequence of $G$ is returned in time $f(t) n^{O(1)}$ witnessing that $\tww(G) \leqslant g(t)$, for some computable functions~$f, g$.  
\end{proof}

\begin{theorem}\label{thm:fpt-alg-terrain}
  \textsc{$k$-Ladder} and \textsc{$k$-Biclique} are FPT in visibility graphs of 1.5D terrains given with a~left-right ordering.
\end{theorem}
\begin{proof}
  \cref{thm:ht-free-terrains} shows that $\tww \sqsubseteq_{\text{eff}} \lambda$ (hence, $\tww \sqsubseteq_{\text{eff}} \beta$) on visibility graphs of 1.5D terrains given with a left-right order.
  We thus conclude by~\cref{thm:win-win}.
\end{proof}

One can observe that visibility graphs of 1.5D terrains have, in general, unbounded twin-width.
Indeed, it is a relatively easy exercise to build, for every integer $t$, a terrain on $4t^2+O(1)$ vertices that contains a semi-induced $T_k$, where $3k^2$ vertices are used for the tripartition $(A,B,C)$ of $T_k$, and $k^2$ vertices are used to strategically block the visibility of $B$ toward $C$. 

\subsection{Visibility graphs of simple polygons are not delineated}\label{subsec:polygons}

We exhibit a family of simple polygons whose visibility graphs are transduction equivalent to the class of all bipartite subcubic graphs $\mathcal B_{\leqslant 3}$, and conclude by~\cref{lem:not-delineated}.
This shows something a priori stronger than that visibility graphs of simple polygons are not delineated. 

\begin{theorem}
  There is a class $\mathcal P$ of visibility graphs of special simple polygons, such that
  \begin{compactitem}
  \item(i) there is a first-order transduction $\mathsf T_1$ satisfying $\mathsf T_1(\mathcal P) \supset \mathcal B_{\leqslant 3}$, and
  \item(ii) there is a first-order transduction $\mathsf T_2$ satisfying $\mathsf T_2(\mathcal \mathcal \mathcal B_{\leqslant 3}) \supset \mathcal P$.
  \end{compactitem}
  In particular, the class $\mathcal P$, and its hereditary closure, have unbounded twin-width --due to $(i)$-- and are monadically dependent, and even monadically stable --due to $(ii)$.
\end{theorem}

\begin{proof}
  The class $\mathcal P$ consists of each image of a bipartite subcubic graph $G=(A,B,E(G))$ by the following transformation $\Pi$.
  Let $A=\{a_1,\ldots,a_s\}$ and $B=\{b_1,\ldots,b_t\}$.
  For each vertex $a_i \in A$, we create the $2d(a_i)+1 \in \{1,3,5,7\}$ first vertices of the list: $d_i, p_i, d'_i, p'_i, d''_i, p''_i, d'''_i$.
  For each vertex $b_j \in B$, we create a vertex $q_j$.
  Let $D$ be the set of vertices of the form $d_i, d'_i, d''_i, d'''_i$, let $P$ be the set of vertices of the form $p_i, p'_i, p''_i$, and $Q = \{q_j~:~b_j \in B\}$.
  
  We make $p_i$ (resp.~$p'_i$, resp.~$p''_i$) adjacent to $q_j$ such that $b_j$ is the neighbor of $a_i$ with largest (resp. second largest, resp.~third largest) index (if they exist).
  We make $d_i$ adjacent to $p_i$, $d'_i$ adjacent to $p_i$ and $p'_i$, $d''_i$ adjacent to $p'_i$ and $p''_i$, and $d'''_i$ adjacent to $p''_i$ (if they exist).
  Finally we turn $D \cup Q$ into a clique.
  This finishes the construction of $\Pi(G)$.
  \cref{fig:class-p} shows how to represent, for every $G \in \mathcal B_{\leqslant 3}$, the graph $\Pi(G)$ as the visibility graph of a simple polygon.
  
\begin{figure}[h!]
  \centering
  \begin{tikzpicture}[point/.style={circle,fill,inner sep=0.03cm}]
    \def\s{1}
    \def\h{1.5}
    \foreach \i in {1,...,4}{
      \node[draw,circle,inner sep=0.03cm] (a\i) at (\i * \s, 0) {$a_\i$} ;
      \node[draw,circle,inner sep=0.03cm] (b\i) at (\i * \s, \h) {$b_\i$} ;
    }
    \foreach \i/\j in {a1/b1,a1/b3,a1/b4,a2/b2,a2/b3,a3/b1,a3/b4,a4/b2,a4/b4}{
      \draw (\i) -- (\j) ;
    }
    \def\z{-0.5}
    \def\t{1.5}
    \pgfmathsetmacro\y{- \z * \t / \h}   
    \def\e{0.1}
    \begin{scope}[xshift=6cm]
      \foreach \i in {1,...,4}{
        \node[point] (pb\i) at (\t * \i,\h) {} ;
        \node at (\t * \i,\h + 0.25) {$q_\i$} ;
      }
      \draw (\t - 1.5 * \e,0) -- (\t - 3 * \y,\z) -- (\t - 0.5 * \e,0) -- (\t - 2 * \y,\z) -- (\t + 0.5 * \e,0) -- (\t - 0 * \y + 3 * \e * \y,\z) -- (\t + 1.5 * \e,0) --
      (2 * \t - \e,0) -- (2 * \t - \y,\z) -- (2 * \t,0) -- (2 * \t + 1.5 * \e * \y,\z) -- (2 * \t + \e,0) --
      (3 * \t - \e,0) -- (3 * \t - \y,\z) -- (3 * \t,0) -- (3 * \t + 2 * \y,\z) -- (3 * \t + \e,0) --
      (4 * \t - \e,0) -- (4 * \t - \e /2,\z) -- (4 * \t,0) -- (4 * \t + 2 * \y,\z) -- (4 * \t + \e,0) -- (pb4) -- (pb3) -- (pb2) -- (pb1) -- (\t - 1.5 * \e,0) ;
    \end{scope}
  \end{tikzpicture}
  \caption{The transformation $\Pi$ and how its images are representable by simple polygons.}
  \label{fig:class-p}
\end{figure}
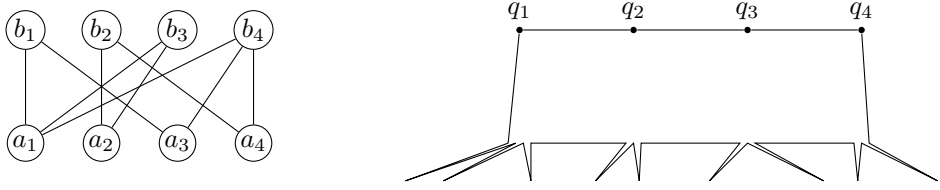

\medskip

$(i)$. We now prove the first item, which in particular implies that $\mathcal P$ has unbounded twin-width by~\cref{thm:transduction} and the second item of~\cref{thm:unbounded-tww}.
Let $G$ be any graph in $\mathcal B_{\leqslant 3}$ and $\Pi(G) \in \mathcal P$.
we want to design an FO transduction $\mathsf T_1$ that undoes $\Pi$.

The transduction $\mathsf T_1$ adds four unary relations $A, A', B, C$ to its input graph $\Pi(G)$.
In the run we are interested in, they form a partition of the vertex set, where $B$ colors the vertices $q_j$, $A$ colors the vertices $p_i$, $A'$ colors the vertices $p'_i$ and $p''_i$, and $C$ colors the remaining vertices (those in $D$); see~\cref{fig:unary-interp-p}.
The first-order unary formula that defines the vertex set is $\nu_1(x) = A(x) \lor B(x)$.
The new edge set is defined by:
$$\varphi_1(x,y) = A(x) \land B(y) \land \Bigg(E(x,y)$$
$$\lor \bigg(\exists z_1 \exists x_1~C(z_1) \land A'(x_1) \land E(x,z_1) \land E(z_1,x_1) \land \Big(E(x_1,y) $$
$$\lor \big(\exists z_2 \exists x_2~C(z_1) \land A'(x_2) \land E(x_1,z_2) \land E(z_2,x_2) \land E(x_2,y)\big)\Big)\bigg)\Bigg).$$
$\mathsf T_1$ indeed creates the bipartite graph $G \in \mathcal B_{\leqslant 3}$ back, since vertices in $C$ are adjacent to at most two vertices in $A \cup A'$ corresponding to the same vertex in $G$. 

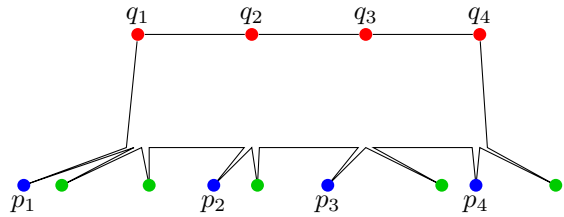
\begin{figure}[h!]
  \centering
  \begin{tikzpicture}[point/.style={circle,fill,inner sep=0.06cm}]
    \def\z{-0.5}
    \def\h{1.5}
    \def\t{1.5}
    \pgfmathsetmacro\y{- \z * \t / \h}   
    \def\e{0.1}
    \begin{scope}
      \foreach \i in {1,...,4}{
        \node[color=red,point] (pb\i) at (\t * \i,\h) {} ;
        \node at (\t * \i,\h + 0.25) {$q_\i$} ;
      }
      \draw (\t - 1.5 * \e,0) -- (\t - 3 * \y,\z) -- (\t - 0.5 * \e,0) -- (\t - 2 * \y,\z) -- (\t + 0.5 * \e,0) -- (\t - 0 * \y + 3 * \e * \y,\z) -- (\t + 1.5 * \e,0) --
      (2 * \t - \e,0) -- (2 * \t - \y,\z) -- (2 * \t,0) -- (2 * \t + 1.5 * \e * \y,\z) -- (2 * \t + \e,0) --
      (3 * \t - \e,0) -- (3 * \t - \y,\z) -- (3 * \t,0) -- (3 * \t + 2 * \y,\z) -- (3 * \t + \e,0) --
      (4 * \t - \e,0) -- (4 * \t - \e /2,\z) -- (4 * \t,0) -- (4 * \t + 2 * \y,\z) -- (4 * \t + \e,0) -- (pb4) -- (pb3) -- (pb2) -- (pb1) -- (\t - 1.5 * \e,0) ;

      \node[color=blue,point] at (\t - 3 * \y,\z) {} ;
      \node[color=blue,point] at (2 * \t - \y,\z) {} ;
      \node[color=blue,point] at (3 * \t - \y,\z) {} ;
      \node[color=blue,point] at (4 * \t - \e /2,\z) {} ;

      \node at (\t - 3 * \y,\z - 0.25) {$p_1$} ;
      \node at (2 * \t - \y,\z - 0.25) {$p_2$} ;
      \node at (3 * \t - \y,\z - 0.25) {$p_3$} ;
      \node at (4 * \t - \e /2,\z- 0.25) {$p_4$} ;

      \node[color=black!20!green,point] at (\t - 2 * \y,\z) {} ;
      \node[color=black!20!green,point] at (\t - 0 * \y + 3 * \e * \y,\z) {} ;
      \node[color=black!20!green,point] at (2 * \t + 1.5 * \e * \y,\z) {} ;
      \node[color=black!20!green,point] at (3 * \t + 2 * \y,\z) {} ;
      \node[color=black!20!green,point] at (4 * \t + 2 * \y,\z) {} ;
    \end{scope}
  \end{tikzpicture}
  \caption{The desired interpretation of unary relations $A$ in blue, $A'$ in green, $B$ in red, and $C$, the remaining 13 vertices.}
  \label{fig:unary-interp-p}
\end{figure}

\medskip

$(ii)$. We now want to build every graph $\Pi(G) \in \mathcal P$ by applying a transduction $\mathsf T_2$ to a bipartite subcubic graph.
We just need to observe that $\Pi(G)$ deprived of the clique on $D \cup Q$ is itself a bipartite subcubic graph, say $H$.
The transduction $\mathsf T_2$ inputs $H$ and adds one unary relation $U$, which in the correct run corresponds to $D \cup Q$.
The unary relation $\nu_2$ does not change the domain, while the binary relation $\varphi_2(x,y)$ adds all edges between pairs in $U$. 
\end{proof}

\subsection{Twin-width is $\alpha$-bounded in visibility graphs of simple polygons}

Here we bound the twin-width of visibility graphs of simply polygons by a function of their independence number.

\begin{theorem}\label{thm:alpha-polygon}
  Twin-width is $\alpha$-bounded in visibility graphs of simple polygons, and effectively $\alpha$-bounded if a geometric representation is given.
\end{theorem}
\begin{proof}
  Let $\mathcal P$ be a simple polygon, and $G$ its visibility graph.
  We identify a vertex of $G$ with its corresponding geometric vertex of $\mathcal P$.
  Let $\prec$ be the total order whose successor relation is a Hamiltonian path of the boundary of $\mathcal P$. 
  Visibility graphs of simple polygons satisfy the double-X property:\footnote{The name is chosen since the \emph{Order Claim} of 1.5D terrains is sometimes called the \emph{X property}.}
  If $b' \prec a \prec b \prec c \prec d \prec c'$, and $ac$, $bd$, $ac'$, $db'$ are all in $E(G)$, then $ad$ is also an edge of $G$ (see~\cref{fig:double-X}).

  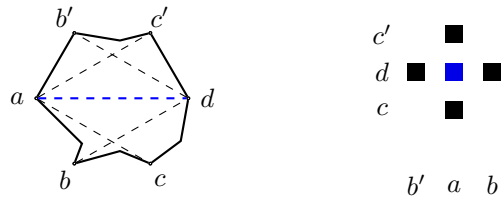
\begin{figure}[h!]
    \centering
    \begin{tikzpicture}
     \foreach \a/\l in {180/a,240/b,300/c,0/d,60/{c'},120/{b'}}{ 
        \node (\l) [draw,circle,inner sep=0.01cm] at (\a:1) {} ;
      }
     \draw[thick] (a) -- (-0.4,-0.6) -- (b) -- (0.1,-0.7) -- (c) --++(0.4,0.3) -- (d) -- (c') --++(-0.4,-0.1) -- (b') -- (a) ;
      \foreach \a/\l in {180/a,240/b,300/c,0/d,60/{c'},120/{b'}}{ 
        \node (\l) [draw,circle,inner sep=0.01cm] at (\a:1) {} ;
        \node at (\a:1.25) {$\l$} ;
      }
      \foreach \i/\j in {a/c,b/d,a/{c'},d/{b'}}{
        \draw[thin,dashed] (\i) -- (\j) ;
      }
      \draw[blue,thick,dashed] (a) -- (d) ;

      \begin{scope}[xshift=4cm, yshift=-0.65cm]
        \foreach \i/\j\l in {0/-0.5/{b'}, 0.5/-0.52/a, 1/-0.5/b, -0.45/0.5/c, -0.45/1/d, -0.45/1.5/{c'}}{
          \node at (\i,\j) {$\l$} ;
        }
        \foreach \i/\j in {0/1,0.5/0.5,0.5/1.5,1/1}{
          \node[fill] at (\i,\j) {} ;
        }
        \node[fill,black!10!blue] at (0.5,1) {} ;
      \end{scope}
    \end{tikzpicture}
    \caption{Left: The double-X property. Right: What it implies in the adjacency matrix ordered along the boundary of the polygon; the four 1 entries in black force the central one in blue.}
    \label{fig:double-X}
  \end{figure}

  Indeed, the boundary order imposes that the line segments $ac$ and $bd$ cross, and that the line segments $ac'$ and $b'd$ cross.
  Let us call $p$ and $q$ these two intersection points.
  It can be observed that $p$ and $q$ have to be on different sides of the line extending $ad$, and that the interior of $apdq$ cannot be intersected by the boundary of $\mathcal P$.
  Hence $ad$ has to be an edge.
  
  This excludes that the complement of a(n arbitrary) permutation is realized by the adjacency matrix of $G$ ordered along $\prec$ (see~\cref{fig:forb-LNk-polygon}).
  \begin{figure}[h!]
  \centering
    \begin{tikzpicture}[scale=.235]

   \foreach \symb/\b/\xsh/\ysh in {{==}/0/-30/-10, {>}/0/-10/-10,{<}/0/0/-10,{<}/1/10/-10,{>}/1/20/-10}{
   \begin{scope}[xshift=\xsh cm,yshift=\ysh cm]     
    \foreach \i in {0,...,8}{
      \foreach \j in {0,...,8}{
        \pgfmathsetmacro{\ip}{\i+1}
        \pgfmathsetmacro{\jp}{\j+1}
        \pgfmathsetmacro{\col}{ifthenelse(\i == mod(\j,3)*3+floor(\j/3),"black",ifthenelse(\b==1,ifthenelse(\i \symb mod(\j,3)*3+floor(\j/3),"black","white"),ifthenelse(\j \symb mod(\i,3)*3+floor(\i/3),"black","white")))}
        \fill[\col] (\i,\j) -- (\i,\jp) -- (\ip,\jp) -- (\ip,\j) -- cycle;
      }
    }
    \draw[line width=0.75pt, scale=1, color=black!20!blue] (0, 0) grid (9, 9);
    \draw[line width=1.25pt, scale=3, color=black!10!yellow] (0, 0) grid (3, 3);
   \end{scope}
   }
   \foreach \symb/\b/\xsh/\ysh in {{!=}/0/-20/-10}{
   \begin{scope}[xshift=\xsh cm,yshift=\ysh cm]     
    \foreach \i in {0,...,8}{
      \foreach \j in {0,...,8}{
        \pgfmathsetmacro{\ip}{\i+1}
        \pgfmathsetmacro{\jp}{\j+1}
        \pgfmathsetmacro{\col}{ifthenelse(\i == mod(\j,3)*3+floor(\j/3),"white","black")}
        \fill[\col] (\i,\j) -- (\i,\jp) -- (\ip,\jp) -- (\ip,\j) -- cycle;
      }
    }
    \draw[line width=0.75pt, scale=1, color=black!20!blue] (0, 0) grid (9, 9);
    \draw[line width=1.25pt, scale=3, color=black!10!yellow] (0, 0) grid (3, 3);

    \draw[red,line width=2pt] (0,0) -- (9,9) ;
    \draw[red,line width=2pt] (0,9) -- (9,0) ;
   \end{scope}
   }
    \end{tikzpicture}
    \caption{The universal pattern forbidden by the double-X property.}
    \label{fig:forb-LNk-polygon}
\end{figure}
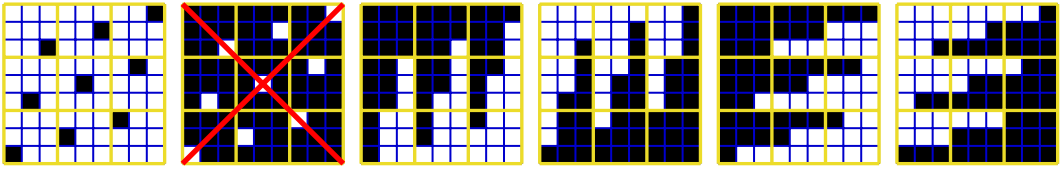

  We now upperbound the size of a universal pattern in $\adj{\prec}{G}$ (among the other five patterns) in terms of $\alpha(G)$, and conclude by~\cref{thm:canonicaler}.
  We will actually not need the universal pattern in its whole, but simply a decreasing subsequence of it.
  (This is made formal in the next paragraph, where we extract a large anti-diagonal induced matching or half-graph.)
  This is convenient since we can thus apply Ramsey's theorem while keeping the ``complexity'' of the initial structure. 

Let $p = \text{Ram}(\text{Ram}(4,\alpha(G)),\alpha(G))$, where $\text{Ram}$ is the function of~Ramsey's theorem (see~\cref{thm:Ramsey}).
Note that if the twin-width of $G$ is larger than a certain function of $p$, we can find in each of the five allowed universal patterns $2p$ vertices of $G$: $a_1 \prec a_2 \prec \ldots \prec a_{p-1} \prec a_p \prec b_p \prec b_{p-1} \prec \ldots \prec b_2 \prec b_1$ such that $a_ib_j \in E(G)$ if and only if $i = j$ (resp. $i \leqslant j$, resp.~$i \geqslant j$).
We denote $\{a_1, \ldots, a_p\}$ (resp.~$\{b_1, \ldots, b_p\}$) by $A$ (resp.~$B$).
We now work toward finding a contradiction.

Let $A' \subseteq A$ induce a clique in $G$ with $|A'|=\text{Ram}(\alpha(G),4)$.
Let $B'$ be the vertices of~$B$ with the same index as a vertex of $A'$, and let $B'' \subseteq B'$ induce a clique in $G$ of size~4.
Finally let $A''$ be the vertices in~$A'$ (or~$A$ for that matter) with same index as a vertex in~$B''$.
We relabel the eight vertices of $A'' \cup B''$ by $\alpha_1 \prec \alpha_2 \prec \alpha_3 \prec \alpha_4 \prec \beta_4 \prec \beta_3 \prec \beta_2 \prec \beta_1$.

First observe that, since they form a clique, $\alpha_1, \alpha_2, \alpha_3, \alpha_4$ are in convex position.
For $\alpha_2\beta_2$ and $\alpha_3\beta_3$ to be in $E(G)$, the vertices $\beta_2$ and $\beta_3$ have to be in the convex (possibly infinite) region delimited by the line segment $\alpha_2 \alpha_3$, the ray starting at $\alpha_2$ and passing through $\alpha_1$, and the ray starting at $\alpha_3$ and passing through $\alpha_4$.
Since $\beta_2$ comes after $\beta_3$ in the boundary order, the quadrangle $\alpha_2 \alpha_3 \beta_3 \beta_2$ has to be non self-intersecting (otherwise $\alpha_2\beta_2$ and $\alpha_3\beta_3$ cannot both be edges, see left of~\cref{fig:visib-contradiction}).
We now claim that $\alpha_2 \alpha_3 \beta_3 \beta_2$ is a convex quadrangle.
Assume for the sake of contradiction that $\beta_2$ is in the interior of the triangle $\alpha_2 \alpha_3 \beta_3$ (this is without loss of generality).
As $\alpha_2\beta_2$ is an edge of $G$, the line segment $\alpha_2\beta_2$ cuts $\mathcal P$ into two simple polygons: $\mathcal P^-$ containing $\alpha_1$, and $\mathcal P^+$ containing $\alpha_3$.
Observe that no line segment starting at $\beta_3$ and fully contained in $\mathcal P$ can intersect $\mathcal P^- \setminus \{\beta_2\}$.
Indeed, since $\beta_2$ is in the interior of $\alpha_2 \alpha_3 \beta_3$, the ray starting at $\beta_3$ and passing through $\beta_2$ remains entirely within $\mathcal P^+$.
However $\alpha_1$ is in $\mathcal P^-$.
Therefore $\beta_3$ and $\beta_1$ cannot see each other; a contradiction (see middle of~\cref{fig:visib-contradiction}).
\begin{figure}[h!]
  \centering
  \begin{tikzpicture}[scale=.85,point/.style={fill,circle,inner sep=0.03cm}]
    \foreach \i/\j/\l/\t/\oh/\ov in {0/3/{a1}/{\alpha_1}/{-0.25}/0, 0.5/2/{a2}/{\alpha_2}/{-0.25}/0,  1.3/1.3/{a3}/{\alpha_3}/0/{-0.25}, 2.1/1/{a4}/{\alpha_4}/0/{-0.25},
    2.5/1.4/{b2}/{\beta_2}/0.25/0, 4/2.2/{b3}/{\beta_3}/{0.25}/0}{
      \node[point] (\l) at (\i,\j) {} ;
      \node at (\i + \oh,\j + \ov) {$\t$} ;
    }
    \draw[thin] (a1) --++(-0.2,-0.5) node (x) {} -- (a2) --++(0.1,-1.2) -- (a3) -- (a4) ;
    \draw[thin] (b3) --++(-0.4,0.5) --++(-0.4,0.1) -- (b2) ;
    \draw[dotted,thin] (a4) --++ (2,0.3) -- (b3) ;
    \draw[dotted,thin] (b2) --++ (-0.3,1.3) node (y) {} --++(-0.3,0.1) node (z) {} -- (a1) ;
    \draw[dashed] (a2) -- (b2) ;
    \draw[dashed,red] (a3) -- (b3) ;

    \begin{scope}[xshift=5.5cm]
    \foreach \i/\j/\l/\t/\oh/\ov in {0/3/{a1}/{\alpha_1}/{-0.25}/0, 0.5/2/{a2}/{\alpha_2}/{-0.25}/0,  1.3/1.3/{a3}/{\alpha_3}/0/{-0.25}, 2.1/1/{a4}/{\alpha_4}/0/{-0.25},
    2.5/1.85/{b2}/{\beta_2}/0/0.25, 4/2.2/{b3}/{\beta_3}/{0.25}/0}{
      \node[point] (\l) at (\i,\j) {} ;
      \node at (\i + \oh,\j + \ov) {$\t$} ;
    }
    \draw[thin] (a1) --++(-0.2,-0.5) node (x) {} -- (a2) --++(0.1,-1.2) -- (a3) -- (a4) ;
    \draw[thin] (b3) --++(-0.4,0.5) --++(-0.4,0.1) -- (b2) ;
    \draw[dotted,thin] (a4) --++ (2,0.3) -- (b3) ;
    \draw[dotted,thin] (b2) --++ (-0.3,1.3) node (y) {} --++(-0.3,0.1) node (z) {} -- (a1) ;
    \draw[dashed] (a2) -- (b2) ;
    \draw[dashed] (a3) -- (b3) ;

    \fill[red,opacity=0.2] (0,3) -- (-0.2,2.5) -- (0.5,2) -- (2.5,1.85) -- (2.2,3.15) -- (1.9,3.25) -- cycle ;
    \node at (1.3,2.65) {$\mathcal P^-$} ;
    \fill[blue,opacity=0.2] (0.5,2) -- (0.6,0.8) -- (1.3,1.3) -- (2.1,1) -- (4.1,1.3) -- (4,2.2) -- (3.6,2.7) -- (3.2,2.8) -- (2.5,1.85) -- cycle ;
    \node at (2.65,1.4) {$\mathcal P^+$} ;
    \end{scope}

    \begin{scope}[xshift=11cm]
    \foreach \i/\j/\l/\t/\oh/\ov in {0/3/{a1}/{\alpha_1}/{-0.25}/0, 0.5/2/{a2}/{\alpha_2}/{-0.25}/0,  1.3/1.3/{a3}/{\alpha_3}/0/{-0.25}, 2.1/1/{a4}/{\alpha_4}/0/{-0.25},
    2.5/2.8/{b2}/{\beta_2}/0/0.25, 4/2.2/{b3}/{\beta_3}/{0.25}/0}{
      \node[point] (\l) at (\i,\j) {} ;
      \node at (\i + \oh,\j + \ov) {$\t$} ;
    }
    \draw[thin] (a1) --++(-0.2,-0.5) node (x) {} -- (a2) --++(0.1,-1.2) -- (a3) -- (a4) ;
    \draw[thin] (b3) --++(-0.4,0.5) --++(-0.4,0.1) -- (b2) ;
    \draw[dotted,thin] (a4) --++ (2,0.3) -- (b3) ;
    \draw[dotted,thin] (b2) --++ (-1.5,-0.2) node (y) {} -- (a1) ;
    \draw[dashed] (a2) -- (b2) ;
    \draw[dashed] (a3) -- (b3) ;
    \draw[dashed] (a2) -- (b3) ;
    \draw[dashed] (a3) -- (b2) ;
    \draw[dashed] (a2) -- (a3) ;
    \draw[dashed] (b2) -- (b3) ;
    \fill[blue,opacity=0.2] (0.5,2) -- (1.3,1.3) -- (4,2.2) -- (2.5,2.8) -- cycle ;
    \end{scope}
  \end{tikzpicture}
  \caption{Left: If $\alpha_2 \alpha_3 \beta_3 \beta_2$ is self-intersecting, at least one edge of $\alpha_2\beta_2, \alpha_3\beta_3$ (here $\alpha_3 \beta_3$) is missing from $G$. Center: If $\beta_2$ lies in the interior of $\alpha_2 \alpha_3 \beta_3$, vertices $\beta_1$ (in $\mathcal P^+$) and $\beta_3$ cannot see each other without blocking the edge $\alpha_2\beta_2$. Right: If $\alpha_2 \alpha_3 \beta_3 \beta_2$ is in convex position (in this order), then both $\alpha_2\beta_3$ and $\alpha_2\beta_3$ are edges; another contradiction.}
  \label{fig:visib-contradiction}
\end{figure}
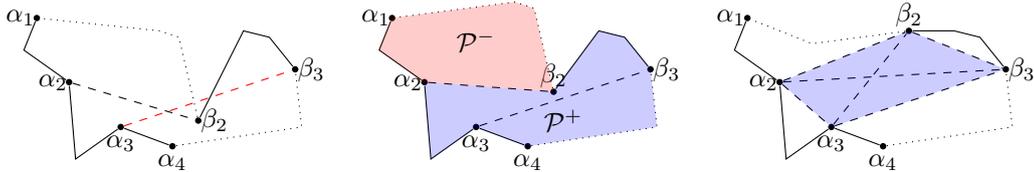

Since the four sides of the convex, non self-intersecting quadrangle $\alpha_2 \alpha_3 \beta_3 \beta_2$ are edges of~$G$, the two diagonals $\alpha_2\beta_3$ and $\alpha_3\beta_2$ are also edges (since $\mathcal P$ cannot intersect the interior of $\alpha_2 \alpha_3 \beta_3 \beta_2$, see right of~\cref{fig:visib-contradiction}); a contradiction to the induced matching or half-graph in between $A$ and $B$. 
\end{proof}

\begin{theorem}\label{thm:fpt-alg-polygon}
  \textsc{$k$-Independent Set} is FPT in visibility graphs of simple polygons given with a geometric representation.
\end{theorem}
\begin{proof}
  This is a consequence of~\cref{thm:alpha-polygon,thm:win-win}, and holds more generally if the input graph only comes with a Hamiltonian path corresponding to the boundary in one of its representations.
\end{proof}

\section{Open questions}

We conclude with some open problems.
We wonder whether the classes of unit segment graphs, visibility graphs of 1.5D terrains, and all tournaments are delineated.

Our way of refuting delineation was to show transduction equivalence to subcubic graphs.
We ask whether this method is complete for hereditary classes.
\begin{conjecture}
  A hereditary class is not delineated if and only if it is transduction equivalent to all subcubic graphs.
\end{conjecture}
This would in particular imply that \emph{monadically independent} can be equivalently replaced by \emph{monadically stable} in the definition of delineation.



\end{document}